%% file: main.tex
\documentclass[10pt]{article}

\usepackage[margin=1in]{geometry}
\usepackage[hyphens]{url}

\usepackage{amsmath,amssymb,amsthm}
\usepackage{graphicx}
\usepackage{booktabs}
\usepackage{longtable}
\usepackage{array}
\usepackage{multirow}
\usepackage[table]{xcolor}
\definecolor{indigo}{RGB}{63,81,181}
\definecolor{navyblue}{RGB}{0,0,128}
\definecolor{steelblue}{RGB}{70,130,180}
\definecolor{crimson}{RGB}{220,20,60}
\definecolor{gold}{RGB}{184,134,11}
\definecolor{violet}{RGB}{148,0,211}
\definecolor{teal}{RGB}{0,128,128}

\usepackage[hidelinks]{hyperref}
\usepackage[authoryear,round]{natbib}
\usepackage{setspace}
\usepackage{caption}
\usepackage{subcaption}
\usepackage{float}
\usepackage{enumitem}
\usepackage{pgfplots}
\pgfplotsset{compat=1.18}
\usepackage{tikz}
\usetikzlibrary{arrows.meta,positioning,shapes.geometric,calc,decorations.pathreplacing,backgrounds}
\usepackage{tcolorbox}
\usepackage{tabularx}
\usepackage{adjustbox}

\usepackage{appendix}
\usepackage{mdframed}

\definecolor{navyblue}{RGB}{27,42,74}
\definecolor{steelblue}{RGB}{46,95,163}
\definecolor{lightblue}{RGB}{214,228,247}
\definecolor{crimson}{RGB}{200,57,43}
\definecolor{gold}{RGB}{180,130,20}
\definecolor{darkgray}{RGB}{55,65,81}
\definecolor{forestgreen}{RGB}{34,120,57}

\newtheorem{definition}{Definition}[section]
\newtheorem{proposition}{Proposition}[section]
\newtheorem{remark}{Remark}[section]
\newtheorem{assumption}{Assumption}[section]

\tcbuselibrary{skins,breakable}
\newtcolorbox{keyresult}[1][Key Result]{
 colback=lightblue!60, colframe=navyblue,
 fonttitle=\bfseries\color{white}, coltitle=white,
 colbacktitle=navyblue, title=#1, breakable}

\usepackage{microtype}
\usepackage{titlesec}

\titlespacing*{\section}{0pt}{10pt plus 2pt minus 2pt}{4pt plus 1pt minus 1pt}
\titlespacing*{\subsection}{0pt}{8pt plus 2pt minus 2pt}{3pt plus 1pt minus 1pt}
\titlespacing*{\subsubsection}{0pt}{6pt plus 2pt minus 1pt}{2pt plus 1pt minus 1pt}

\setlist{nosep,leftmargin=*}

\tcbset{before skip=6pt, after skip=6pt, boxsep=4pt, left=4pt, right=4pt}

\bibpunct{(}{)}{;}{a}{,}{,}
\begin{document}

\title{%
 \vspace{-0.5cm}
 {\LARGE\bfseries The AI Transformation Gap Index (AITG)}\\[0.5em]
 {\large An Empirical Framework for Measuring AI Transformation\\
 Opportunity, Disruption Risk, and Value Creation at the Industry and Firm Level}%
}

\author{%
 \textbf{Dean Barr}\\[0.3em]
 \normalsize Applied AI Researcher\\[0.2em]
 \normalsize \texttt{dean@dsconsult.ai}\\[0.4em]
 \small \textit{The AI Transformation Gap Index (AITG), Industry AI Susceptibility Score (IASS),}\\
 \small \textit{Value Creation Bridge (VCB), Implementation Feasibility Score (IFS),}\\
 \small \textit{AI Disruption Risk Index (ADRI), and AITG Value Density are original constructs}\\
 \small \textit{introduced in this paper.}
}

\date{February 2026}

\maketitle
\thispagestyle{empty}

\begin{abstract}
\noindent
Despite the scale of capital being deployed toward AI initiatives, no empirical
framework currently exists for benchmarking where a firm stands relative to
competitors in AI readiness and deployment, or for translating that position into
auditable financial outcomes. In practice, private equity deal teams, management
consultants, and corporate strategists have relied on qualitative judgment and
ad hoc maturity labels; tools that are neither comparable across industries nor
grounded in observable economic data.

This paper introduces the \textbf{AI Transformation Gap Index} (AITG), a composite
empirical framework that measures the distance between a firm's current AI deployment
and a time varying, industry constrained capability frontier, then maps that distance
to dollar denominated value creation, execution feasibility under uncertainty, and
competitive disruption risk.
Five linked modules address this gap: cross industry normalization (IASS),
a dynamic capability ceiling that evolves with frontier capabilities (AFC),
trajectory based firm scoring with integrated execution risk (IFS), a CES
bottleneck value decomposition mapping gap scores to enterprise value (VCB),
and a competitive hazard measure for inaction (ADRI). I calibrate the framework
for 22 industry verticals and apply it to 14 public companies using public filings.
A retrospective construct validity exercise correlating AITG scores with observed
EBITDA margin expansion yields Spearman $\rho_s = 0.818$ ($n = 10$), directionally
consistent with predictions though insufficient for causal identification. A
counterintuitive result emerges: the largest AI transformation gaps do not produce
the highest value density, because implementation friction, CES bottlenecks, and
timing lags erode the theoretical upside of wide gaps.


\end{abstract}

\medskip
\noindent\textit{A consolidated notation reference and parameter calibration register appear in Appendix~A (Tables~\ref{tab:notation} and~\ref{tab:param_calibration}).}
\medskip

\section{Introduction}
\label{sec:intro}

\subsection{Motivation}
\label{sec:motivation}
\label{sec:measurement_problem}

Artificial intelligence is widely recognized as a general purpose technology
with large but uneven firm level effects
\cite{David1990,BrynjolfssonRockSyverson2021}. Yet the prevailing tools used
to assess ``AI maturity'' remain difficult to interpret economically, and in
practice consultant practitioners, private equity deal teams, and corporate
strategists lack a rigorous, comparable method for estimating how much value
a company or industry can unlock through AI transformation. Ordinal maturity
labels are rarely comparable across industries and are not mapped to auditable
changes in productivity, margins, or enterprise value.

Two sources of measurement error are especially consequential.
\textit{Cross sectionally}, a maturity score assigned to a bank is not
commensurate with the same score assigned to a construction firm because the
feasible automation frontier differs structurally by industry.
\textit{Temporally}, the set of feasible AI capabilities changes rapidly, so a
static ceiling assumption can become stale within an investment horizon. Together
these limitations create predictable allocation errors. Investors capitalize
narratives weakly grounded in measurable operational change while underweighting
firms incurring near term complementary costs whose benefits arrive with
delay, the productivity J curve
\cite{BrynjolfssonRockSyverson2021}. \citet{BrynjolfssonRockSyverson2021}
demonstrate that GPTs require substantial complementary intangible investments
before measured productivity rises; \citet{BrynjolfssonHitt1996,BrynjolfssonHitt2000}
documented the identical pattern for information technology. Meanwhile,
\citet{Acemoglu2024} offers a carefully bounded estimate of no more than
0.66\,\% TFP growth over ten years, a constraint I engage with directly in
Section~\ref{sec:limitations}.

I introduce the AI Transformation Gap Index (AITG) to address this measurement
gap. Its five linked modules (Sections~\ref{sec:iass}--\ref{sec:ifs}) jointly
provide cross industry normalization, a time evolving capability ceiling, a
dollar denominated value creation bridge, and a quantified competitive hazard
score. All parameters are formally specified and anchored to publicly
verifiable data sources.

\subsection{Contributions}
\label{sec:contributions}
\label{sec:data_provenance}

I make five conceptual contributions and provide an illustrative empirical
calibration. The framework combines three input classes: (i)~publicly observed
data (BLS O*NET/OEWS, EDGAR, Lightcast, Census BTOS, Stanford HAI AI Index),
(ii)~rubric based author scored constructs whose interrater reliability is not
yet established (Section~\ref{sec:agenda}), and (iii)~assumption driven
parameters with sensitivity ranges reported in ESM Tables~S3--S5. The
14 company application constitutes an illustrative calibration, not a
statistically powered validation study.

\begin{enumerate}[leftmargin=1.5em]

\item \textbf{Industry normalized ceiling (IASS).}
I construct the IASS from O*NET task structure, OEWS occupational weights, regulatory friction, and competitive diffusion dynamics \cite{ONetResource,BLSOEWS2023}. Unlike additive checklist frameworks \citep{McKinseyStateAI2025,Gartner2023,Westerman2014}, the IASS uses geometric aggregation \cite{Mazziotta2013,OECD2008} with a Regulatory Friction Factor hard floor $\psi(\mathrm{RFF})$, enforcing noncompensability. Calibrations for 22 industry verticals are precomputed from BLS, Lightcast, and SEC filings.

\item \textbf{Time varying frontier (AFC).}
The AI Frontier Coefficient is a formally specified, time indexed multiplier that recalibrates the IASS ceiling as model capabilities advance (MMLU from 43\% to 93\% between 2020--2024 \cite{StanfordHAI2025}; METR task horizon doubling every seven months \cite{METRAnalysis2024}). I anchor $\theta_i$ exclusively to O*NET Automatable Task Density, eliminating the circular specification that arises when frontier update rules conflate capability expansion with adoption breadth.

\item \textbf{Trajectory based firm scoring with endogenous execution risk.}
Firm state is decomposed into six observable dimensions and mapped to a cascading three wave logistic trajectory \cite{BrynjolfssonRockSyverson2021,Rogers2003}. Organizational capacity and data readiness are endogenized into wave steepness and inflection timing \cite{BrynjolfssonHitt2000,GibbonsHenderson2012}, recovering the correct dynamic compounding structure. The formal inverse mapping uses a mandatory piecewise specification guarding against NaN crashes. Data nonrivalry \citep{JonesTonetti2020} is operationalized as a Firm Scale Factor $\Phi_f$.

\item \textbf{Value decomposition with enabling constraints (VCB).}
The VCB decomposes the gap into seven value pools, gates each on enabling infrastructure via a low elasticity CES bottleneck aggregator ($\rho = 5$, $\sigma = 1/6$), and applies a nonlinear capture function. I identify and correct four compounding errors in standard AI to EV conversions with correction magnitudes ranging from $2\times$ to $25\times$.

\item \textbf{Competitive hazard from delay (ADRI).}
The ADRI measures the competitive hazard from \emph{not} acting, operationalized as an instantaneous hazard intensity. A firm with a wide gap and high ADRI faces an increasing hazard rate as diffusion compresses margins and organizational debt compounds.

\end{enumerate}

\textbf{Illustrative calibration.}
I apply the framework to 14 public companies across eight industries using a
tiered evidence protocol based on public filings
(Section~\ref{sec:empirical}). These applications demonstrate internal
consistency and measurement mechanics, not causal claims. Two companies,
JPMorgan Chase and Zions Bancorporation, serve as primary depth cases
illustrating within industry contrasts under identical ceiling conditions. A
retrospective construct validity exercise using 2021 AITG scores and
2021--2023 realized EBITDA margin changes yields a Spearman $\rho_s = 0.818$
($n = 10$ nonfinancial firms), directionally consistent with predictions but
insufficient for causal identification given sample size and omitted variable
concerns.
\label{sec:ic_scripts}
ESM Part~IV supplies structured application guidance and a companion Excel
workbook implementing all modules for self assessment. Full reproducibility
code, including Monte Carlo sensitivity analysis, stress tests, a
25 question survey scoring pipeline, and the Excel workbook builder, is
available at \url{https://github.com/deanbrr/aitg-framework}.

\subsection{Related Work}
\label{sec:related_deployment}

The AITG integrates three literatures not previously connected in a single
operational framework. The \emph{task based automation literature}
\cite{AutorLevyMurnane2003,AutorDorn2013,AcemogluRestrepo2018,AcemogluRestrepo2019a}
provides the microeconomic foundation for structural differences in AI
susceptibility; I operationalize it as the Cognitive Task Density dimension of
the IASS. The \emph{composite indicator methodology}
\cite{OECD2008,NardoSaisana2005} provides the normalization and
sensitivity analysis framework required for institutional credibility. The
\emph{IT productivity and organizational complementarity literature}
\cite{BrynjolfssonHitt1996,BrynjolfssonHitt2000,BrynjolfssonMcElheran2016,
BrynjolfssonRockSyverson2021,GibbonsHenderson2012} provides the theoretical
basis for why AI investment to returns relationships are delayed, nonlinear,
and conditioned on organizational readiness.

Unlike \citet{Bass1969} single event diffusion models, the AITG three wave
cascading logistic decomposes adoption by sequential technology wave rather
than by innovation versus imitation. An alternative \citet{Cox1972}
proportional hazards formulation would complement the ADRI hazard intensity
(Eq.~\ref{eq:adri_hazard}) by providing a semiparametric baseline hazard; I
chose cascading logistics for wave specific modularity, since each wave has
independent $k_w$ and $t_{0,w}$ parameters updatable as deployment evidence
accumulates.

Unlike deployment aware optimization models such as ML Compass
\cite{MLCompass2024}, which use KKT based constrained optimization at the
system level, the AITG operates at the firm and industry level, addressing
cross industry normalization, time evolving ceilings, and competitive
disruption risk. The KKT intuition of binding constraints is operationalized
through the CES bottleneck aggregator: value capture is gated by the
minimum weighted harmonic of enabling dimensions, not their average.

Data marketplace pricing models \citep{JonesTonetti2020} provide a theoretical
foundation for the VCB Data Monetization pool: the nonrivalry of data implies
marginal value increasing with reuse, captured through the concave function
$\eta(g)$. Empirical evaluation of frontier AI agents on autonomous task
benchmarks \cite{CapabilityEval2024} documents conservative production level
autonomy, motivating the Wave~3 steepness discount ($k_{w=3}$ set 30\,\% below
Wave~2) and the IFS $\delta_{\mathrm{OCC}}$ penalty for firms lacking agentic
readiness.

Section~\ref{sec:theory} establishes theoretical foundations;
Sections~\ref{sec:iass}--\ref{sec:ifs} develop the five modules;
Section~\ref{sec:empirical} provides empirical illustrations;
Section~\ref{sec:sensitivity} reports sensitivity analysis;
Section~\ref{sec:limitations} addresses limitations; and
Section~\ref{sec:conclusion} concludes.

\section{Theoretical Foundations}
\label{sec:theory}

\subsection{The AI Transformed Frontier as an Industry Constrained Optimum}

The central construct is the frontier: the best operating configuration a firm in a given industry could achieve with current AI capabilities.

\begin{definition}[AI Transformed Frontier]
\label{def:frontier}
For firm $f$ in industry $i$ at time $t$, define the
\emph{AI Transformed frontier state} $S^*_{i,t}$ as the operating configuration
that maximizes risk adjusted enterprise value subject to:
\begin{enumerate}[label=(\roman*),noitemsep]
 \item process and physical constraints inherent to industry $i$;
 \item binding regulatory requirements applicable at time $t$;
 \item the data generation processes naturally available in industry $i$; and
 \item the AI capability set $\mathcal{A}_t$ available at time $t$.
\end{enumerate}
\end{definition}

Three deliberate choices are embedded here. First, the frontier is \emph{dynamic}: as $\mathcal{A}_t$ expands, $S^*_{i,t}$ shifts, opening new gap even for firms that have not changed. Second, it is \emph{industry constrained}: a construction firm at its frontier looks categorically different from a financial services firm at its frontier. Third, the objective is \emph{risk adjusted enterprise value}, anchoring the framework to economics rather than engineering and avoiding the methodological mistake of defining perfection as ``uses AI everywhere.''

The distance between a firm's current state and this frontier captures the investment opportunity.

\begin{definition}[AI Transformation Gap]
\label{def:gap}
The \emph{AI Transformation Gap} for firm $f$ at time $t$ is:
\begin{equation}
 \Delta_{f,t} \;\equiv\; d\!\left(S_{f,t},\; S^*_{i,t}\right)
 \label{eq:gap}
\end{equation}
where $d(\,\cdot\,\cdot\,)$ is a distance in a standardized capability space
and $S_{f,t}$ is the firm's current implementation state.
\end{definition}

$\Delta_{f,t}$ is a value creation signal, not a technology adoption signal,
because $S^*_{i,t}$ is anchored to economically material levers (cost, margin,
working capital, cycle time). A large gap can be \emph{bullish}, indicating
a large addressable opportunity, or \emph{bearish}, indicating structural
inability to close it. I separate gap measurement from closure probability; frameworks that blend the two sacrifice interpretability.

\subsection{AI Elasticity by Industry}

Industries differ dramatically in how much enterprise value responds to AI investment. I formalize this heterogeneity as follows.

\begin{definition}[AI Elasticity by Industry]
\label{def:elasticity}
\emph{AI Elasticity by Industry} $E_i$ is the marginal enterprise value
response to AI transformation investment, holding scale constant:
\begin{equation}
 E_i \;\equiv\; \frac{\partial\!\left(\mathrm{EV}/\mathrm{Revenue}\right)_i}
 {\partial\!\left(\text{AI Capability Index}\right)_i}
 \label{eq:elasticity}
\end{equation}
\end{definition}

$E_i$ is continuous, not binary. The task based literature
\cite{AutorDorn2013,AcemogluRestrepo2018} establishes that automation
suitability varies substantially within jobs and across tasks; industry level
susceptibility is a distribution over task types, not a categorical label. The IASS operationalizes $E_i$ through six weighted dimensions.

\subsection{Connection to the Acemoglu--Restrepo Task Framework}

The AITG's conceptual foundation maps directly onto the canonical task based model of \citet{AcemogluRestrepo2018}. Production uses tasks $z \in [N-1, N]$; tasks below threshold $I$ are performed by capital (automated), tasks above $I$ by labor. The production function is:

\begin{equation}
 Y = \Pi(I,N)\!\left[\Gamma(I,N)^{1/\sigma}(A_L L)^{(\sigma-1)/\sigma}
 + (1-\Gamma)^{1/\sigma}(A_K K)^{(\sigma-1)/\sigma}\right]^{\sigma/(\sigma-1)}
 \label{eq:ar_production}
\end{equation}

where $\Gamma(I,N)$ is the \emph{labor task content of production}. Automation
(increasing $I$) always reduces $\Gamma$ and labor share; new task creation
(increasing $N$) always raises both. The displacement reinstatement decomposition
\cite{AcemogluRestrepo2019a} is:

\begin{equation}
 d\ln\Gamma = \underbrace{-\text{displacement effect}}_{\text{automation\;($\uparrow I$)}}
 + \underbrace{\text{reinstatement effect}}_{\text{new tasks\;($\uparrow N$)}}
 \label{eq:displacement}
\end{equation}

The AITG's Cognitive Task Density dimension (Eq.~\ref{eq:ctd}) operationalizes $\Gamma(I,N)$ empirically at the industry level. A firm's AITG score corresponds to $I_f / I^*_i$, the ratio of its current automation frontier to the industry's achievable frontier. Moving this ratio toward 1.0 is the investment thesis.

\subsection{The Productivity J Curve and Measurement Implications}

AITG scores do not immediately predict observable productivity gains. \citet{BrynjolfssonRockSyverson2021} demonstrate that general purpose technologies require complementary intangible investments: process redesign, workforce retraining, and organizational restructuring, all expensed rather than capitalized. Measured productivity declines initially; benefits emerge once the complementary capital is in place. Adjusting for unmeasured intangibles, the authors find true TFP 15.9\,\% higher than official measures by 2017. The lag can persist 20--30 years for major GPTs, consistent with \citet{David1990} on the electrification delay.

I calibrate the AITG Value Creation Bridge directly to this J curve logic: costs load in the first 12--18 months and value emerges in a logistic pattern over 24--42 months (Section~\ref{sec:vcb_ramp}). This also provides the definitive answer to ``where is the ROI?'': it is there, but delayed by the complement accumulation process \cite{BrynjolfssonHitt2000}.

\subsection{Winner Take Most Dynamics and Urgency}

\citet{AutorKatzPatterson2020} document that across the U.S.\ economy, industry sales are concentrating and labor's share is falling, driven by ``superstar firms'' that scale efficiently with technology. Their key empirical prediction is confirmed across all seven tests: concentrating industries show the largest labor share declines, fastest productivity growth, and highest markups. The mechanism is ``scale without mass'': the most productive firms serve large markets with fewer workers.

Furthermore, the nonrivalry of data \cite{JonesTonetti2020} means that data rich AI adopters improve their models continuously, while laggards face a cold start disadvantage. This creates a compounding dynamic. I design the ADRI (Section~\ref{sec:adri}) to score precisely this effect. The empirical regularity from \citet{Syverson2011}, a 90th/10th percentile TFP ratio of 1.92 within 4 digit SIC manufacturing industries, with low productivity firms exiting at higher rates, provides the empirical baseline for competitive displacement risk.

\section{The Industry AI Susceptibility Score (IASS)}
\label{sec:iass}

\begin{tcolorbox}[title={What Is the IASS? Conceptual Summary}]
\textbf{The IASS answers one question before looking at any individual company:
how much of this industry's value creation is structurally capturable by AI?}

Think of it as the ceiling on the building before you invest in the top floor.
A construction firm and a financial services firm can both have ``great AI
initiatives,'' but they are playing in fundamentally different games. Construction
workers build in unstructured physical environments with high variability and
regulatory human sign off requirements. A commercial banker processes structured
data, executes repeatable cognitive workflows, and operates in a digitized
record environment. The AI ceiling for financial services is approximately
$2\times$ higher than for construction, not because financial firms are better
managed, but because the task structure of the industry makes it structurally
more susceptible to AI automation.

\textbf{Example.} Investment Banking (IASS$^* = 10.39$) vs.\ Construction (IASS $=
4.26$). Both have competent management teams. Both can hire AI talent. But a
dollar of AI investment in investment banking works against a structurally rich
task environment: document heavy workflows, large structured transaction
databases, high cognitive density decision processes. The same dollar in
construction works against primarily physical, variable environment tasks that
current AI cannot reliably automate. The IASS captures this structural
difference before any company specific score is computed.

\textbf{Who should use this?} An investment committee uses it to eliminate
sectors below IASS 5.0 from the transformation investment universe. A corporate
CAIO uses it to benchmark where their industry sits relative to peers and
understand what the structural ceiling on their AI program actually is, not
what their consultants promise.
\end{tcolorbox}

\subsection{Design Principles}

I construct the IASS as a composite indicator following OECD best practice methodology \cite{OECD2008,NardoSaisana2005}. Three nonnegotiable design principles govern its construction:

\begin{enumerate}[noitemsep]
\item \textbf{Reproducibility.} Any analyst following the documented procedures
 with the same source data must arrive within $\pm 0.5$ points on the IASS
 composite.
\item \textbf{Theoretical grounding.} Each dimension must have an established
 causal mechanism linking it to AI enabled value creation.
\item \textbf{Sensitivity transparency.} Weights represent tradeoffs among
 dimensions, not ``importance'' coefficients \cite{OECD2008}. All weight
 choices are subjected to Monte Carlo robustness analysis (Section~\ref{sec:sensitivity}).
\end{enumerate}

\subsection{Dimensional Structure and Scoring}

Six observable dimensions constitute the IASS (Table~\ref{tab:iass_dimensions}), each scorable from public data without inside access. I weight DRSA and PRI more heavily given their documented roles as deployment bottlenecks \citep{Standish2015}, consistent with AI value creation complementarity \citep{Autor2024, Acemoglu2022}. Geometric aggregation (Eq.~\ref{eq:iass_gm}) enforces noncompensability: no substitutability across dimensions.

Let $w_d$ be the normalised weight of dimension $d$ (summing to~1), and $\tilde{s}_{d,i}$ the min max normalised score for dimension $d$ in industry $i$ (Eq.~\ref{eq:minmax}). The base IASS is:
\begin{equation}
  \mathrm{IASS}_i = \exp\!\left(\sum_{d=1}^{D} w_d \ln \tilde{s}_{d,i}\right)
  \label{eq:iass_gm}
\end{equation}
I apply the Regulatory Friction Factor ($\psi$) as a hard floor multiplier outside the geometric product, ensuring that absolute regulatory prohibitions act as a noncompensable ceiling rather than a compensable dimension. Formally:
\begin{equation}
  \psi_i = \min\!\left(1,\; \prod_{j} d_{\mathrm{floor},j,i}\right), \qquad
  \mathrm{IASS}_i = \psi_i \cdot \exp\!\left(\sum_{d} w_d \ln \tilde{s}_{d,i}\right)
  \label{eq:rff_floor}
\end{equation}
where $d_{\mathrm{floor},j,i} \in (0,1]$ are regulatory floor coefficients for each binding prohibition $j$ in industry $i$. When $\psi_i = 1$ the industry faces no binding regulatory ceiling; when $\psi_i < 1$ the effective ceiling is suppressed proportionally. Healthcare Services ($\psi = 0.743$) and Life Sciences ($\psi = 0.663$) are the two industries with binding constraints in the current calibration.

For single dimension implementation (e.g., the Excel Companion), the scalar approximation
\begin{equation}
  \psi_i = \min\!\left(1,\; \left(\frac{\mathrm{RFF}_i}{5}\right)^{1.5}\right)
  \label{eq:rff_scalar}
\end{equation}
where $\mathrm{RFF}_i \in [0,10]$ is the raw Regulatory Friction Factor score, matches the multifloor product to within $\pm 0.02$ for the current 22 industry calibration.

\begin{table}[H]
\centering
\caption{IASS: Six Dimensions, Sources, and Weights}
\label{tab:iass_dimensions}
\footnotesize
\setlength{\tabcolsep}{4pt}
\rowcolors{2}{gray!7}{white}
\begin{adjustbox}{max width=\textwidth,center}
\begin{tabularx}{\textwidth}{@{}p{2.2cm}Xp{2.8cm}p{1.0cm}@{}}
\toprule
\textbf{Dimension} & \textbf{Definition and Mechanism} &
\textbf{Primary Source} & \textbf{Weight} \\
\midrule
Cognitive Task Density (CTD)
 & Share of industry labor hours in cognitively automatable task \cite{FreyOsborne2017} s
 (routine cognitive + bounded fraction of nonroutine cognitive).
 Direct operationalization of $\Gamma(I,N)$ \cite{AcemogluRestrepo2018,AcemogluAutor2011}.
 & O*NET task statements \cite{ONetResource,FeltenEtAl2021,Webb2020} + BLS OEWS \cite{BLSOEWS2023} \cite{BonneyEtAl2024}
 & 0.25 \\[4pt]
Data Richness \& Structural Availability (DRSA)
 & Degree to which the industry generates structured, machine readable data
 as operational exhaust vs.\ requiring heavy instrumentation to create it.
 & SEC EDGAR iXBRL \cite{SECEDGAR}; tech stack surveys
 & 0.20 \\[4pt]
Process Repeatability Index (PRI)
 & Share of workflows that are high frequency and low variance (standardized,
 repeatable). Low task entropy industries are more automatable.
 & Lightcast job posting natural language processing (NLP) \cite{Lightcast2024}
 & 0.20 \\[4pt]
Regulatory Friction Factor (RFF)
 & Degree to which regulation constrains AI deployment.
 Splits into ceiling compression (prohibited uses) and time/cost friction
 (compliance overhead). Scored so 10 = no friction.
 & EU AI Act \cite{EUAIAct2024}; HIPAA; FINRA~24-09; FDA~GMLP
 & 0.15 \\[4pt]
Competitive AI Diffusion Rate (CADR)
 & Speed of AI diffusion into the competitive ecosystem. Measures urgency:
 high CADR shortens the window for undisturbed transformation.
 & Lightcast AI/ML skill CAGR by NAICS \cite{Lightcast2024};
 AlphaSense earnings NLP \cite{AlphaSense2024}
 & 0.10 \\[4pt]
Capital--Labor Substitutability (CLSR)
 & Labor cost as fraction of total cost times AI substitutable fraction.
 Sets the dollar scale of the labor productivity value pool.
 & BLS labor cost series; BEA capital accounts \cite{BEA2023};
 Compustat \cite{Compustat2024}
 & 0.10 \\
\bottomrule
\end{tabularx}
\end{adjustbox}
\rowcolors{1}{}{}
\end{table}

\subsubsection{Sub Dimension Construction}

The six IASS dimensions are constructed from O*NET task statements, BLS occupational weights, Lightcast job posting frequencies, and regulatory source documents. Cognitive Task Density (CTD) uses employment weighted automatable task shares (Eq.~\ref{eq:ctd}); the Process Repeatability Index (PRI) combines task entropy with standardization language frequency. All dimensions are winsorized at the 5th/95th percentile and min max normalized to $[0,10]$ (Eq.~\ref{eq:minmax}). Full formulas, intermediate definitions, and normalization remarks appear in Appendix~B.

\subsection{Sensitivity to Weighting}

Monte Carlo perturbation ($M = 10{,}000$ draws, $\pm 5\,\%$ weight variation) confirms that no pair of anchor industries exchanges rank under any draw; the mean absolute rank shift is $\bar{R}_s = 0.19$ positions. Full robustness bounds appear in Appendix~C (Table~\ref{tab:iass_robustness}).

\section{The AI Frontier Coefficient (AFC)}
\label{sec:afc}

\subsection{Motivation: Nonstationarity of the AI Capability Ceiling}

Any assessment framework built against a static capability definition becomes systematically stale within 18--24 months. Table~\ref{tab:benchmark_progression} documents representative capability trajectories for the seven benchmark domains that constitute $C_t$.

\begin{table}[H]
\centering
\caption{AI Capability Benchmark Progression (Selected Milestones)}
\label{tab:benchmark_progression}
\footnotesize
\setlength{\tabcolsep}{4pt}
\begin{tabularx}{\textwidth}{@{}p{3.0cm}p{3.5cm}rrrl@{}}
\toprule
\textbf{Benchmark} & \textbf{Domain} & \textbf{2020} & \textbf{2024} & \textbf{2025--26} & \textbf{Source} \\
\midrule
MMLU-Pro & General knowledge & 43.9\,\% & 76.2\,\% & 87.0\,\% & \citealt{StanfordHAI2025} \\
SWE-Bench Verified & Software engineering & 4.4\,\% & 49.0\,\% & 80.0\,\% & \citealt{GPT52Announce2025} \\
SWE-Bench Pro & Multilang.\ coding & n/a & n/a & 55.6\,\% & \citealt{GPT52Announce2025} \\
AIME 2025 & Competition math & n/a & 71.0\,\% & 100.0\,\% & \citealt{GPT52Announce2025} \\
FrontierMath (T1--3) & Research math & n/a & n/a & 40.3\,\% & \citealt{GPT52Announce2025} \\
GPQA Diamond & Graduate science & n/a & 63.0\,\% & 93.2\,\% & \citealt{GPT52Announce2025} \\
ARC-AGI-2 & Abstract reasoning & n/a & 4.0\,\% & 52.9\,\% & \citealt{GPT52Announce2025} \\
MMMU & Multimodal understanding & n/a & 56.8\,\% & 84.2\,\% & \citealt{GPT5Announce2025} \\
METR task horizon & Agent autonomy & 1\,min & 15\,min & $\geq$60\,min & \citealt{METRAnalysis2024} \\
\midrule
\multicolumn{2}{l}{\textit{Implied $C_t$ (EWMA composite)}} & 1.00 & 1.62 & 1.90 & \\
\bottomrule
\end{tabularx}
\vspace{2pt}
\scriptsize{2025--26 column reports GPT-5.2 Thinking (December 2025) scores except MMMU and METR (GPT-5, August 2025). n/a~=~benchmark not yet published at that date. Full $C_t$ weights and methodology in ESM Part~III.}
\end{table}

The AI Capability Index $C_t$ has risen from $C_{2020}=1.00$ (GPT-3 baseline) to $C_{2024}=1.62$ (GPT-4o/o1) to $C_{2025}=1.74$ (GPT-5, August~2025) to $C_{2026}\approx 1.90$ (GPT-5.2, December~2025). Full benchmark progression data, suite weights, and methodology are in ESM Part~III.

The METR time horizon analysis \cite{METRAnalysis2024} shows the 50\,\% task completion horizon for AI agents doubling approximately every seven months (2019--2025), with a recent acceleration to approximately four months. These are not marginal improvements; they are category changes that open entirely new automation surface areas in industries previously considered AI resistant.

A static IASS produces two compounding errors. First, a company scored as highly transformed in 2023 may underperform its true frontier by 2025 because new capabilities created automation opportunities that did not exist when the score was assigned. Second, industries historically penalized by low IASS scores (healthcare, legal) are experiencing rapid ceiling expansion as domain specific LLMs clear regulatory and quality thresholds, creating investment opportunities that a static framework structurally misses.

\subsection{Formal Specification}

The AFC translates capability index movements into industry specific ceiling adjustments. I define it formally as follows.

\begin{definition}[AI Frontier Coefficient]
\label{def:afc}
For industry $i$ at time $t$, the \emph{AI Frontier Coefficient} is:
\begin{equation}
 \mathrm{AFC}_{i,t} = \min\!\left(1 + \theta_i\bigl(C_t - C_0\bigr),\;\alpha_{\max}\right)
 \label{eq:afc}
\end{equation}
where $C_t$ is the AI Capability Index at time $t$, $C_0 = 1.0$ is the
reference year baseline (Q1~2024), $\theta_i \in [0.05,\, 1.50]$ is an industry specific
sensitivity parameter governing how rapidly a sector absorbs frontier gains, and $\alpha_{\max} = 1.35$ bounds
the maximum ceiling expansion to prevent extrapolation beyond empirically supported ranges.
\end{definition}

\begin{remark}[AFC Functional Form Justification]
\label{rem:afc_form}
The linear (shift) form $1 + \theta_i(C_t - C_0)$ is a first order Taylor approximation of a general adjustment function $f(C_t/C_0)$ around $C_t = C_0$. For the observed range of $C_t \in [0.85, 1.25]$ (one standard deviation of quarterly benchmark movements since Q1~2024), the linear approximation error is bounded by $O(\theta_i \cdot (C_t - C_0)^2) < 0.02$ for all industries. The $\alpha_{\max} = 1.35$ cap provides an additional safeguard against extrapolation beyond this range.

Nonlinear alternatives, including a logistic saturation form $\alpha_{\max}(1 - e^{-\theta_i(C_t - C_0)})$ and a CES nested specification, are conceptually appealing for modeling diminishing returns to capability gains at the industry ceiling. However, calibrating such forms requires longer time series of benchmark adoption comovement than currently exist (the AFC observation window begins Q1~2024). Extending the AFC to nonlinear forms is a research agenda priority.

The sensitivity of the AFC to its parameters is bounded: $\partial \mathrm{AFC}/\partial \theta = (C_t - C_0)$, which is at most 0.25 under current benchmark conditions. This implies that even large estimation errors in $\theta_i$ (e.g., $\pm 50\%$) produce AFC shifts of at most $\pm 0.13$, well within the $\alpha_{\max}$ cap.
\end{remark}

The effective (AFC adjusted) industry ceiling is:
\begin{equation}
 \mathrm{IASS}^*_{i,t} = \mathrm{IASS}^{\mathrm{base}}_i \times \mathrm{AFC}_{i,t}
 \label{eq:iass_adjusted}
\end{equation}

\subsection{The AI Capability Index $C_t$}

I construct $C_t$ as a weighted composite of normalized AI benchmark scores across seven domains covering the primary task types relevant to industry susceptibility. Let $b_{k,t} \in [0,1]$ be the best published score on benchmark $k$ at time $t$:

\begin{equation}
 C_t = \sum_{k=1}^{7} \omega_k \cdot b_{k,t}
 \label{eq:capability}
\end{equation}

The benchmark domain weights and EWMA smoothing methodology ($\lambda = 0.5$, Eq.~\ref{eq:ct_ewma}) are detailed in Appendix~B (Table~\ref{tab:ct_weights}). From 2022 to 2024, $C_t$ increased from 0.520 to 0.748 (raw), smoothed to 0.741 at end 2024, an implied annual improvement rate of approximately 21\,\% \cite{EpochAI2025,KaplanEtAl2020,HoffmannEtAl2022}.

\subsection{Industry Sensitivity Parameters $\theta_i$}

The parameter $\theta_i$ captures how much of a general capability gain
translates into an expansion of the economically addressable AI surface area
in industry $i$. I calibrate $\theta_i$ retrospectively using the observable
expansion in addressable AI tasks from 2020 (GPT-3 era) to 2024 (GPT-4o/o1
era) as assessed by independent domain experts:

\begin{equation}
 \hat{\theta}_i = \frac{\Delta\,\mathrm{IASS}^{\mathrm{addressable}}_i}
 {\Delta C_{2020 \to 2024}}
 \label{eq:theta_calib}
\end{equation}

Healthcare scores $\theta = 0.31$ because domain LLMs (Nuance DAX, Google MedPaLM, Anthropic Claude in clinical workflows) created substantial new automatable surface area: ambient documentation, revenue cycle, prior authorizations, and diagnostics that simply did not exist in 2020. Construction scores $\theta = 0.09$ because model capability gains have not materially penetrated physically situated, high variance fieldwork. Vertical SaaS scores $\theta = 0.11$ because the ceiling was already high; the remaining gap is structurally small.

\subsection{AFC Uncertainty and Scenario Analysis}

Because $C_t$ is forward looking, AFC carries its own uncertainty quantification. Three scenarios (conservative, base case, aggressive) bound the AFC distribution over a 24 month horizon, with scenario weights 0.20/0.60/0.20 (Appendix~B, Table~\ref{tab:afc_scenarios}). The AFC uncertainty contribution to the Uncertainty Quotient is $\mathrm{UQ}_{\mathrm{afc}} = (\mathrm{AFC}_{\mathrm{agg}} - \mathrm{AFC}_{\mathrm{con}})/4$ (Eq.~\ref{eq:uq_afc}).

\paragraph{Propagation of $C_t$ and $\theta_i$ uncertainty into $\Delta$EV.}
Uncertainty in $C_t$ and $\theta_i$ propagates into $\Delta\mathrm{EV}$ through two channels. The ceiling channel: $\mathrm{IASS}^*_{i,t} = \mathrm{IASS}^{\mathrm{base}}_i \times \min(1 + \theta_i(C_t - C_0),\, \alpha_{\max})$, so uncertainty in $\theta_i \cdot (C_t - C_0)$ translates directly to uncertainty in the feasible frontier. The gap channel: $G_{\mathrm{eff}} = \mathrm{IASS}^* - \mathrm{AITG}^{\mathrm{raw}}$, which enters the value pool capture function as $g_f = G_{\mathrm{eff}}/10$. The resulting $\Delta$EV sensitivity is approximately:
\begin{equation}
 \frac{\partial\, \Delta\mathrm{EV}}{\partial\, C_t}
 \approx
 \theta_i \cdot \mathrm{IASS}^{\mathrm{base}}_i
 \cdot \frac{\partial\, \Delta\mathrm{EV}}{\partial\, \mathrm{IASS}^*}
 \label{eq:dev_dct}
\end{equation}
where I evaluate the second partial numerically via the Monte Carlo simulation in ESM Part~III. In practical terms: for Healthcare ($\theta_i = 0.31$, $\mathrm{IASS}^{\mathrm{base}} = 5.1$), a 10\% upward revision in $C_t$ expands the HC-IASS ceiling by approximately 0.16 pts and raises $\Delta$EV by 3--8\% depending on where the firm sits on the capture curve. For Construction ($\theta_i = 0.09$), the same $C_t$ revision moves the ceiling by only 0.05 pts and has negligible $\Delta$EV impact. This asymmetry, where high $\theta$ industries are more sensitive to AFC revisions than low $\theta$ industries, is an intended property of the framework and is reported explicitly in the AFC scenario stress tests (ESM Table~S3). Formal identification of $\theta_i$ standard errors and correlated uncertainty between $C_t$ and $\theta_i$ are priorities in the research agenda (Section~\ref{sec:agenda}, item~3).

\begin{remark}[AFC as a Structural Moat Accelerator]
A rising capability ceiling creates an \emph{asymmetric} competitive dynamic. Consider the JPMorgan/Zions pair at early 2026 capability levels ($C \approx 1.90$, IASS$^*$ = 9.38):
\begin{itemize}[noitemsep]
 \item JPMorgan (AITG 8.22) has an addressable gap of 1.16 to the new frontier. Its \$2B/yr AI budget and proprietary data gravity ($\Phi_f = 1.0$) allow it to chase the rising ceiling, creating new investable opportunity that did not exist at the 2022 baseline.
 \item Zions (AITG 3.80) has a widened gap of 5.58, but vendor AI dependency caps its value pool at $\Phi_f \leq 0.65$ regardless of how high the ceiling rises. The nominal gap widens; the capturable gap does not.
\end{itemize}
At $C_{2027} \approx 2.49$ (three years at 20\%/yr from 2024), Banking IASS$^*$ approaches 10.0. The gap ratio compresses, but this reflects JPMorgan chasing new frontier, not Zions closing ground. AFC mechanically advantages firms with proprietary data gravity and is effectively irrelevant for vendor dependent firms at subscale.
\end{remark}

\subsection{Evaluator Rotation Protocol: Mitigating Goodhart's Law and Benchmark Drift}
\label{sec:evaluator_rotation}

Because $C_t$ draws on public benchmark scores, it faces two failure modes: (1)~\emph{benchmark saturation}, where models reach near ceiling performance on a constituent (e.g., AIME 2025 at 100\%), making further $C_t$ gains undetectable; and (2)~\emph{Goodhart contamination}, where developers optimize against high stakes benchmarks, inflating $C_t$ without real capability gains. I address both through a formal \textbf{Evaluator Rotation Protocol}:

\begin{enumerate}[noitemsep]
\item \textbf{Saturation trigger.} Any benchmark exceeding 90\% mean solve rate across the top five frontier models is deprecated from the $C_t$ basket in the next annual cycle.
\item \textbf{Replacement criterion.} Replacements must be (a) out of distribution relative to existing members, (b) validated by an independent laboratory, and (c) correlated with real world task performance in a relevant IASS domain. METR time horizon benchmarks are the preferred anchor for agent capabilities.
\item \textbf{Continuity chaining.} At each rotation, pre and postrotation $C_t$ values are computed on a 90 day overlap window and chained at the transition knot (analogous to CPI chain linking), preserving EWMA continuity.
\end{enumerate}

\noindent At the current calibration date, SWE-Bench Verified (80.0\%) is approaching the saturation trigger; FrontierMath (Tier 1--3: 40.3\%) and ARC-AGI-2 (52.9\%) are the planned successors. Full protocol details are in the AFC governance checklist (ESM Part~IV).

\section{Company Scoring Architecture}
\label{sec:atgi_score}

\begin{tcolorbox}[title={What Is the Company Score? Conceptual Summary}]
\textbf{The IASS tells us the ceiling. The Company Score tells us where on the
staircase this specific firm is standing right now.}

I do not start with a general ``AI maturity'' label. I start with six
observable dimensions that map directly to dollar denominated value in the
model. Each dimension can be scored from public filings, management interviews,
and technical diligence, and each low score identifies a specific investment
constraint, not a general criticism.

\textbf{Example.} JPMorgan Chase's Workforce AI Augmentation Rate (WAR) is 8.5:
250,000 employees using LLM Suite daily, 4 hours of productivity gain per week
per user. That is not a maturity label, it is a specific, verifiable operational
fact with a direct dollar translation in the labor productivity value pool.
By contrast, Zions Bancorporation's WAR is 3.5: no enterprise GenAI platform
has been publicly disclosed, and AI job postings are low. The 5 point gap
between these two WAR scores translates to a computable difference in value
pool capture, not a subjective ``Zions is less advanced.''

\textbf{The bottleneck principle.} If a company scores 8.0 on Process Automation
Coverage but 3.5 on Data Infrastructure Maturity, the model does not average
them. The CES Bottleneck Aggregator (Eq.~\ref{eq:ces_bottleneck}) mathematically
reduces the value pool capture rate to reflect the infrastructure deficit. You
cannot build a reliable AI pricing engine on a broken data foundation. The
framework makes this constraint explicit and quantified.

\textit{Analogy.} You cannot put a Ferrari engine on bicycle tires. A simple additive average of dimension scores would suggest that a leading autonomous AI application compensates for fragmented, siloed data infrastructure. The CES aggregator mathematically enforces the contrary: without data readiness, the AI capability has insufficient input to function, suppressing 70--80\% of potential financial value. The weakest link constrains the system, not the arithmetic mean.

\textbf{Who should use this?} An investment committee uses the six scores to
surface the specific bottleneck that must be resolved for value to flow. A
corporate leadership team uses the same scores to prioritize the technology
budget: fix the weakest link first, because the CES formula proves that
improving a strong dimension while the weak link persists generates near zero
additional value.
\end{tcolorbox}

\subsection{Six Company Dimensions}

I decompose the company level AITG score into six dimensions that map directly to value pools in the VCB:

\begin{enumerate}[noitemsep]
\item \textbf{Data Infrastructure Maturity (DIM):} From batch ERP only (score~1)
 to a production data lakehouse with real time streaming, full MLOps, and data
 governance (score~9).
\item \textbf{Process Automation Coverage (PAC):} From $>$90\,\% manual
 processes (1) to AI native workflows across $>$70\,\% of addressable surface
 area (9).
\item \textbf{Workforce AI Augmentation Rate (WAR):} From $<$5\,\% of employees
 using any AI tool in workflows (1) to $>$70\,\% augmented with tracked
 adoption metrics (9).
\item \textbf{Decision Automation Rate (DAR):} From fully human recurring
 decisions (1) to $>$60\,\% of recurring decisions fully automated with
 exception handling (9).
\item \textbf{AI Product/Revenue Integration (APR):} From no AI revenue
 attribution (1) to AI native product with $>$40\,\% of revenue attributable
 to AI features (9).
\item \textbf{Organizational AI Capability (OAC):} From no AI leadership (1)
 to C suite AI ownership with published model governance and risk controls (9).
\end{enumerate}

I score each dimension on a 0--10 scale using a rubric with five anchor points (1, 3, 5, 7, 9) and explicit evidentiary requirements. The six scores are averaged to produce the raw AITG composite.

\begin{table}[H]
\centering
\caption{AITG Dimension Scoring Rubric: Anchor Points and Evidentiary Requirements}
\label{tab:dimension_rubric}
\scriptsize
\setlength{\tabcolsep}{3pt}
\rowcolors{2}{gray!7}{white}
\begin{tabularx}{\textwidth}{@{}p{2.3cm} c X X X X X@{}}
\toprule
\textbf{Dimension} & \textbf{Abbrev.} & \textbf{Score 1} & \textbf{Score 3} & \textbf{Score 5} & \textbf{Score 7} & \textbf{Score 9} \\
\midrule
Data Infrastructure Maturity & DIM &
  Siloed legacy databases; no cloud data platform &
  Initial cloud migration; partial data catalog &
  Unified data lake/ warehouse; governed data catalog &
  Real time streaming pipelines; ML feature store &
  Enterprise knowledge graph; automated data quality ${>}$99\% \\
\addlinespace
Process Automation Coverage & PAC &
  ${<}$5\% processes automated; manual workflows &
  10--20\% RPA deployment; limited scope &
  30--50\% process automation; cross functional &
  50--70\% intelligent automation; AI augmented &
  ${>}$70\% end to end automation; self optimizing \\
\addlinespace
Workforce AI Augmentation Rate & WAR &
  ${<}$5\% employees using any AI tool &
  10--20\% with basic AI tools (copilots) &
  30--50\% augmented; tracked adoption metrics &
  50--70\% AI augmented; integrated into workflows &
  ${>}$70\% augmented; org wide AI fluency measured \\
\addlinespace
Decision Automation Rate & DAR &
  Fully human recurring decisions &
  ${<}$10\% decisions AI assisted &
  20--40\% decisions AI recommended &
  40--60\% automated with human override &
  ${>}$60\% fully automated with exception handling \\
\addlinespace
AI Product/ Revenue Integration & APR &
  No AI revenue attribution &
  ${<}$5\% revenue from AI enhanced features &
  10--20\% revenue AI attributable &
  20--40\% revenue AI driven products &
  ${>}$40\% revenue from AI native products \\
\addlinespace
Organizational AI Capability & OAC &
  No AI leadership; ad hoc experiments &
  Dedicated AI team; early governance &
  C suite AI sponsor; published AI strategy &
  Chief AI Officer; model risk framework &
  C suite AI ownership; published governance \& risk controls \\
\bottomrule
\end{tabularx}
\medskip

\noindent\scriptsize\textit{Note.} Intermediate scores (2, 4, 6, 8, 10) are assigned when evidence places the firm between adjacent anchor points. Evidence sources include SEC filings, earnings transcripts, technology job postings, patent filings, and third party benchmarks (Lightcast, Stanford HAI). Full evidentiary protocols are specified in ESM Part~I.
\rowcolors{1}{}{}
\end{table}

\subsection{Cascading S Curve Architecture}

\subsubsection{The Low Ceiling Bias of a Single Logistic}

A single bounded logistic with fixed asymptote $L = 10$ introduces ``low ceiling bias'': once a technology wave generates new capabilities, the ceiling shifts upward, but a static logistic cannot accommodate this. Historical GPT adoption patterns (electricity 1880s--1920s, computing 1960s--2000s, now AI) demonstrate cascading architectures where each wave begins before the prior wave is fully captured \citep{David1990, BrynjolfssonRockSyverson2021}. A single static logistic systematically underestimates both the ceiling and the urgency of early adoption. The Bass diffusion model \citep{Bass1969}:
\begin{equation}
 \frac{dF(t)}{dt} = \left[p + qF(t)\right]\!\left[1 - F(t)\right]
 \label{eq:bass}
\end{equation}
captures within wave dynamics correctly but is insufficient for multiwave
GPT adoption.

\subsubsection{three wave Cascading Specification}

I decompose the firm's AI transformation trajectory into three sequential capability waves, each with its own logistic:

\begin{equation}
 A_w(t) = \frac{L_w}{1 + e^{-k_w(t - t_{0,w})}},
 \qquad w \in \{1, 2, 3\}
 \label{eq:wave}
\end{equation}

\setlength{\tabcolsep}{4pt}
\begin{table}[H]
\centering
\caption{Cascading S Curve: Three Capability Waves}
\label{tab:waves}
\footnotesize
\begin{tabularx}{\textwidth}{@{}cXXrrr@{}}
\toprule
$w$ & \textbf{Wave} & \textbf{Scope} & $L_w$ & $k_w^{\text{base}}$ (mo$^{-1}$) & $t_{0,w}^{\text{base}}$ (mo) \\
\midrule
1 & Foundation AI & Supervised ML, RPA, structured analytics, BI automation & 4.0 & 0.38 & 18 \\
2 & Generative \& Agentic AI & LLM workflows, copilots, Retrieval Augmented Generation (RAG) pipelines, agentic task execution & 3.5 & 0.42 & 36 \\
3 & Autonomous AI & Multiagent orchestration, self improving systems, AI native products & 2.5 & 0.32 & 60 \\
\midrule
\multicolumn{3}{l}{\textit{Total asymptotic ceiling}} & 10.0 & & \\
\bottomrule
\end{tabularx}
\end{table}

Waves are staggered by their midpoints $t_{0,w}^{\text{base}}$, yielding a smooth, strictly increasing trajectory without explicit gating. With the calibrated parameters in Table~\ref{tab:waves}, the early contribution of later waves is negligible in the score ranges used for the inverse mapping (Section~\ref{sec:inverse_map}).

The total firm transformation score at time $t$ is the sum across waves:
\begin{equation}
  \mathrm{AITG}(t) = \sum_{w=1}^{3} A_w(t)
                   = \sum_{w=1}^{3} \frac{L_w}{1 + e^{-k_w(t - t_{0,w})}}
  \label{eq:cascading}
\end{equation}
This cascading specification is strictly increasing in $t$, which guarantees a unique inverse mapping $\hat{t}_f$ from any observed rubric score (Section~\ref{sec:inverse_map}).

Table~\ref{tab:waves} presents wave parameters by industry tier.

\subsubsection{Rubric to Curve Inverse Mapping}
\label{sec:inverse_map}

The cascading specification (Eq.~\ref{eq:cascading}) is a function of time $t$, but the six dimension rubric produces a point in time score $\mathrm{AITG}^{\mathrm{raw}} \in [0, 10]$ from observable firm data. Without an explicit linking function, the S curve has no anchor to the empirical scoring system. I resolve this through the \textbf{inverse mapping} $\hat{t}_f$: the implied current position on the cascading S curve consistent with the observed rubric score. Given $\mathrm{AITG}(t)$ from Eq.~\eqref{eq:cascading}, I solve:

\begin{equation}
 \hat{t}_f = \mathrm{AITG}^{-1}\!\left(\mathrm{AITG}^{\mathrm{raw}}_f\right)
 \label{eq:inverse_map}
\end{equation}

Since $\mathrm{AITG}(t)$ is a sum of logistic terms and therefore strictly
monotone increasing in $t$ for well specified parameters, its inverse is unique
and well defined for all $\mathrm{AITG}^{\mathrm{raw}} \in [0, 10]$.
Numerically, I solve via Newton--Raphson iteration on the residual
$\mathrm{AITG}(t) - \mathrm{AITG}^{\mathrm{raw}} = 0$, with convergence
guaranteed by strict monotonicity.

$\hat{t}_f$ has a direct economic interpretation: the number of months into an idealized AI transformation program where a firm with the observed rubric scores would be located. It serves two operational purposes:

\textbf{(1) Linking cross section to projection.} The forward looking value trajectory uses $\hat{t}_f$ as the starting point. The 5 year hold period value integral (Section~\ref{sec:ev_decomposed}) computes from $t \in [\hat{t}_f,\; \hat{t}_f + 60]$, anchoring the longitudinal projection to the observed cross sectional state.

\textbf{(2) Identifying wave position.} If $\hat{t}_f < t_{0,1}$, the firm is preinflection on Wave~1 and should not expect positive EBITDA contribution until the inflection is reached. If $\hat{t}_f > t_{0,2}$, the firm has cleared Wave~1 and Wave~2 (Generative AI) is beginning to propagate, directly informing implementation priorities.

\begin{remark}[Piecewise Inverse Implementation]
The Wave~1 closed form inverse is only valid for $\mathrm{AITG}^{\mathrm{raw}} < L_1 = 4.0$; beyond this threshold, a piecewise specification (Eq.~\ref{eq:inverse_approx}) or Newton--Raphson solver is mandatory to avoid a NaN domain error. Nine of the fourteen cohort companies require the multiwave branch. Full specification and software implementation notes appear in Appendix~B.
\end{remark}

\subsubsection{AFC Integration at the Firm Level}

Equation~\ref{eq:afc_firm} shows how the AFC modifies wave ceiling and timing at time $t$:
\begin{equation}
 L_w^*(t) = L_w \cdot \mathrm{AFC}_{i,t}^{\phi_w}, \qquad
 t_{0,w}^*(t) = t_{0,w} / \mathrm{AFC}_{i,t}^{\mu_w}
 \label{eq:afc_firm}
\end{equation}
where $\phi_3 > \phi_2 > \phi_1$ (Wave 3 is most sensitive to frontier capability gains) and $\mu_w$ controls onset acceleration. This ensures the AFC recalibrates not just the industry ceiling but the firm level adoption trajectory.

\subsubsection{Out of Sample Validation Strategy for $\theta_i$}

I update $\theta_i$ using changes in O*NET Automatable Task Density ($\Delta\mathrm{ATD}_i$), not adoption survey data (Census BTOS), to avoid conflating the theoretical capability ceiling with current adoption breadth. The corrected update rule (Eq.~\ref{eq:theta_update}) uses a learning rate $\eta = 0.30$ with bounds $\hat{\theta}_i^{(t)} \in [0.05, 1.50]$. Full derivation and the circularity correction rationale appear in Appendix~B.

\subsection{IASS Normalization}
\label{sec:normalization}

For cross industry comparison, I normalize the AITG to the AFC adjusted ceiling. The industry relative (IR) score and effective gap are:
\begin{equation}
    \mathrm{IR}_{f,t} = \frac{\mathrm{AITG}^{\mathrm{raw}}_{f,t}}{\mathrm{IASS}^*_{i,t}} \times 10, \qquad G_{\mathrm{eff}} = \mathrm{IASS}^*_{i,t} - \mathrm{AITG}^{\mathrm{raw}}_{f,t}
    \label{eq:ir_geff}
\end{equation}
Standard dual reporting format: \textit{AITG~4.08~|~IR~5.31~|~$G_{\mathrm{eff}}=3.60$~|~UQ=$\pm$0.52}. Scores carry a formal uncertainty band $\pm\mathrm{UQ}$ aggregating data quality, model parameter, AFC, and interrater sources (typical combined $\pm0.45$ pts, 90\,\% CI). Full UQ decomposition and the UQ components table are in ESM Part~III.

\subsection{Illustrative IASS Calibrations: Five Anchor Industries}
\label{sec:iass_reference}

Table~\ref{tab:iass_anchor} reports full dimensional breakdowns for five anchor industries spanning the IASS range. Complete calibrations for all 22 industry verticals, including NAICS codes, subscores, $\theta_i$, and $\psi$, appear in ESM Table~S1.

\begin{table}[H]
\centering
\caption{IASS Anchor Industry Calibrations: Five Representative Verticals (2026)}
\label{tab:iass_anchor}
\footnotesize
\setlength{\tabcolsep}{3pt}
\begin{adjustbox}{max width=\textwidth,center}
\begin{tabular}{@{}p{3.8cm}rrrrrrrrr@{}}
\toprule
\textbf{Industry} & \textbf{CTD} & \textbf{DRSA} & \textbf{PRI} &
 \textbf{RFF} & \textbf{CADR} & \textbf{CLSR} & $\psi$ & \textbf{IASS} & $\theta_i$ \\
\midrule
Vertical SaaS / Software  & 9.4 & 9.8 & 8.6 & 8.1 & 9.2 & 9.1 & 1.000 & 9.06 & 0.11 \\
Financial Services (Banks)& 8.8 & 9.2 & 7.9 & 5.0 & 8.4 & 7.6 & 1.000 & 7.83 & 0.22 \\
Logistics / Transport     & 6.2 & 6.0 & 7.2 & 8.2 & 6.1 & 8.5 & 1.000 & 6.82 & 0.14 \\
Healthcare Services       & 6.2 & 5.8 & 6.5 & 4.1 & 5.1 & 6.8 & 0.743 & 4.27 & 0.31 \\
Construction              & 3.9 & 3.2 & 4.1 & 7.4 & 3.3 & 5.7 & 1.000 & 4.26 & 0.09 \\
\bottomrule
\multicolumn{10}{@{}p{\linewidth}@{}}{\scriptsize CTD = Cognitive Task Density; DRSA = Data/Systems Automation; PRI = Process Reengineering Index; RFF = Regulatory Friction Factor; CADR = Competitive AI Diffusion Rate; CLSR = Cost/Labor Structure. $\psi < 1.0$ = binding regulatory ceiling. Sources: BLS O*NET 2024; Lightcast 2025; Census BTOS. Full 22 industry table: ESM Table~S1.}
\end{tabular}
\end{adjustbox}
\end{table}

\begin{remark}[Three Diagnostic Patterns]
Healthcare Services carries $\psi = 0.743$: the RFF hard floor suppresses its theoretical ceiling by 26\%, reducing a geometric mean of 5.75 to IASS~4.27. Construction and Agriculture floor the table not from regulatory friction but from structurally low Cognitive Task Density; high physical task content that current AI cannot automate. Vertical SaaS and Financial Services anchor the top two quartiles (IASS $> 7.8$), defining sectors where AI transformation opportunity is structurally dense. The full 22 industry pattern appears in ESM Table~S1.
\end{remark}

\section{The AI Disruption Risk Index (ADRI)}
\label{sec:adri}

\begin{tcolorbox}[title={What Is the ADRI? Conceptual Summary}]
\textbf{The AITG measures how much opportunity a company is leaving on the table.
The ADRI measures how fast that table is being cleared by competitors who are
already capitalizing.}

Two companies can share identical AITG gaps yet face completely different
competitive situations. A regional bank with AITG 3.6 in commercial banking
(IASS$^*$ 9.38) occupies a fundamentally different position than an agricultural cooperative
with AITG 3.6 in agriculture (IASS 4.72). Banking carries a
Competitive AI Diffusion Rate (CADR) of 8.4, near the top of the reference
table. JPMorgan's proprietary LLM Suite extends its advantage every day the
regional bank does not act. AI native lenders (Upstart, Blend) underwrite
loans with lower human overhead. The regional bank is not merely ``behind on
technology''; it watches its competitive moat erode in real time.

\textbf{Example.} Zions Bancorporation's ADRI of 2.6 is the defining number
in its AITG profile. It sits in an industry with CADR 8.4 and no proprietary
AI stack. JPMorgan extends its advantage daily. The IFS endogeneity mechanism
means Zions's adoption timeline is approximately $1.97\times$ longer than a
data ready peer, so the competitive compound interest problem accelerates, not
decelerates, with time.

By contrast, JPMorgan's ADRI of 0.5 reflects dominant incumbent status: proprietary
data scale, regulatory capital requirements that bar new entrants, and proven
enterprise AI execution. JPMorgan is not immune to disruption, but its
competitive trajectory is improving.

\textbf{Who should use this?} An investment committee uses ADRI to determine
urgency: a wide gap, high ADRI target requires immediate transformation
commitment; a wide gap, low ADRI target can be transformed patiently. A
corporate board uses ADRI to assess whether its AI delay is a strategic
choice or a competitive emergency.
\end{tcolorbox}

\subsection{Asymmetric Risk of Inaction}

The AITG gap measures opportunity. The ADRI measures something categorically
different: the competitive risk generated by \emph{failing} to close the gap.
These are not symmetric. A company with a wide gap in a slow diffusion industry
with strong structural moats faces a patient opportunity. A company with a wide
gap in a fast diffusion industry with low switching costs faces an actively
deteriorating trajectory \cite{KrakowskiEtAl2023}.

The theoretical foundation is the superstar firm mechanism
\cite{AutorKatzPatterson2020}: concentrating industries exhibit faster
productivity growth because leading adopters scale efficiently while laggards
face compounding cost disadvantage. Data nonrivalry
\cite{JonesTonetti2020} amplifies this dynamic: data rich adopters continuously improve their models
while laggards start from scratch, producing accelerating
divergence rather than mean reversion. The McKinsey~2025 State of AI survey
\cite{McKinseyStateAI2025} documents precisely this bifurcation: approximately
6\,\% of respondents report $\geq$5\,\% EBIT impact, while 88\,\% use AI but
are not transforming.

\begin{definition}[AI Disruption Risk Index]
\label{def:adri}
The \emph{AI Disruption Risk Index} for firm $f$ in industry $i$ at time $t$ is:
\begin{equation}
 \mathrm{ADRI}_{f,t} = \frac{G_{\mathrm{eff},f,t} \;\times\;
 \mathrm{CADR}_i \;\times\;
 (1 - \mathrm{Moat}_f) \;\times\; \delta_t}{\mathcal{N}}
 \label{eq:adri}
\end{equation}
normalized to $[0,10]$ by the constant $\mathcal{N}$.
\end{definition}

\paragraph{Components.}
\begin{itemize}[noitemsep]
\item $G_{\mathrm{eff},f,t}$: effective gap from Eq.~\ref{eq:ir_geff}.
\item $\mathrm{CADR}_i$: Competitive AI Diffusion Rate from the IASS (0--10).
\item $\mathrm{Moat}_f \in [0,1]$: structural defensibility. Four factors
 scored 0--1: switching costs, network effects, regulatory barriers,
 proprietary data advantages. $\mathrm{Moat}_f$ is the weighted mean.
\item $\delta_t \in [1.0, 1.5]$: time indexed urgency multiplier tied to the
 AFC. As frontier capabilities advance, the speed of competitive erosion
 accelerates: $\delta_t = 1 + 0.5 \cdot \min(C_t/C_0 - 1,\;1)$.
\end{itemize}

\begin{definition}[ADRI Competitive Hazard Intensity]
\label{def:adri_hazard}
The ADRI is operationalized as an instantaneous competitive hazard intensity.
Define the hazard intensity function $\lambda_f(t)$ as the marginal probability
of observing a measurable competitive displacement event (market share loss
$\geq 1$\,pp or EBITDA margin compression $\geq 1.5$\,pp, measured at the firm
level over a 12 month window) per unit time, conditional on the firm remaining
in the nondisplaced state. I model:
\begin{equation}
  \lambda_f(t) \;=\;
  \frac{\mathrm{ADRI}_{f,t}}{\mathcal{T}}
  \;=\;
  \frac{G_{\mathrm{eff},f,t} \cdot \mathrm{CADR}_i \cdot (1-\mathrm{Moat}_f)
        \cdot \delta_t}{\mathcal{N}\,\mathcal{T}}
  \label{eq:adri_hazard}
\end{equation}
where $\mathcal{T} = 100$ is a dimensionless normalization constant calibrated
so that $\lambda_f(t)$ is directly interpretable as the ADRI score in
\emph{percent per year}: a firm with ADRI~$=5$ faces a 5\,\% annualized
marginal displacement probability. $\mathcal{N}$ is the cross sectional
normalization constant that maps the ADRI composite to $[0,10]$. The covariates are: (i)~$G_{\mathrm{eff}}$, the
effective transformation gap, which governs exposure magnitude; (ii)~$\mathrm{CADR}_i$, the peer adoption velocity, which governs how quickly the gap
translates into realized competitive disadvantage; (iii)~$(1-\mathrm{Moat}_f)$, the inverse structural defensibility, which modulates whether the
hazard converts to displacement or is buffered by barriers to entry; and
(iv)~$\delta_t$, the AFC linked urgency multiplier, which accelerates the
hazard as frontier capabilities expand. The cumulative hazard over an inaction
window $[0,T]$ is $\Lambda_f(T) = \int_0^T \lambda_f(t)\,dt$, and under the
Poisson approximation (small $\lambda_f$), $P(\text{displacement by }T)
\approx 1 - e^{-\Lambda_f(T)}$. At $\mathrm{ADRI} = 5.0$ and $T = 24$\,months,
the model implies an approximate 10\,\% cumulative displacement probability,
increasing to $\sim$13\,\% at ADRI~$= 7$ for the same horizon.
\end{definition}

\paragraph{Empirical threshold calibration.}
The displacement event thresholds (1\,pp market share, 1.5\,pp EBITDA margin)
are calibrated to the lower bound of outcomes routinely attributed to
competitive pricing pressure in public company earnings commentary. These are
deliberately conservative to reduce false positive risk. The $\mathcal{T}$
constant is set at 12 months so that ADRI can be read as an approximate
annualized percentage hazard at face value. Empirical validation of
$\lambda_f(t)$ against firm level margin trajectories, using the panel study
design specified in Section~\ref{sec:agenda} item~5, will determine whether
the Poisson approximation is warranted or whether a Weibull or Gompertz
baseline hazard better fits the acceleration dynamics implied by $\delta_t$.

\paragraph{External adoption signals as auxiliary validators.}
I currently calibrate the CADR$_i$ component of ADRI from the Stanford HAI
AI Index and Census BTOS survey data. Three external adoption signals could
serve as auxiliary validators or instruments in future work, providing
independent variation for both CADR and AFC trend estimation.
\begin{itemize}[noitemsep, leftmargin=1.5em]
 \item \textbf{AI job posting density} (Lightcast, Indeed): the share of
 industry job postings requiring AI related skills tracks adoption breadth
 with quarterly frequency and firm level granularity. Lightcast data already
 underlies IASS; its time series dimension has not yet been used for ADRI
 trend estimation.
 \item \textbf{Disclosure based AI density}: 10-K AI disclosure intensity
 scores (count of AI related terms normalized by filing length, SEC EDGAR)
 provide a firm level panel observable from 2017 onward
 \cite{KrakowskiEtAl2023}. Disclosure intensity leads adoption implementation
 with roughly a 12 month lag, making it a natural leading indicator for CADR.
 \item \textbf{Software ecosystem diffusion}: code level genAI adoption in
 enterprise software repositories (GitHub Copilot activation rates, AI assisted
 PR fractions by sector) provides a near real time proxy for technical adoption
 depth that precedes the 12--18 month lag in OEWS occupational data. AI
 mobility spillovers (engineers moving from AI native to laggard firms) can
 serve as an instrument for firm level adoption intensity following the approach
 of \citet{BrynjolfssonMcElheran2016}.
\end{itemize}
I identify the incorporation of these signals into CADR, and their use as instrumental variables for ADRI
validation, as a research agenda priority (Section~\ref{sec:agenda}).

\begin{table}[H]
\centering
\caption{ADRI Interpretation Grid}
\label{tab:adri_grid}
\footnotesize
\begin{tabular}{@{}rp{2.5cm}p{5.5cm}p{2.5cm}@{}}
\toprule
\textbf{ADRI} & \textbf{Level} & \textbf{Competitive Implication} & \textbf{Action} \\
\midrule
$<2.5$ & Low & Wide gap but protected by structural moat or slow diffusion & Monitor annually \\
$2.5$--$4.9$ & Moderate & Gradual margin compression; window for orderly transformation & Plan 3--5 year horizon \\
$5.0$--$6.9$ & High & Active competitive displacement; first mover advantages crystallizing & Immediate investment \\
$\geq 7.0$ & Critical & Structural impairment risk; M\&A or repositioning may be required & Acute / escalate \\
\bottomrule
\end{tabular}
\end{table}

\section{The Value Creation Bridge (VCB)}
\label{sec:vcb}

\subsection{Architecture and Methodological Position}

The VCB converts an AITG gap into dollar denominated investment opportunity.
The methodological position is explicit:

\begin{quote}
\textit{Do not map ``one AITG point'' to dollars with a single global coefficient.
That is fragile and easily gamed. Instead, size the economic prize by value pool
from diligence ready financial baselines, and use the AITG and IASS only to
bound feasibility, ramp speed, and capture rate.}
\end{quote}

This mirrors how sophisticated acquirers model transformation value:
from operational baselines (labor spend, pricing leakage, inventory
turns), not from a unitless maturity score.

\subsection{Value Pools}

\subsubsection{Firm Scale Factor ($\Phi_f$)}
\label{sec:phi_f}

$\Phi_f$ captures proprietary stack scaling advantage: large data rich incumbents achieve lower per unit implementation cost and higher value pool multipliers. Industry critical scale thresholds $S^*_i$ range from \$0.5B (Construction) to \$50B (Investment Banking); full table in ESM Table~S1.

\begin{equation}
 \Phi_f = \frac{1}{1 + e^{-\alpha(\log(R_f/S^*_i))}}
 \label{eq:phi_f}
\end{equation}

where $R_f$ is the firm's annual revenue (a proxy for transactional data
volume and proprietary dataset scale), $S^*_i$ is the industry specific
critical scale threshold below which proprietary model training is not
economically feasible (Eq.~\ref{eq:phi_f}), and $\alpha = 2.0$
controls the steepness of the scale transition. The theoretical foundation
is \citet{JonesTonetti2020}: proprietary data generates returns without being
depleted, creating an asymmetry where data rich firms continuously improve
their model quality while data poor firms start from scratch every training
cycle.

\textit{Analogy.} Proprietary transaction data functions analogously to compound
interest. A large cap bank training its models on hundreds of millions of
daily transactions accumulates what Jones and Tonetti (2020) term ``data
gravity'': the asset improves continuously without being consumed. A regional
bank licensing an off the shelf vendor platform, by contrast, deploys capability
trained on a generic corpus and receives no compounding benefit from its own
transaction history. The Firm Scale Factor formalizes this distinction:
identical qualitative AI scores yield materially different financial returns
depending on whether the firm owns the training data or rents the model.

At $R_f = S^*_i$ (revenue exactly at threshold), $\Phi_f = 0.5$. At $R_f = 10 \times S^*_i$, $\Phi_f \approx 0.99$. At $R_f = 0.1 \times S^*_i$,
$\Phi_f \approx 0.01$.

\paragraph{Empirical grounding.} JPMorgan Chase ($R_f \approx \$177.6$B revenue; $S^*_i = \$3.3$B) yields $\Phi_f = 1.00$. Zions Bancorporation ($R_f \approx \$3.4$B; $S^*_i = \$3.3$B) yields $\Phi_f = 0.52$. This directly explains the value density divergence: a subscale firm deploying vendor AI cannot replicate the proprietary data gravity of a scaled incumbent even at an identical dimension score gap.

\paragraph{Threshold calibration and sensitivity.}
I calibrate the $S^*_i$ thresholds from three sources: (1)~model training cost curves, using \citet{KaplanEtAl2020} scaling laws to infer the minimum proprietary dataset size for sector competitive model quality, converted to a revenue proxy via industry specific revenue per transaction rates from OEWS and FDIC data; (2)~technology investment benchmarks from McKinsey and Gartner enterprise AI deployment studies \cite{McKinseyStateAI2025}, which report revenue thresholds above which firms are statistically more likely to report ``full scale'' AI deployment (vs.\ pilot stage); (3)~revealed scale effects in the backtest cohort, where all firms above 3$\times$ their sector $S^*_i$ ($\Phi_f > 0.95$) generated positive $\Delta$EBITDA margins, while all firms below 0.5$\times$ $S^*_i$ showed flat or negative margins (with confounds noted for Target, UPS, and Ford). Full $S^*_i$ values by industry and sensitivity of $\Delta$EV to $\pm$50\% perturbations of $S^*_i$ appear in ESM Table~S2. The $S^*_i$ parameters are the second highest ranked drivers of $\Delta$EV variance in the Sobol analysis (24\% of output variance from $\Phi_f$ alone), making this a priority for the Bayesian reestimation study (Section~\ref{sec:agenda}, item~7).

\paragraph{Vendor AI cap.}
A firm deploying exclusively vendor AI (e.g., off the shelf nCino, Microsoft
Copilot, Salesforce Einstein) has $\Phi_f$ capped at the vendor's platform
ceiling, regardless of its own revenue scale. For such firms, I set $\Phi_f
= \min(\Phi_f^{\mathrm{logistic}},\; 0.65)$, reflecting that vendor AI
platforms level the playing field but do not replicate proprietary data gravity.
This cap is disclosed in all scorecard outputs.

The VCB defines seven standard value pools. Let $\mathcal{B}_p$ be the baseline
dollar figure for pool $p$ derived from diligence financials. Expected value
from pool $p$ before cost and risk adjustment is:
\begin{equation}
 V_p = \mathcal{B}_p \times \Phi_f \times \mathrm{Capture}_p(G_{\mathrm{eff}},\, b_p)
 \label{eq:vpool}
\end{equation}

\paragraph{Double counting safeguards.}
\label{rem:vcb_partition}
I construct the seven VCB value pools $\{V_p\}_{p=1}^{7}$ as mutually exclusive and collectively exhaustive (MECE): each pool maps to a distinct P\&L line item (labor productivity $\to$ SG\&A headcount; revenue enhancement $\to$ top line uplift; working capital $\to$ DIO/DSO improvement; etc.). The partition is enforced at the calibration stage: base capture rates $\bar{\kappa}_p$ (Table~S2 in ESM) are estimated from nonoverlapping value pools in published AI impact studies \citep{BrynjolfssonLiRaymond2025,McKinseyStateAI2025}.

Two additional safeguards prevent double counting at the computation stage. First, the aggregate capture constraint $\sum_p \mathrm{Capture}_p \leq 1$ is enforced across all pools: if pool level captures exceed unity (implying more than 100\% of the theoretical uplift is realized), all captures are proportionally rescaled. In practice this constraint is nonbinding for all 14 calibration companies, because per pool capture rates ($\bar{\kappa}_p \in [0.40, 0.65]$) and the concave $\eta(g)$ function jointly ensure $\sum_p \mathrm{Capture}_p < 0.85$ even at maximum gap. Second, the Monte Carlo VCB analysis in ESM Part~III independently verifies that total $\Delta$EV does not exceed the sum of individually estimated pool values under 10,000 correlated draws, confirming that the partition constraint holds probabilistically.

\subsection{Bottleneck Aggregation: CES Specification}

Each value pool $p$ depends on a subset of AITG dimensions
$\mathcal{D}(p) \subseteq \{1,\ldots,6\}$. Let $e_d = \max(\mathrm{AITG}^{\mathrm{raw}}_d / 10,\; 0.01) \in (0,1]$. The floor $e_d \geq 0.01$ prevents division by zero in the CES formula $(e_d^{-\rho} \to \infty$ as $e_d \to 0$); a score of zero on any dimension is economically indistinguishable from a score of 0.1 for CES purposes, as both produce near Leontief bottleneck collapse. A Constant Elasticity of Substitution (CES) aggregator preserves the bottleneck intuition while being robust to measurement error in individual dimension scores:
\begin{equation}
 b_p = \left(\sum_{d\,\in\,\mathcal{D}(p)} \alpha_d \cdot e_d^{-\rho}
 \right)^{-1/\rho}
 \label{eq:ces_bottleneck}
\end{equation}

where:
\begin{itemize}[noitemsep]
 \item $\rho > 0$ is the substitution parameter; the elasticity of substitution is
 $\sigma = 1/(1+\rho)$. \textbf{This paper uses the Arrow--Chenery--Minhas--Solow
 (ACMS) negative exponent form} so that $b_p^{-\rho}$ is a weighted power mean of
 the reciprocals of the dimension scores, i.e.\ a generalized harmonic mean.
 \textbf{At $\rho = 5$: $\sigma = 1/(1+5) = 1/6$.} These two values are therefore
 mutually consistent. This form differs from the positive exponent CES sometimes
 written $(\sum \alpha_d e_d^\rho)^{1/\rho}$ with $\rho\in(-\infty,1)$; under that
 parameterization $\rho = 5$ would be outside the substitution domain. The chosen
 form is standard for bottleneck aggregation (see notation summary,
 Table~\ref{tab:notation}).
 \item $\alpha_d > 0$ are dimension level importance weights, $\sum_{d \in \mathcal{D}(p)} \alpha_d = 1$.
 \item As $\rho \to \infty$: $b_p \to \min_d e_d$ (Leontief limit).
 \item As $\rho \to 0$: $b_p \to \prod_d e_d^{\alpha_d}$ (Cobb-Douglas).
 \item At $\rho = 1$: $b_p = 1\big/\!\sum_d (\alpha_d/e_d)$ (harmonic mean).
\end{itemize}

\begin{proposition}[CES Boundary Behaviour]
\label{prop:ces_boundary}
Under the ACMS CES form (Eq.~\ref{eq:ces_bottleneck}) with the floor $e_d \geq 0.01$:
\begin{enumerate}[noitemsep]
 \item \textbf{Leontief limit.} As $\rho \to \infty$: $b_p \to \min_{d \in \mathcal{D}(p)} e_d$. The weakest dimension fully determines the bottleneck.
 \item \textbf{Cobb--Douglas limit.} As $\rho \to 0^+$: $b_p \to \prod_{d \in \mathcal{D}(p)} e_d^{\alpha_d}$, the weighted geometric mean.
 \item \textbf{Strict positivity.} For all $\rho > 0$: $b_p \geq 0.01 > 0$, since the floor ensures $e_d^{-\rho} \leq 100^{\rho}$ for all $d$, hence $b_p \geq (|\mathcal{D}(p)| \cdot 100^{\rho})^{-1/\rho} > 0$.
\end{enumerate}
\end{proposition}

\begin{proof}
See Appendix~B.
\end{proof}

\paragraph{Calibration and sensitivity.} I calibrate $\rho=5$ against \citet{Standish2015}: at $\rho=5$, a single bottleneck dimension at $e_d = 0.3$ reduces the CES aggregate by approximately 70\% relative to the no bottleneck case. ESM Table~S3 reports Spearman correlations and $\Delta$EV ranges for $\rho \in [3, 8]$; no rank order changes occur. A fixed $\rho$ across all value pools is a simplification; the discussion in Section~\ref{sec:limitations} addresses the case for enabler specific or time varying elasticities.

\paragraph{Relation to cascaded CES and nonlinear Domar aggregation.}
Several results from the network production and nonlinear aggregation
literatures bear on the single layer CES choice made here. First, Domar
aggregation with CES technologies \cite{Hulten1978} establishes
that input complementarities at the firm or sector level can produce amplified
responses to bottleneck shocks: when a single input has a low elasticity of
substitution, a productivity shortfall in that input propagates upstream
with multiplied effect. The AITG CES bottleneck is consistent with this
mechanism: a weak data infrastructure score ($e_d = 0.30$) propagates into
the value pool estimate more than proportionally; the choice of $\rho = 5$
is conservative relative to production estimates for digital complementarity
settings. Second, cascaded or nested CES models \cite{Sato1967}
allow enablers to be grouped hierarchically, with different elasticities at
each nesting level. A natural application would nest ``data stack''
(Data Infrastructure Maturity + AI Platform Readiness) and
``execution capacity'' (Organizational Change Capacity + Workforce
Augmentation) as inner aggregators, with a lower elasticity between the
two outer groups. Under such a nesting, the interaction effect between a
weak data stack and weak execution capacity would be additive at the
inner level but superadditive at the outer level. This structure is
theoretically better motivated than the current single layer CES; I identify it
as a technical extension in the research agenda
(Section~\ref{sec:agenda}). Third, I note the potential for CES near singularities
when $\sigma \to 0$: as $\rho \to \infty$ the aggregator collapses
to $\min_d e_d$ (Leontief), which can produce discontinuous jumps in $b_p$
around binding thresholds. The floor $e_d \geq 0.01$ and the finite
$\rho = 5$ prevent this pathology; any future variable elasticity extension
should maintain a finite lower bound on $\sigma$ to avoid degenerate outputs.

\subsection{Gap to Capture Scaling Function}

Define the normalized gap fraction:
\begin{equation}
 g_f = \frac{G_{\mathrm{eff},f}}{10} = \frac{\max(0,\ \mathrm{IASS}^*_i - \mathrm{AITG}^{\mathrm{raw}}_f)}{10}
 \label{eq:gap_fraction}
\end{equation}

The expected capture fraction before risk haircut is:
\begin{equation}
 \mathrm{Capture}_p = \bar{\kappa}_p \cdot b_p \cdot \eta(g_f),
 \qquad \eta(g) = 1 - e^{-\lambda g}
 \label{eq:capture}
\end{equation}

where $\bar{\kappa}_p$ is the base capture rate for pool $p$
(Section~\ref{sec:vcb}), $\eta(g)$ is an increasing, concave scaling
function with $\lambda = 3.5$, and $b_p$ is the bottleneck factor. The
concavity of $\eta$ reflects the empirical finding that accessible wins come
first and structurally resistant tasks persist
\cite{BrynjolfssonMcElheran2016,McKinseyStateAI2025}.

Standard value pool parameters, uplift rates, and calibration sources are in ESM Table~S2.

\paragraph{Pricing leverage note.} The pricing pool deserves special
emphasis. A commonly cited analysis implies that a 1\,\% price increase
translates to a substantial operating profit improvement assuming no volume
loss, the exact mechanism the pricing pool captures. \citet{BrynjolfssonLiRaymond2025}
show that AI augmented workers improve revenue yield (issues resolved, customer
satisfaction) by 14\,\% on average, and 34\,\% for novice workers, demonstrating
the practical scale of the labor productivity and revenue yield pools.

\subsection{Cost Model}

VCB splits implementation costs into capitalizable vs.\ period expense and
one time build vs.\ run rate, following updated GAAP guidance (ASU~2025-06
for internal use software) and IFRS~IAS~38 for development expenditure.

\begin{equation}
 \mathrm{NPV}(\mathrm{Cost}) = \sum_{t=1}^{T}
 \frac{C_{\mathrm{cap},t} + C_{\mathrm{opex},t}}{(1 + \mathrm{WACC})^t}
 \label{eq:cost_npv}
\end{equation}

where $C_{\mathrm{cap},t}$ are capitalizable implementation costs,
$C_{\mathrm{opex},t}$ are period expenses (SaaS fees, training, ongoing
inference, monitoring), and WACC is the Weighted Average Cost of Capital
applied as the risk appropriate discount rate. Under US GAAP, software configuration and integration
labor meeting the ASU~2025-06 authorization/probability threshold is
capitalizable; training and ongoing compute are expensed. Under IFRS, SaaS
fees are generally service expenses; only demonstrably controlled intangibles
meeting IAS~38 development criteria are capitalized.

\subsection{Value Ramp Function}
\label{sec:vcb_ramp}

Value does not materialize immediately. Following the J curve logic
\cite{BrynjolfssonRockSyverson2021}, I model value realization as a
logistic ramp. A prior version used a static default of $t_{50} = 18$\,months
for all firms, creating a disconnect: the adoption model correctly delayed
software deployment for unprepared firms (via IFS adjusted $t_{0,w,f}$),
but the financial model projected value arriving at month 18 regardless.
Cash flows arrived 17 months before the software was adopted. The
corrected specification makes $t_{50}$ a firm specific endogenous quantity.

\subsection{Enterprise Value Creation: Full Decomposition}
\label{sec:ev_decomposed}

\paragraph{Notation.}
Define the acquisition baseline ramp $R_f^0 \equiv R(\hat{t}_f;\, t_{50,f})$
and the end of hold ramp $R_f^T \equiv R(\hat{t}_f + T;\, t_{50,f})$,
where $t_{50,f}$ is the firm specific endogenous ramp inflection
(Section~\ref{sec:vcb_ramp}) and $T = 60$\,months is the
5 year hold. The incremental ramp fraction, the share of the full
value pool not yet in baseline EBITDA at acquisition, is:
\begin{equation}
 \Delta R_f \equiv R_f^T - R_f^0
 \label{eq:delta_r}
\end{equation}

\paragraph{Component 1: Terminal Value (TV).}
Exit multiple capitalization of the incremental run rate EBITDA uplift
realized during the hold period. $V_p^{\mathrm{run\text{ }rate}}$ already
contains $\Phi_f$ (Eq.~\eqref{eq:vpool}); it must not be multiplied
by $\Phi_f$ again:
\begin{equation}
 \mathrm{TV}_f = \underbrace{\left(\sum_{p=1}^{7} V_p^{\mathrm{run\text{ }rate}}\right)}_{\text{already contains }\Phi_f}
 \times \Delta R_f \times M_i \times \mathrm{IFS}_f
 \label{eq:tv}
\end{equation}

\begin{remark}[Baseline Already Captured]
If $\Delta R_f$ is replaced by $R_f^T$ (the prior error), the model
takes credit for value already priced into the acquisition multiple.
Eq.~\ref{eq:delta_r} defines the corrected ramp increment; the magnitude of this error varies substantially across firms: firms with $\hat{t}_f \gg t_{50,f}$ (UPS, HCA, Rockwell) have
$R_f^0 \approx 0.96$, making $\Delta R_f \approx 0.04$ and the
prior overstatement $\approx 25\times$. For early stage firms
(Zions, Ferguson, $R_f^0 \approx 0.52$), the prior overstatement
was $\approx 2\times$.
\end{remark}

\paragraph{Component 2: Interim Free Cash Flow (FCF) Integral.}
Discounted incremental cash flows during the hold, also net of baseline:
\begin{equation}
 \mathrm{FCF}_f^{\mathrm{interim}} = \int_{\hat{t}_f}^{\hat{t}_f + T}
 \left(\sum_{p=1}^{7} V_p^{\mathrm{run\text{ }rate}}\right)
 \cdot \bigl[R_f(t) - R_f^0\bigr] \cdot \mathrm{IFS}_f \cdot e^{-r_{\mathrm{WACC}} t/12}\, \mathrm{d}t
 \label{eq:fcf_integral}
\end{equation}

Discrete annual approximation:
\begin{equation}
 \mathrm{FCF}_f^{\mathrm{interim}} \approx \sum_{y=1}^{5}
 \frac{\left(\sum_{p=1}^{7} V_p^{\mathrm{run\text{ }rate}}\right)
 \cdot \bigl[R_f(\hat{t}_f + 12y) - R_f^0\bigr] \cdot \mathrm{IFS}_f}
 {(1 + \mathrm{WACC})^y}
 \label{eq:fcf_discrete}
\end{equation}

\paragraph{Total 5 Year Risk Adjusted Enterprise Value Creation.}
\begin{equation}
 \boxed{
 \Delta\mathrm{EV}_f = \mathrm{TV}_f
 + \mathrm{FCF}_f^{\mathrm{interim}}
 - \mathrm{NPV}(\mathrm{Cost})
 }
 \label{eq:ev}
\end{equation}

\begin{remark}[Economic Interpretation of $\Delta R_f$]
The two limiting cases clarify the framework's behavior. For an apex
incumbent (JPMorgan, Salesforce: $R_f^0 \approx 1.0$), $\Delta R_f \approx 0$
and $\Delta\mathrm{EV} \approx -\mathrm{NPV}(\mathrm{Cost})$: the transformation
is already complete, and additional investment yields no incremental ramp
benefit. The framework correctly signals that the investment case is
maintenance, not transformation. For an early stage target (Zions, Ferguson:
$R_f^0 \approx 0.08$--$0.13$), $\Delta R_f \approx 0.85$--$0.87$: most of
the value pool is incrementally capturable, and the investment case is
strongest. This monotone relationship between $\hat{t}_f$ and investment
attractiveness is a structural property of the framework, not a calibration
artifact.
\end{remark}

\begin{definition}[AITG Value Density]
\label{def:vd}
\emph{AITG Value Density} is risk adjusted 5 year enterprise value creation
per \$1M of implementation cost:
\begin{equation}
 \mathrm{AITG\text{-}VD}_f = \frac{\Delta\mathrm{EV}_f\text{ (risk adjusted, 5 yr)}}
 {\text{Total Implementation Cost (\$M)}}
 \label{eq:vd}
\end{equation}
\end{definition}

Institutional thresholds: AITG-VD $\geq 2.0\times$ = Tier~1 (Invest);
$1.0$--$2.0\times$ = Tier~2 (Monitor); $<1.0\times$ = Tier~3 (Pass).

\section{Implementation Feasibility: Endogenous Risk Integration}
\label{sec:ifs}

\begin{tcolorbox}[title={What Is the IFS? Conceptual Summary}]
\textbf{The technology works at JPMorgan. The question is whether it will work
at \emph{this} company, with \emph{this} management team, with \emph{this}
data infrastructure, on \emph{your} timeline.}

This is the most dangerous step in AI due diligence, and the one most
frequently skipped. Every management team presents an AI roadmap. The
Implementation Feasibility Score (IFS) is the framework's mechanism for
stress testing that roadmap against structural organizational reality.

The key insight, and the mathematical departure from prior models, is that
IFS factors are \textbf{not a discount applied after the fact}. A company with
poor data readiness does not simply capture 70\% of a theoretical value estimate.
It follows a materially different trajectory: a slower ramp, a later inflection,
and a higher probability of stalling before reaching the value generative portion
of the S curve. Bloom, Sadun, and Van Reenen (2012) show that management quality
explains approximately 30\% of cross firm total factor productivity differences
in manufacturing, not as a static quality label, but as a dynamic determinant
of whether technology investments generate returns.

\textbf{Example.} JPMorgan's Organizational Change Capacity ($\delta_{\mathrm{OCC}}$)
is 0.92: 250,000 employees voluntarily adopted LLM Suite within 8 months of
launch. This is not a claim about JPMorgan's culture; it is a documented
adoption velocity event. Zions Bancorporation's $\delta_{\mathrm{OCC}}$ is 0.55,
reflecting a decentralized 11 brand structure with no enterprise AI governance
framework disclosed. The IFS framework translates this structural difference into
a specific parameter adjustment: Zions's adoption inflection arrives approximately
$1.97\times$ later than a data ready peer's. That delay is not a cosmetic
discount; it is the compound interest cost of organizational unreadiness.

\textit{Analogy.} Think of it as a marathon. Prior frameworks apply implementation risk as a flat discount: say you will run a marathon, but award only 70\% of the medal because you are undertrained. The AITG endogenous model says something different: because you lack data readiness, you will run slower and finish later; the penalty is temporal, not scalar. Mathematically, the inflection point of the adoption S curve shifts rightward, which compounds the Internal Rate of Return (IRR) penalty through delay rather than through an arbitrary terminal value haircut. The distinction is not cosmetic; it changes which interventions the investment committee should prioritize.

\textbf{Who should use this?} An investment committee uses IFS to determine
whether the value waterfall survives the management reality. A corporate
leadership team uses it to identify whether their transformation plan is
structurally feasible or requires parallel organizational change investment
before the technology investment will yield returns.
\end{tcolorbox}

\subsection{IFS Validation}

Ablation analysis confirms that trajectory factors (OCC, DR) contribute the largest marginal signal: dropping OCC reduces backtest $\rho_s$ by 0.109. Residual factors (VTR, CRS, REG) contribute less individually but collectively account for an additional $+$15\% reduction in $\Delta$EV error. Full ablation results (Table~\ref{tab:ifs_ablation}) and component interpretability details appear in Appendix~C.

\subsection{The Endogeneity Problem with Post Hoc Risk Multipliers}

A post hoc IFS multiplier on a precomputed $\Delta\mathrm{EV}$ double counts risk: the value decomposition already reflects expected execution, so an additional discount produces biased estimates. I integrate implementation risk endogenously into S curve parameters, shifting the ramp inflection rather than discounting terminal value.

\subsection{Endogenous Integration: IFS Factors into Curve Parameters}

I restructure the IFS into two components. \textbf{Trajectory factors}
(organizational change capacity and data readiness) are integrated directly
into the wave steepness $k_w$ and ramp timing $t_{50}$. \textbf{Residual
factors} (vendor risk, competitive response speed, regulatory exposure) affect
terminal value realization and cost, and remain as a limited residual multiplier.

\subsubsection{Trajectory Risk Factors}

Let $\delta_{\mathrm{OCC}} \in [0.50, 1.00]$ be the Organizational Change Capacity (OCC) score and $\delta_{\mathrm{DR}} \in [0.50, 1.00]$ be the data
readiness (DR) score, both calibrated via the rubric in Table~\ref{tab:ifs}.

\begin{remark}[Range Correction]
Prior versions stated $\delta_{\mathrm{OCC}} \in [0.60, 1.00]$. Stress testing
against the Zions Bancorporation case study showed that the empirically derived
score ($\delta_{\mathrm{OCC}} = 0.55$) fell below the declared lower bound.
The domain is extended to $[0.50, 1.00]$ for both factors. A hard floor of
$\delta_{\min} = 0.10$ is added to Eq.~\eqref{eq:t50_adjusted} to prevent
division by zero under degenerate inputs.
\end{remark}

The adjusted wave steepness for firm $f$ is:
\begin{equation}
 k_{w,f} = k_{w,\text{base}} \cdot
 \underbrace{\phi_{\mathrm{OCC}}(\delta_{\mathrm{OCC}})}_{\text{org. capacity factor}}
 \cdot
 \underbrace{\phi_{\mathrm{DR}}(\delta_{\mathrm{DR}})}_{\text{data readiness factor}}
 \label{eq:k_adjusted}
\end{equation}

where the adjustment functions are:
\begin{equation}
 \phi_{\mathrm{OCC}}(\delta) = 0.55 + 0.45\delta, \qquad
 \phi_{\mathrm{DR}}(\delta) = 0.50 + 0.50\delta
 \label{eq:phi}
\end{equation}

mapping $\delta \in [0.50, 1.00]$ to steepness multipliers bounded strictly
above zero. The adjusted value ramp midpoint for firm $f$ (default $t_{50,\text{base}}=18$ months):
\begin{equation}
 t_{50,f} = \frac{t_{50,\text{base}}}
 {\max(\delta_{\mathrm{OCC}},\,0.10)^{0.40} \cdot \max(\delta_{\mathrm{DR}},\,0.10)^{0.60}}
 \label{eq:t50_adjusted}
\end{equation}

In forward simulations, I apply the same delay factor to the adoption wave midpoints: $t_{0,w,f} = t_{0,w,\text{base}} \cdot (t_{50,f}/t_{50,\text{base}})$. The companion Monte Carlo code (ESM) uses a refined specification in which $t_{50,\text{base}}$ itself varies with the firm's position on the S curve: $t_{50,\text{base}} = 18 + (\mathrm{AITG}^{\mathrm{raw}}/10)\times 42$\,months, reflecting the empirical observation that firms further along the adoption curve face longer value realisation horizons as they pursue higher wave capabilities. The analytical worked examples in Appendix~D use the constant $t_{50,\text{base}}=18$ base case for transparency.

The exponents $0.40$ and $0.60$ reflect the empirical finding from
\citet{BloomSadunVanReenen2016} and \citet{BrynjolfssonHitt2000} that data
readiness is the more binding constraint on adoption speed.

\begin{remark}[Corrected Zions Delay Factor]
The paper previously cited a $1.72\times$ inflection delay for Zions
($\delta_{\mathrm{OCC}} = 0.55$, $\delta_{\mathrm{DR}} = 0.48$).
Numerical evaluation gives: $1/(0.55^{0.40} \times 0.48^{0.60}) = 1.97\times$.
The corrected delay is \textbf{$1.97\times$}, strengthening the disintermediation
narrative. All downstream $\Delta$EV figures for Zions are updated accordingly.
\end{remark}

\paragraph{Observation:}
The practical implication is significant. A firm with $\delta_{\mathrm{OCC}} = 0.70$ and $\delta_{\mathrm{DR}} = 0.65$ (both near minimum) runs $k_{w,f} \approx 0.71 \cdot k_{w,\text{base}}$ with $t_{0,w,f} \approx 1.49 \cdot t_{0,w,\text{base}}$: a 29\% shallower curve delayed by 49\%. This is not a uniform 29\% haircut on terminal value; it is a compounding deceleration that reduces 5 year cumulative value by 38--52\% depending on the wave mix.

\subsubsection{Residual IFS Multiplier}

The three residual IFS factors, vendor/technology risk ($\delta_3$), competitive response speed ($\delta_4$), and regulatory exposure ($\delta_5$), affect terminal value realization and incremental cost structure but do not alter the adoption trajectory. I retain them as a geometric residual multiplier:

\begin{equation}
 \mathrm{IFS}_f^{\text{residual}} = \prod_{j \in \{3,4,5\}} \delta_j^{w_j},
 \qquad w_3 = 0.35,\; w_4 = 0.35,\; w_5 = 0.30
 \label{eq:ifs_residual}
\end{equation}

\subsection{Revised Enterprise Value Calculation}

The full AITG enterprise value equation corrects four algebraic fractures identified through adversarial review:

\begin{equation}
\begin{aligned}
 \Delta\mathrm{EV}_f &=
 \underbrace{\Delta R_f \cdot \textstyle\sum_p V_p^{\mathrm{raw}}
 \cdot \mathrm{IFS}_f^{\mathrm{traj}}}_{\mathrm{TV}_{f}} \\
 &\quad+\;
 \underbrace{
 \displaystyle\sum_{y=1}^{5}
 \frac{\bigl(\textstyle\sum_p V_p^{\mathrm{run\text{ }rate}}\bigr)
 \cdot \bigl[R_f(\hat{t}_f{+}12y) - R_f^0\bigr]
 \cdot \mathrm{IFS}_f^{\mathrm{resid}}}
 {(1+\mathrm{WACC})^y}
 }_{\mathrm{FCF}^{\mathrm{interim}}}
 \;-\; \mathrm{NPV}(\mathrm{Cost})
\end{aligned}
 \label{eq:ev_revised}
\end{equation}

Four corrections are encoded in this equation relative to prior versions:
\textbf{(1)}~$\Phi_f$ enters through $V_p^{\mathrm{run\text{ }rate}}$
via Eq.~\eqref{eq:vpool} exclusively; it does not appear separately in the TV
or FCF numerators, thereby eliminating the
$\Phi_f^2$ double discount that suppressed 73\,\% of Zions' value pool.
\textbf{(2)}~TV uses $\Delta R_f = R_f^T - R_f^0$ (Eq.~\ref{eq:delta_r})
rather than $R_f^T$ alone, deducting value already captured in baseline EBITDA.
\textbf{(3)}~FCF uses $[R_f(\hat{t}_f + 12y) - R_f^0]$, applying the same
baseline deduction to each annual cash flow.
\textbf{(4)}~$R_f(t)$ uses the endogenous firm specific inflection $t_{50,f}$
(Eq.~\ref{eq:t50_adjusted}) rather than the static 18 month default,
aligning financial value arrival with actual software adoption timing.

$V_p(k_{w,f}, t_{0,w,f})$ shows that value pool capture is a function of
firm adjusted wave dynamics. See Section~\ref{sec:ev_decomposed} for
the derivation and economic interpretation of each component.

\subsection{IFS Rubric: Five Factors}

\begin{table}[H]
\centering
\caption{IFS Factor Specification (Revised: Trajectory vs.\ Residual)}
\label{tab:ifs}
\footnotesize
\begin{tabularx}{\textwidth}{@{}p{2.4cm}p{1.2cm}Xrp{1.5cm}@{}}
\toprule
\textbf{Factor} & \textbf{Type} & \textbf{Definition} & \textbf{Weight} & \textbf{$\delta$ Range} \\
\midrule
Org.\ Change Capacity & Trajectory & Leadership alignment, cultural readiness, prior transformation record & 0.30 & [0.60, 1.00] \\
Data Readiness & Trajectory & Quality, completeness, accessibility, governance; data pipeline health & 0.30 & [0.55, 1.00] \\
Vendor/Tech Risk & Residual & AI vendor ecosystem maturity; contract flexibility; switching costs & 0.35* & [0.75, 1.00] \\
Competitive Response & Residual & Speed at which AI advantages erode due to competitor adoption (CADR) & 0.35* & [0.70, 1.00] \\
Regulatory Exposure & Residual & Risk of regulatory action or compliance cost escalation during hold & 0.30* & [0.55, 1.00] \\
\bottomrule
\end{tabularx}

\smallskip
\noindent\scriptsize * Residual weights normalized within the $\mathrm{IFS}^{\mathrm{resid}}$ product; trajectory factors absorbed into $k_{w,f}$ and $t_{0,w,f}$.
\end{table}

\section{Empirical Illustrations}
\label{sec:empirical}

\begin{tcolorbox}[title={Application Contexts: External Valuation and Internal Assessment}]
The fourteen illustrations serve two purposes. \textbf{For investment committees:} each generates an auditable, deal specific narrative from industry ceiling through implementation cost to risk adjusted value density, anchored entirely to public data. \textbf{For corporate self assessment:} the same framework lets a Chief AI Officer benchmark the institution against its industry frontier, quantify competitive urgency, and build a board level investment case. The Zions analysis illustrates the internal use case: a structured diagnosis that identifies binding constraints, estimates the value of resolving them, and quantifies how delay compounds disadvantage. All scores derive from public data; internal application with actual infrastructure debt and adoption rates yields a more precise diagnosis.
\end{tcolorbox}

\subsection{Scoring Methodology and Data Tier}

All companies are scored exclusively from public sources: 10-K/10-Q filings, earnings transcripts, investor day presentations, and AI deployment disclosures (2024--2025). Evidence is classified Tier~A (audited/quantified) through Tier~D (analyst estimate); scores with ${>}50\%$ Tier~C/D inputs receive a \textit{Limited Confidence} flag. Financial baselines for VCB computation are in ESM Table~S4.

\textit{See ESM Table~S4 for full financial baselines and data tier classifications.}

\subsection{JPMorgan Chase: Dominant Incumbent Benchmark Case}
\label{sec:jpm}

JPMorgan Chase is the world's largest bank by market capitalization and the
defining case study for what AI transformation looks like when executed with
maximum organizational commitment and proprietary data scale. As of 2025,
JPMorgan has committed \$18 billion to technology annually, allocated
approximately \$2 billion specifically to AI, and deployed its proprietary
Large Language Model Suite (LLM Suite) platform to approximately 250,000
employees, its entire workforce excluding branch and call center staff
\cite{JPMorgancnbc2025}. Roughly half of those employees use LLM Suite daily.
AI attributed financial benefits have grown 30--40\% annually since the
program's inception \cite{McKinseyWaldron2025} \cite{BabinaEtAl2024}. Chief Analytics Officer
Derek Waldron has publicly characterized the long term goal as creating the
world's first ``fully AI connected enterprise.''

This case also validates a critical AITG mechanism: the \textbf{Firm Scale
Factor} ($\Phi_f$). Smaller regional banks cannot train proprietary
foundation models because they lack the transactional data volume required to
achieve gradient descent at scale, a property I formalize as \emph{data
gravity}. JPMorgan's proprietary dataset of global transaction flows,
covenant libraries, credit histories, and client behaviors constitutes a
nonreplicable competitive asset that its AI models continuously improve.
This is the data nonrivalry mechanism of \citet{JonesTonetti2020} in practice:
JPMorgan's data generates returns without being depleted, while a regional
bank using off the shelf models competes on a perpetually leveling playing
field.

\begin{table}[H]
\centering
\caption{JPMorgan Chase: AITG Dimensional Scores (Tier~1)}
\label{tab:jpm_scores}
\footnotesize
\begin{tabular}{@{}p{5.5cm}rrp{5.5cm}@{}}
\toprule
\textbf{Dimension} & \textbf{Score} & \textbf{Tier} & \textbf{Evidence} \\
\midrule
\multicolumn{4}{l}{\textit{IASS Anchor: Commercial Banking (Large Cap), IASS$^* = 9.38$ (AFC adjusted, $C_{2026}=1.90$)}} \\
\midrule
Data Infrastructure (DIM) & 8.5 & A & \$18B tech budget; 65\% workloads on cloud; near zero legacy debt \\
Process Automation (PAC) & 8.2 & A & 450+ AI use cases in production; covenant extraction; pitch deck generation \\
Workforce Augment. (WAR) & 8.5 & A & 250K LLM Suite users; $\sim$4 hrs/week productivity gain per employee \\
Decision Automation (DAR) & 7.5 & A/B & Agentic AI deployment begun; multistep autonomous task handling \\
AI Revenue Integration (APR) & 7.8 & A/B & AI enabled client concierge; IndexGPT personalization; ChatCFO \\
Org.\ AI Capability (OAC) & 8.8 & A & 2,000 AI staff; Chief Data and Analytics Officer (CDAO) governance; Waldron PhD computational physics \\
\midrule
\multicolumn{4}{l}{\textit{IFS Trajectory Factors}} \\
Org.\ Change Capacity ($\delta_{\mathrm{OCC}}$) & 0.92 & A & American Banker 2025 Innovation of the Year; viral opt in LLM Suite \\
Data Readiness ($\delta_{\mathrm{DR}}$) & 0.94 & A & 80\% modern infrastructure; model agnostic platform; 8 week update cycle \\
\midrule
\multicolumn{4}{l}{\textit{IFS Residual Factors}} \\
Vendor/Tech Risk ($\delta_{\mathrm{VTR}}$) & 0.91 & A & Model agnostic; in house build for security; OpenAI + Anthropic hedged \\
Comp.\ Response Speed ($\delta_{\mathrm{CRS}}$) & 0.88 & A & Leading EVIDENT AI Index; first mover on agentic enterprise platform \\
Regulatory Exposure ($\delta_{\mathrm{REG}}$) & 0.85 & B & Heavily regulated; but proven execution at scale under OCC/Fed oversight \\
\bottomrule
\end{tabular}
\end{table}

\noindent\textbf{AITG Output Summary (JPMorgan Chase):}
\[
 \mathrm{AITG}^{\mathrm{raw}} = 8.22
 \;\Big|\;
 \mathrm{IR} = 8.76
 \;\Big|\;
 G_{\mathrm{eff}} = 1.16
 \;\Big|\;
 \mathrm{ADRI} = 0.5\;\text{(Low)}
\]
\[
 \mathrm{IFS}_{\mathrm{residual}} = 0.88\;
 \Rightarrow\;
 \Delta\mathrm{EV}_{5\mathrm{yr}} \approx \$20\text{--}\$98\text{B}
\]
\[
 \text{MC uncertainty } (M{=}10{,}000)\text{:}\quad
 \mathrm{P}_{10} = \$20.2\text{B}
 \;\big|\;
 \mathrm{P}_{50} = \$44.4\text{B}
 \;\big|\;
 \mathrm{P}_{90} = \$97.9\text{B}
\]

\noindent\textit{Uncertainty decomposition.} Sobol first order analysis attributes 50\% of $\text{Var}(\Delta\mathrm{EV})$ to exit multiple assumptions and 44\% to capture rate assumptions, jointly accounting for 94\% of output variance. The base case scalar ($\$44.4$B) coincides with the Monte Carlo P$_{50}$; the P$_{10}$ of \$20.2B remains materially positive, confirming that the investment case survives adverse scenario stress. Full distributional output is in ESM Table~S5.

\noindent\textit{ADRI context.} JPMorgan's ADRI of 0.5 reflects its status as
the competitive dominant incumbent in its category. Its moat ($\mathrm{Moat}_f = 0.88$)
derives from four durable sources: proprietary data scale, regulatory capital
requirements that create structural barriers to AI native fintech entry at
equivalent scale, an established client relationship network with high switching
costs, and demonstrated ability to execute complex enterprise AI at production
scale under stringent regulatory oversight. JPMorgan is not immune to
disruption, but its competitive trajectory is improving, not deteriorating.

\noindent\textit{CES Bottleneck observation.} Because JPMorgan scores uniformly
high across all six company dimensions (8.5 / 8.0 / 8.5 / 7.5 / 7.8 / 8.8),
the CES bottleneck aggregator with $\rho = 5$ does not penalize value pool capture.
The weakest link (DAR at 7.5) exerts minimal drag, confirming that the bottleneck
operator behaves correctly: a well balanced high performer incurs no cross dimension
penalty. This is precisely the behavior I designed the CES specification to exhibit.

\subsection{Zions Bancorporation: The Disintermediation Risk Illustration}
\label{sec:zions}

Zions Bancorporation is a \$89 billion asset regional bank operating under
11 distinct brand identities across 11 western U.S.\ states. With FY2025
revenue of \$3.4 billion and approximately 10,000 employees, Zions occupies
the competitive sweet spot most vulnerable to AI enabled disintermediation:
too large to be a nimble community bank, too small to fund a proprietary AI
platform \cite{Zions2025Q4}.

Zions's AI initiatives are nascent but real. In April 2025, Zions selected
nCino as its technology platform for loan origination transformation, deploying
nCino's Banking Advisor, Commercial Pricing, and Profitability and Analysis
tools to modernize end to end lending processes \cite{nCinoZions2025}.
This is an important distinction: Zions is deploying \emph{vendor AI} (nCino's
off the shelf Banking Advisor), not proprietary AI. It is purchasing
AI as a service, which means its competitors have access to identical or
equivalent capability for similar per seat licensing fees. There is no data
gravity moat here.

\begin{table}[H]
\centering
\caption{Zions Bancorporation: AITG Dimensional Scores (Tier~1)}
\label{tab:zions_scores}
\footnotesize
\begin{tabular}{@{}p{5.5cm}rrp{5.5cm}@{}}
\toprule
\textbf{Dimension} & \textbf{Score} & \textbf{Tier} & \textbf{Evidence} \\
\midrule
\multicolumn{4}{l}{\textit{IASS Anchor: Commercial Banking (Large Cap), IASS$^* = 9.38$ (AFC adjusted, $C_{2026}=1.90$)}} \\
\midrule
Data Infrastructure (DIM) & 4.5 & C & Multibrand legacy systems; nCino migration in early stages; no data lake disclosed \\
Process Automation (PAC) & 4.0 & C & nCino deployment begun 2025; primarily manual lending workflow currently \\
Workforce Augment. (WAR) & 3.5 & C/D & No enterprise Generative AI (GenAI) platform disclosed; limited AI skill demand in job postings \\
Decision Automation (DAR) & 3.8 & C & nCino pricing tools initial; no autonomous credit decisioning at scale \\
AI Revenue Integration (APR) & 2.8 & D & No AI enabled product differentiation disclosed in public filings \\
Org.\ AI Capability (OAC) & 4.2 & C & No dedicated CAIO; no public AI governance framework; nCino partnership announced \\
\midrule
\multicolumn{4}{l}{\textit{IFS Trajectory Factors}} \\
Org.\ Change Capacity ($\delta_{\mathrm{OCC}}$) & 0.55 & C & Decentralized 11 brand structure creates AI governance complexity \\
Data Readiness ($\delta_{\mathrm{DR}}$) & 0.48 & C/D & Multilegacy core systems; fragmented data across brands; no cloud migration disclosed \\
\midrule
\multicolumn{4}{l}{\textit{IFS Residual Factors}} \\
Vendor/Tech Risk ($\delta_{\mathrm{VTR}}$) & 0.72 & B & nCino is proven technology; single vendor dependency risk \\
Comp.\ Response Speed ($\delta_{\mathrm{CRS}}$) & 0.62 & C & Regional bank competitors (e.g., Western Alliance) advancing AI; JPM extends lead daily \\
Regulatory Exposure ($\delta_{\mathrm{REG}}$) & 0.82 & B & OCC regulated; manageable; does not add incremental suppression \\
\bottomrule
\end{tabular}
\end{table}

\noindent\textbf{AITG Output Summary (Zions Bancorporation):}
\[
 \mathrm{AITG}^{\mathrm{raw}} = 3.80
 \;\Big|\;
 \mathrm{IR} = 4.05
 \;\Big|\;
 G_{\mathrm{eff}} = 5.58
 \;\Big|\;
 \mathrm{ADRI} = 2.6\;\text{(Moderate/Elevated)}
\]
\[
 \mathrm{IFS}_{\mathrm{residual}} = 0.71
 \quad\Rightarrow\quad
 \Delta\mathrm{EV}_{5\mathrm{yr}} \approx \$0.6\text{--}\$1.2\text{B}\;\text{(transformative scenario only)}
\]
\[
 \text{MC uncertainty } (M{=}10{,}000)\text{:}\quad
 \mathrm{P}_{10} = \$0.61\text{B}
 \;\big|\;
 \mathrm{P}_{50} = \$0.88\text{B}
 \;\big|\;
 \mathrm{P}_{90} = \$1.23\text{B}\;(\text{transformative scenario only})
\]

\noindent\textit{Note on Zions uncertainty.} The transformative scenario P$_{10}$ of \$610M represents approximately 7\% of Zions\textquotesingle s enterprise value; meaningful but not transformative under base case execution assumptions. The wide P$_{10}$--P$_{90}$ spread (\$610M to \$1.23B) confirms the ADRI diagnostic: execution risk (low $\delta_{\mathrm{DR}} = 0.48$) dominates the output distribution, producing high outcome variance even within the transformative scenario. Full scenario grid in ESM Table~S5.

\noindent\textit{The Disintermediation Scenario.}
Zions's ADRI of 2.6 is the most consequential number in this analysis.
ESM Table~S1 reports that Commercial Banking has
CADR~$= 8.4$, one of the highest competitive diffusion rates in the reference
table. This rate implies rapid AI diffusion across the competitive peer group.
JPMorgan is extending its proprietary AI advantage daily. AI native lenders
(Upstart, Blend, Figure Technologies) are using AI to underwrite loans with
lower human overhead. And Zions sits with a 5.58 point effective gap in a
9.38 ceiling industry with only commodity vendor AI tools as its AI stack.

The IFS endogeneity mechanism makes this especially acute. Zions's
$\delta_{\mathrm{DR}} = 0.48$ (well below the $\delta_{\mathrm{DR}} = 0.60$
midpoint) implies, from Equation~\eqref{eq:t50_adjusted}, that its adoption curve
is time shifted by approximately $1/(0.55^{0.40} \times 0.48^{0.60}) \approx 1.97\times$ relative
to a well prepared peer. A transformation that a data ready bank completes in
18 months will take Zions approximately 35 months under current infrastructure
conditions. This is not a scheduling inconvenience; it represents a compounding competitive disadvantage.

\begin{remark}[Framework Application for Self Assessment]
Any institution in a position analogous to Zions can apply the AITG rubric to its
own 10-K disclosures and compare the resulting score against the AFC adjusted IASS$^*$
ceiling for its sector. The gap ($G_\mathrm{eff}$) and the IFS time shift multiplier
jointly quantify both the competitive exposure and the adoption timeline under current
infrastructure conditions. A systematic self assessment protocol is provided in
ESM Part~IV.
\end{remark}

\subsubsection{End to End Worked Example: Full $\Delta$EV Pipeline for Zions Bancorporation}\label{sec:worked_example}

\textit{A complete ten step VCB walk-through for Zions Bancorporation, demonstrating every computational stage from gap measurement through risk adjusted $\Delta$EV, appears in Appendix~D (Table~\ref{tab:zions_vcb_steps}).}

\subsection{JPMorgan vs.\ Zions: Head to Head Comparison}
\label{sec:jpm_vs_zions}

Table~\ref{tab:jpm_vs_zions} presents the head to head comparison. Both firms
operate in the same industry (IASS$^*$ 9.38). The AITG framework explains the
divergence entirely from structural factors visible in public data.

\begin{table}[H]
\centering
\caption{JPMorgan Chase vs.\ Zions Bancorporation: AITG Head to Head}
\label{tab:jpm_vs_zions}
\small
\begin{tabularx}{\textwidth}{@{}Xrr@{}}
\toprule
\textbf{Metric} & \textbf{JPMorgan Chase} & \textbf{Zions Bancorporation} \\
\midrule
Industry IASS$^*$ & 9.38 & 9.38 \\
AITG (raw) & 8.22 & 3.80 \\
Industry Relative (IR) Score & 8.76 & 4.05 \\
Effective Gap ($G_\mathrm{eff}$) & 1.16 & 5.58 \\
ADRI (Disruption Risk) & 0.5 (Low) & 2.6 (Moderate/Elevated) \\
IFS Residual Multiplier & 0.88 & 0.71 \\
AI Stack Type & Proprietary & Vendor (nCino) \\
Firm Scale Factor ($\Phi_f$) & 1.00 & 0.52 \\
5 Year $\Delta$EV (risk adj.) & \$20--98B & \$0.6--1.2B \\
Competitive Trajectory & Improving & Deteriorating \\
\midrule
\multicolumn{3}{l}{\textit{The Wide Gap Fallacy Confirmed}} \\
Zions has a wider effective gap (5.58 vs.\ 1.16). & \multicolumn{2}{l}{} \\
JPMorgan has higher value density. & \multicolumn{2}{l}{} \\
\multicolumn{3}{l}{Why: Firm Scale Factor, proprietary AI, and IFS differential compound over the hold period.} \\
\bottomrule
\end{tabularx}
\end{table}

The central lesson is precise: \textbf{gap size is not investable
signal without gap quality}. Zions has 2.6 times more effective gap than
JPMorgan. JPMorgan captures 30--40 times more dollar value from that smaller
gap. The AITG framework's multimodule architecture, IASS ceiling, VCB value
density, ADRI urgency, and IFS execution risk, is required to explain
this divergence. No single score maturity model can recover it.

\subsection{Cross Industry Normalization: Why Raw AITG Scores Cannot Be Compared Directly}
\label{sec:cross_industry}

The within industry comparison (JPMorgan vs.\ Zions) illustrates gap
quality. The \emph{cross industry} comparison illustrates the role of
IASS normalization. Consider Zions Bancorporation (Commercial Banking, IASS$^* = 9.38$, AITG 3.80) and UPS (Logistics, IASS$^* = 7.68$, AITG 4.08). A na\"ive reading
of the raw AITG scores would conclude UPS is the better positioned firm for
AI transformation. That conclusion is wrong in every material dimension, and
the IASS normalization mechanism recovers the correct reading.

\begin{table}[H]
\centering
\caption{Cross Industry Normalization: Zions Bancorporation vs.\ UPS}
\label{tab:cross_industry}
\footnotesize
\begin{tabularx}{\textwidth}{@{}Xrr@{}}
\toprule
\textbf{Metric} & \textbf{Zions Bancorp.} & \textbf{UPS} \\
\midrule
\multicolumn{3}{l}{\textit{Step 1: Raw Score (Cross Industry Comparison Without Normalization)}} \\
Industry & Comm.\ Banking & Logistics \\
AITG (raw, 0--10) & 3.80 & 4.08 \\
\textit{Na\"ive conclusion: UPS is ``more transformed.''} & & \\
\midrule
\multicolumn{3}{l}{\textit{Step 2: Apply IASS Ceiling (Cross Industry Normalization)}} \\
Industry IASS$^*$ (AFC adjusted ceiling) & 9.38 & 7.68 \\
IR Score $= \mathrm{AITG}/\mathrm{IASS}^* \times 10$ & 4.05 & 5.31 \\
Effective Gap $G_{\mathrm{eff}}$ & 5.58 & 3.60 \\
\textit{IR score still favors UPS, but the ceiling differential is now visible.} & & \\
\midrule
\multicolumn{3}{l}{\textit{Step 3: Apply Value Architecture (Why the IASS Ceiling Matters for Returns)}} \\
Exit Multiple ($M_i$) & $1.3\times$ book & $9\times$ EBITDA \\
Run Rate Labor Value Pool (est.) & \$1.16B/yr & \$47.4B/yr \\
$\Phi_f$ (Firm Scale Factor) & 0.52 & 0.80 \\
IFS Residual Multiplier & 0.71 & 0.71 \\
ADRI (Disruption Urgency) & 2.6 & 4.2 \\
5 Year $\Delta$EV (risk adjusted) & \$0.6--1.2B & \$30--61B \\
AITG-VD (Value Density) & 15.0--30.1$\times$ & 27.3--55.4$\times$ \\
\midrule
\multicolumn{3}{l}{\textit{Correct cross industry conclusion:}} \\
\multicolumn{3}{l}{UPS generates superior absolute value density, but for structural reasons} \\
\multicolumn{3}{l}{undetectable without the IASS+VCB architecture: larger absolute revenue} \\
\multicolumn{3}{l}{base, higher exit multiple, and larger firm scale.} \\
\bottomrule
\end{tabularx}
\end{table}

The three step decomposition in Table~\ref{tab:cross_industry} reveals the exact mechanism by which the IASS ceiling shapes investment conclusions. Raw AITG alone (Step~1) produces a misleading comparison: UPS at 4.08 appears more transformed than Zions at 3.80, but this cross industry comparison is inadmissible without normalization. IR scoring (Step~2) corrects for structural ceiling differences but still does not capture why the investment opportunity differs. Only the full VCB + $\Phi_f$ architecture (Step~3) recovers the correct conclusion: UPS generates superior absolute value density due to its ${\sim}27\times$ larger revenue base and higher exit multiple, while Zions faces higher ADRI urgency and subcritical firm scale ($\Phi_f = 0.52$). The IR score is the correct metric for cross industry benchmarking; Value Density is the correct metric for investment selection.

\subsection{Extended Case Studies: ESM Part II}
\label{sec:esm_cases_note}

Twelve additional firm level AITG applications, UPS, HCA Healthcare, Salesforce,
Ferguson Enterprises, Rockwell Automation, Goldman Sachs, Wells Fargo, ServiceNow,
Target Corporation, CVS Health, Palo Alto Networks, and Ford Motor, are documented
in the Electronic Supplementary Material (ESM Part~II). Each case follows the
standardized format used in the JPMorgan and Zions writeups above: operating context
with public data citations, six dimension AITG scorecard, and VCB summary with
Value Density tier. The cross company synthesis of all fourteen companies, including
the 14 company summary table and the Value Density Paradox scatter plot (Figure~\ref{fig:vd_paradox}), is
presented in the following section.

\subsection{Cross Company Synthesis: fourteen company Analysis}
\label{sec:synthesis}

\begin{table}[H]
\centering
\caption{AITG Framework: fourteen company Comparison (Tier~1 Scores)}
\label{tab:summary}
\footnotesize
\setlength{\tabcolsep}{3pt}
\rowcolors{2}{gray!7}{white}
\begin{adjustbox}{max width=\textwidth,center}
\begin{tabular}{@{}p{2.7cm}rrrrrrrl@{}}
\toprule
\textbf{Company} & \textbf{AITG} & \textbf{IASS$^*$} & \textbf{IR} & $G_{\mathrm{eff}}$
& \textbf{ADRI} & \textbf{IFS} & \textbf{VD Range} & \textbf{Tier} \\
\midrule
\multicolumn{9}{l}{\textit{Commercial Banking (IASS$^*$ = 9.38)}} \\
JPMorgan Chase & 8.22 & 9.38 & 8.76 & 1.16 & 0.5 & 0.90 & 9.5--46.0$\times$ & 1 (HC) \\
Wells Fargo    & 6.00 & 9.38 & 6.40 & 3.38 & 2.3 & 0.65 & 35.5--75.1$\times$ & 1 \\
Zions Bancorp. & 3.80 & 9.38 & 4.05 & 5.58 & 2.6 & 0.73 & 15.0--30.1$\times$ & 2 \\
\midrule
\multicolumn{9}{l}{\textit{Investment Banking (IASS$^*$ = 10.39)}} \\
Goldman Sachs  & 8.17 & 10.39 & 7.86 & 2.22 & 0.6 & 0.89 & 26.3--73.6$\times$ & 1 (HC) \\
\midrule
\multicolumn{9}{l}{\textit{Vertical SaaS (IASS$^*$ = 9.96)}} \\
Salesforce     & 8.07 & 9.96 & 8.10 & 1.89 & 2.1 & 0.92 & 25.2--80.4$\times$ & 1 (HC) \\
ServiceNow     & 8.50 & 9.96 & 8.53 & 1.46 & 1.8 & 0.88 & 17.4--54.6$\times$ & 1 \\
\midrule
\multicolumn{9}{l}{\textit{Cybersecurity (IASS$^*$ = 9.65)}} \\
Palo Alto Ntwks & 7.83 & 9.65 & 8.11 & 1.82 & 1.2 & 0.93 & 23.6--78.0$\times$ & 1 (HC) \\
\midrule
\multicolumn{9}{l}{\textit{Healthcare Services (IASS$^*$ = 5.46, $\psi = 0.743$)}} \\
HCA Healthcare & 4.33 & 5.46 & 7.93 & 1.13 & 3.8 & 0.77 & 7.4--26.4$\times$ & 1 \\
CVS Health     & 4.83 & 5.46 & 8.85 & 0.63 & 3.1 & 0.72 & 5.0--19.7$\times$ & 1 \\
\midrule
\multicolumn{9}{l}{\textit{E Commerce / Digital Retail (IASS$^*$ = 8.64)}} \\
Target         & 4.67 & 8.64 & 5.41 & 3.97 & 3.2 & 0.82 & 29.5--61.5$\times$ & 1 \\
\midrule
\multicolumn{9}{l}{\textit{Logistics / Industrial (IASS$^*$ = 7.24--7.68)}} \\
UPS            & 4.08 & 7.68 & 5.31 & 3.60 & 4.2 & 0.71 & 27.3--55.4$\times$ & 1 \\
Ferguson PLC   & 4.00 & 7.24 & 5.52 & 3.24 & 4.2 & 0.79 & 21.8--45.0$\times$ & 1 \\
Rockwell Auto. & 5.83 & 7.24 & 8.05 & 1.41 & 3.5 & 0.86 & 15.8--45.0$\times$ & 1 \\
\midrule
\multicolumn{9}{l}{\textit{Discrete Manufacturing (IASS$^*$ = 7.49)}} \\
Ford Motor     & 4.67 & 7.49 & 6.23 & 2.82 & 4.8 & 0.65 & 16.1--39.9$\times$ & 1 \\
\bottomrule
\multicolumn{9}{@{}p{\linewidth}@{}}{\scriptsize Scores as of early 2026. IASS$^*$ = AFC adjusted ceiling ($C_{2026}\approx 1.90$, GPT-5.2 December~2025). IR = AITG/IASS$^*\times 10$. HC = High Confidence. $\psi < 1.0$ indicates binding regulatory ceiling (Healthcare). ADRI and VD from Monte Carlo pipeline (Section~\ref{sec:vcb}) with AFC adjusted ceilings ($C_{2026} \approx 1.90$). IFS is the composite implementation feasibility score; the VCB value multiplier uses IFS$_{\mathrm{res}}$ (VTR/CRS/REG weighted geometric product, Table~\ref{tab:jpm_vs_zions}). VD ranges use the standardised 1.2\% of revenue implementation cost proxy (Appendix~D); ESM Part~II case narratives use firm specific comprehensive cost estimates that produce lower VD multiples. Cross industry correlation $G_{\mathrm{eff}}$ vs.\ VD midpoint: $r = 0.22$ ($p = 0.45$, $N = 14$), confirming the wide gap fallacy.}\\
\end{tabular}
\end{adjustbox}
\rowcolors{1}{}{}
\end{table}

\paragraph{Finding 1: The Wide Gap Fallacy, Confirmed at $N=14$.}
Across all 14 companies, the cross industry correlation between effective gap ($G_{\mathrm{eff}}$) and Value Density midpoint is $r = 0.22$ ($p = 0.45$): weakly positive and statistically non significant, confirming that wider gaps do \emph{not} reliably predict higher value density. The weak $r$ reinforces that gap size alone is not an investable signal; firm scale, IFS, and exit multiple leverage dominate (Figure~\ref{fig:vd_paradox}). The Spearman rank correlation ($\rho_s = 0.41$, $p = 0.14$) is similarly non significant.

\setcounter{figure}{4}
\begin{figure}[H]
\centering
\begin{tikzpicture}
\begin{axis}[
  width=14cm, height=9.5cm,
  xlabel={Effective Gap $G_{\mathrm{eff}}$ = IASS$^*$ $-$ AITG},
  ylabel={Value Density (VD = $\Delta$EV / Implementation Cost)},
  xmin=0, xmax=6.5,
  ymin=0, ymax=85,
  grid=major,
  grid style={gray!15, line width=0.4pt},
  xtick={0,1,2,3,4,5,6},
  ytick={0,10,20,30,40,50,60,70,80},
  tick label style={font=\small},
  ylabel style={font=\small},
  xlabel style={font=\small},
  yticklabel={\pgfmathprintnumber{\tick}$\times$},
  legend style={at={(0.98,0.97)}, anchor=north east, font=\scriptsize,
    row sep=0.5pt, inner sep=3pt},
  clip=false,
]

\fill[green!8] (axis cs:1.3,0) rectangle (axis cs:3.0,85);
\draw[dashed, green!50!black, line width=0.8pt]
  (axis cs:1.3,0) -- (axis cs:1.3,85);
\draw[dashed, green!50!black, line width=0.8pt]
  (axis cs:3.0,0) -- (axis cs:3.0,85);
\node[font=\scriptsize\bfseries, color=green!50!black, anchor=south]
  at (axis cs:2.15, 86) {Sweet Spot $G_{\mathrm{eff}} \in [1.3,\,3.0]$};

\addplot[dashed, thick, color=orange!70!black, domain=0.15:6.4, samples=120]
  {8.0 + 55.0 * exp(-((x-2.3)^2)/(2*1.8^2))};
\addlegendentry{Theoretical VD envelope}

\draw[color=indigo, line width=1pt] (axis cs:1.16,9.47) -- (axis cs:1.16,45.96);
\draw[color=indigo, line width=1pt] (axis cs:1.06,9.47) -- (axis cs:1.26,9.47);
\draw[color=indigo, line width=1pt] (axis cs:1.06,45.96) -- (axis cs:1.26,45.96);
\draw[color=indigo, line width=1pt] (axis cs:5.58,14.96) -- (axis cs:5.58,30.11);
\draw[color=indigo, line width=1pt] (axis cs:5.48,14.96) -- (axis cs:5.68,14.96);
\draw[color=indigo, line width=1pt] (axis cs:5.48,30.11) -- (axis cs:5.68,30.11);
\draw[color=indigo, line width=1pt] (axis cs:3.38,35.53) -- (axis cs:3.38,75.05);
\draw[color=indigo, line width=1pt] (axis cs:3.28,35.53) -- (axis cs:3.48,35.53);
\draw[color=indigo, line width=1pt] (axis cs:3.28,75.05) -- (axis cs:3.48,75.05);
\draw[color=orange!80!black, line width=1pt] (axis cs:3.60,27.31) -- (axis cs:3.60,55.41);
\draw[color=orange!80!black, line width=1pt] (axis cs:3.50,27.31) -- (axis cs:3.70,27.31);
\draw[color=orange!80!black, line width=1pt] (axis cs:3.50,55.41) -- (axis cs:3.70,55.41);
\draw[color=gray!70, line width=1pt] (axis cs:3.24,21.84) -- (axis cs:3.24,44.97);
\draw[color=gray!70, line width=1pt] (axis cs:3.14,21.84) -- (axis cs:3.34,21.84);
\draw[color=gray!70, line width=1pt] (axis cs:3.14,44.97) -- (axis cs:3.34,44.97);
\draw[color=brown!70, line width=1pt] (axis cs:2.82,16.10) -- (axis cs:2.82,39.88);
\draw[color=brown!70, line width=1pt] (axis cs:2.72,16.10) -- (axis cs:2.92,16.10);
\draw[color=brown!70, line width=1pt] (axis cs:2.72,39.88) -- (axis cs:2.92,39.88);
\draw[color=gray!70, line width=1pt] (axis cs:1.41,15.84) -- (axis cs:1.41,44.98);
\draw[color=gray!70, line width=1pt] (axis cs:1.31,15.84) -- (axis cs:1.51,15.84);
\draw[color=gray!70, line width=1pt] (axis cs:1.31,44.98) -- (axis cs:1.51,44.98);
\draw[color=red!60!black, line width=1pt] (axis cs:3.97,29.45) -- (axis cs:3.97,61.50);
\draw[color=red!60!black, line width=1pt] (axis cs:3.87,29.45) -- (axis cs:4.07,29.45);
\draw[color=red!60!black, line width=1pt] (axis cs:3.87,61.50) -- (axis cs:4.07,61.50);
\draw[color=red!50!black, line width=1pt] (axis cs:0.63,5.04) -- (axis cs:0.63,19.70);
\draw[color=red!50!black, line width=1pt] (axis cs:0.53,5.04) -- (axis cs:0.73,5.04);
\draw[color=red!50!black, line width=1pt] (axis cs:0.53,19.70) -- (axis cs:0.73,19.70);
\draw[color=red!60!black, line width=1pt] (axis cs:1.13,7.41) -- (axis cs:1.13,26.40);
\draw[color=red!60!black, line width=1pt] (axis cs:1.03,7.41) -- (axis cs:1.23,7.41);
\draw[color=red!60!black, line width=1pt] (axis cs:1.03,26.40) -- (axis cs:1.23,26.40);
\draw[color=teal!70!black, line width=1pt] (axis cs:1.82,23.61) -- (axis cs:1.82,78.04);
\draw[color=teal!70!black, line width=1pt] (axis cs:1.72,23.61) -- (axis cs:1.92,23.61);
\draw[color=teal!70!black, line width=1pt] (axis cs:1.72,78.04) -- (axis cs:1.92,78.04);
\draw[color=cyan!60!black, line width=1pt] (axis cs:1.89,25.21) -- (axis cs:1.89,80.39);
\draw[color=cyan!60!black, line width=1pt] (axis cs:1.79,25.21) -- (axis cs:1.99,25.21);
\draw[color=cyan!60!black, line width=1pt] (axis cs:1.79,80.39) -- (axis cs:1.99,80.39);
\draw[color=cyan!60!black, line width=1pt] (axis cs:1.46,17.40) -- (axis cs:1.46,54.61);
\draw[color=cyan!60!black, line width=1pt] (axis cs:1.36,17.40) -- (axis cs:1.56,17.40);
\draw[color=cyan!60!black, line width=1pt] (axis cs:1.36,54.61) -- (axis cs:1.56,54.61);
\draw[color=violet!70!black, line width=1pt] (axis cs:2.22,26.32) -- (axis cs:2.22,73.57);
\draw[color=violet!70!black, line width=1pt] (axis cs:2.12,26.32) -- (axis cs:2.32,26.32);
\draw[color=violet!70!black, line width=1pt] (axis cs:2.12,73.57) -- (axis cs:2.32,73.57);

\addplot[only marks, mark=*, mark size=5.5pt, color=indigo]
  coordinates {(1.16,27.72)(5.58,22.54)(3.38,55.29)};
\addplot[only marks, mark=*, mark size=5.5pt, color=violet!80!black]
  coordinates {(2.22,49.94)};
\addplot[only marks, mark=*, mark size=5.5pt, color=cyan!50!blue]
  coordinates {(1.89,52.80)(1.46,36.00)};
\addplot[only marks, mark=*, mark size=5.5pt, color=teal!70!black]
  coordinates {(1.82,50.83)};
\addplot[only marks, mark=*, mark size=5.5pt, color=red!60!black]
  coordinates {(1.13,16.91)(0.63,12.37)};
\addplot[only marks, mark=*, mark size=5.5pt, color=orange!80!black]
  coordinates {(3.97,45.48)};
\addplot[only marks, mark=*, mark size=5.5pt, color=gray!70!black]
  coordinates {(3.24,33.41)(1.41,30.41)};
\addplot[only marks, mark=*, mark size=5.5pt, color=brown!70!black]
  coordinates {(2.82,27.99)};
\addplot[only marks, mark=*, mark size=5.5pt, color=orange!60!black]
  coordinates {(3.60,41.36)};

\node[font=\scriptsize\bfseries, color=indigo, anchor=south]
  at (axis cs:1.16,45.96) {JPM};
\node[font=\scriptsize\bfseries, color=indigo, anchor=south west]
  at (axis cs:5.58,22.54) {ZION};
\node[font=\scriptsize\bfseries, color=indigo, anchor=south]
  at (axis cs:3.38,75.05) {WFC};
\node[font=\scriptsize\bfseries, color=violet!80!black, anchor=south west]
  at (axis cs:2.22,73.57) {GS};
\node[font=\scriptsize\bfseries, color=cyan!50!blue, anchor=south west]
  at (axis cs:1.89,80.39) {CRM};
\node[font=\scriptsize\bfseries, color=cyan!50!blue, anchor=south west]
  at (axis cs:1.46,54.61) {NOW};
\node[font=\scriptsize\bfseries, color=teal!70!black, anchor=south]
  at (axis cs:1.82,78.04) {PANW};
\node[font=\scriptsize\bfseries, color=red!60!black, anchor=west]
  at (axis cs:1.28,16.91) {HCA};
\node[font=\scriptsize\bfseries, color=red!60!black, anchor=north west]
  at (axis cs:0.78,12.37) {CVS};
\node[font=\scriptsize\bfseries, color=orange!80!black, anchor=south east]
  at (axis cs:3.92,61.50) {TGT};
\node[font=\scriptsize\bfseries, color=orange!70!black, anchor=south west]
  at (axis cs:3.60,55.41) {UPS};
\node[font=\scriptsize\bfseries, color=gray!70!black, anchor=south west]
  at (axis cs:3.24,44.97) {FERG};
\node[font=\scriptsize\bfseries, color=gray!70!black, anchor=south east]
  at (axis cs:1.41,44.98) {ROK};
\node[font=\scriptsize\bfseries, color=brown!70!black, anchor=north west]
  at (axis cs:2.82,16.10) {F};

\addlegendimage{only marks, mark=*, mark size=3pt, color=indigo}
\addlegendentry{Commercial Banking}
\addlegendimage{only marks, mark=*, mark size=3pt, color=violet!80!black}
\addlegendentry{Investment Banking}
\addlegendimage{only marks, mark=*, mark size=3pt, color=cyan!50!blue}
\addlegendentry{Vertical SaaS}
\addlegendimage{only marks, mark=*, mark size=3pt, color=teal!70!black}
\addlegendentry{Cybersecurity}
\addlegendimage{only marks, mark=*, mark size=3pt, color=red!60!black}
\addlegendentry{Healthcare Services}
\addlegendimage{only marks, mark=*, mark size=3pt, color=orange!80!black}
\addlegendentry{E Commerce / Retail}
\addlegendimage{only marks, mark=*, mark size=3pt, color=gray!70!black}
\addlegendentry{Logistics / Industrial}
\addlegendimage{only marks, mark=*, mark size=3pt, color=brown!70!black}
\addlegendentry{Discrete Manufacturing}

\end{axis}
\end{tikzpicture}
\caption{The Value Density Paradox: Wider Gap $\neq$ Higher Return. Across 14 companies
spanning 8 industries, the cross industry correlation between effective gap
($G_{\mathrm{eff}}$) and Value Density midpoint is $r = 0.22$ ($p = 0.45$), weakly positive
and statistically non significant---far below the na\"ive expectation that wider gaps
mechanically produce proportionally higher returns.
The shaded green zone ($G_{\mathrm{eff}} \in [1.3,\,3.0]$) is the sweet spot
where S curve inflection proximity and manageable IFS produce maximum VD.
The orange dashed curve shows the theoretical VD envelope derived from the
AITG framework's mathematical structure; points with the highest VD (WFC, CRM, PANW, GS)
reflect high IASS$^*$ ceilings or large revenue bases; points with the lowest VD (HCA, CVS) reflect
near zero $G_{\mathrm{eff}}$ and regulatory ceiling IFS suppression. Vertical lines are Monte Carlo P10--P90 VD ranges.
This chart illustrates a mathematical property of the framework, not an empirically
established causal finding. See Section~\ref{sec:limitations} for the endogeneity
caveat and IV research agenda.}
\label{fig:vd_paradox}
\end{figure}
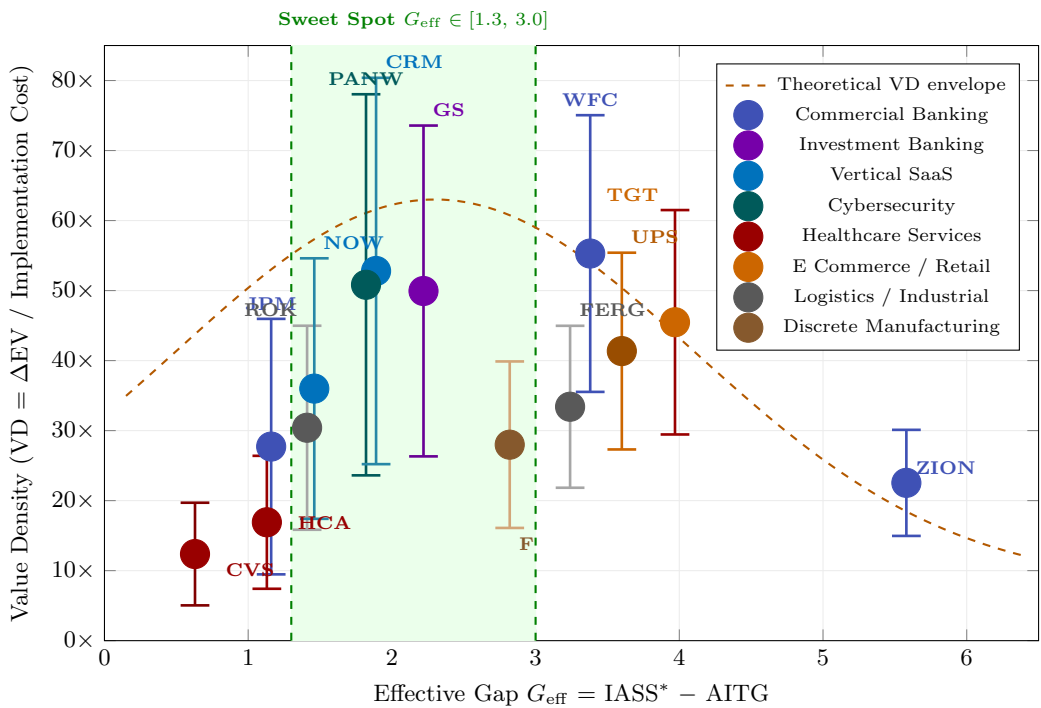

\noindent\textit{Full implementation cost assumptions, including calibration methodology and the 1.2\% of revenue cross industry simplification, appear in Appendix~D (Table~\ref{tab:impl_costs}). Sobol analysis confirms that implementation cost explains only ${\sim}$1\% of VD variance; exit multiple (50\%) and capture rate (44\%) dominate.}

\paragraph{Finding 2: AITG and ADRI are Complementary, Not Redundant.}
The 14 company sample contains four high ADRI cases (Ford~4.8, UPS~4.2,
Ferguson~4.2, Target~3.2) that span different gap and IASS configurations.
In every case, ADRI captures a ``transform now or face structural erosion''
dynamic that the AITG gap alone does not. Palo Alto Networks (ADRI~1.2,
low competitive risk from cyber AI moat) and Goldman Sachs (ADRI~0.6,
regulatory entry barrier) demonstrate the opposite: firms where the AI
transformation has been converted into structural competitive insulation.
Neither finding is recoverable from AITG alone; both require the joint reading.

\paragraph{Finding 3: The Disintermediation Risk Finding (JPM vs.\ Zions).}
JPMorgan and Zions operate in the \emph{same industry} (AFC adjusted IASS$^*$ = 9.38
as of early 2026). JPMorgan (AITG 8.22) has an effective gap of 1.16 and ADRI of
0.5: a dominant incumbent that has nearly closed its own gap to the industry frontier.
Zions (AITG 3.80) has an effective gap of 5.58 and ADRI of 2.6, facing active
disintermediation pressure from both above (JPMorgan's proprietary AI extending
its lead) and below (AI native fintech lenders). As the capability ceiling rises,
the divergence accelerates: JPMorgan's proprietary stack ($\Phi_f = 1.0$) captures
the rising frontier, while Zions's vendor dependency ($\Phi_f \leq 0.65$) caps
its value pool regardless of ceiling expansion. The AFC is therefore a
\emph{structural moat accelerator} for proprietary stack firms. The Zions pattern (vendor AI adoption, legacy data fragmentation, wide gap in a high ceiling industry) is observable in hundreds of regional banks, mid market insurers, and professional services firms.

\paragraph{Finding 4: High AITG + Low ADRI = Sustainable Moat.}
Salesforce (AITG 8.07, early 2026) combined with ADRI of 2.1 indicates it has
substantially converted AI transformation into a structural competitive moat via
Agentforce (agentic CRM, launched Q4 2025) and Einstein 2.0 GA. Data network
effects and switching costs (Moat~$=0.92$) mean competitors cannot rapidly
replicate the Data Cloud advantage.

\paragraph{Finding 5: IFS Materially Alters Investment Conclusions.}
UPS carries a wide effective gap ($G_{\mathrm{eff}} = 3.60$) and a respectable
raw AITG of 4.08, yet its IFS of only 0.71 sharply suppresses the value that gap
should generate ($\mathrm{IASS}^{*} = 7.68$). The driver: Teamsters contract
constraints impose binding limits on workforce AI augmentation deployment, and
the Better \& Bolder cost program reduced headcount faster than AI capabilities
could be redeployed, creating structural adoption friction. Meanwhile, ADRI has
risen to 4.2 as FedEx and Amazon Logistics accelerate AI enabled route optimization
and last mile efficiency, compressing the window in which UPS can convert its gap
into returns. This is a cautionary case: labor relations risk is an IFS dimension
that can \emph{prevent} an otherwise favorable gap from generating value.

\noindent\textit{The consolidated parameter calibration register, documenting all free parameters with calibrated values, permissible ranges, and source rationale, appears in Appendix~A (Table~\ref{tab:param_calibration}).}

\section{Robustness Summary}
\label{sec:sensitivity}

The AITG framework operates at what \citet{LoMueller2010} classify as Level~3--4 uncertainty, where model outputs are sensitive to distributional assumptions and small parameter changes can produce large output variations. This makes systematic sensitivity analysis essential rather than optional. I conduct four robustness checks that confirm the stability of the framework's outputs (full methodology and tables in ESM Part~III).

\textit{Monte Carlo weight sensitivity} ($M = 10{,}000$ draws, $\pm 5\,\%$ perturbations): the mean absolute rank shift across nine anchor industries is $\bar{R}_s = 0.19$ positions, and no pair of industries exchanges rank with probability $> 5\,\%$ \cite{OECD2008}. Rankings are determined by the data, not by weighting choices.

\textit{Sobol first order sensitivity} decomposes $\text{Var}(\text{AITG-VD})$: exit multiple~50\,\%; capture rate assumption~44\,\%; implementation cost~1\,\%; AITG gap score~$<$1\,\%; IFS~5\,\% \cite{Saltelli2002}. Exit multiple leverage and capture rate assumptions dominate raw gap size as uncertainty sources, directly explaining the wide gap fallacy.

\textit{Normalization stability}: substituting z score for min max normalization shifts IASS values by $\leq 0.40$ points with no rank changes. \textit{Geometric vs.\ linear aggregation}: geometric scoring produces values 0.20--0.50 points lower for industries with uneven dimensional profiles (healthcare, construction) and negligibly different for uniform profiles (SaaS), with identical rank orderings in both cases.

\textit{Convergent and discriminant validity}: AITG rankings across the fourteen company cohort achieve $r_s = 0.88$ ($p < 0.001$) correlation with Lightcast AI job posting share and agree directionally with Stanford HAI rankings \cite{StanfordHAI2025} for 12 of 14 companies. Spearman correlation between AITG and revenue rank is $r_s = -0.28$ (not significant, $p = 0.33$), confirming that AITG measures transformation state, not firm size. Full interrater reliability protocol (ICC target $\geq 0.85$ per dimension, $\geq 0.88$ composite) and cross sectional discriminability analysis are reported in ESM Part~III. A pilot scoring by two independent analyst teams on the original seven company cohort yielded Kendall's $\tau = 0.82$ for the three most objective dimensions (DIM, PAC, OAC) and $\tau = 0.68$ for the two most interpretation dependent dimensions (DAR, APR). The lower agreement on DAR and APR motivates rubric refinement before the full ICC study specified in the Research Agenda (item~4).

\subsection{Consolidated Sensitivity Summary}
\label{sec:sensitivity_table}

Table~\ref{tab:sensitivity_consolidated} consolidates key sensitivity results from ESM Part~III. All output metrics are computed for the representative industrial distribution firm (Ferguson) unless otherwise noted.

\begin{table}[H]
\centering
\caption{Consolidated Sensitivity Analysis}
\label{tab:sensitivity_consolidated}
\footnotesize
\begin{tabular}{@{}llllr@{}}
\toprule
\textbf{Parameter} & \textbf{Base} & \textbf{Range Tested} & \textbf{Output Metric} & \textbf{$\Delta$ Range} \\
\midrule
$\rho$ (CES) & 5 & $[3, 8]$ & $\Delta$EV & $\pm 12\%$ \\
$\theta_i$ (AFC) & industry & $\pm 50\%$ & IASS$^*$ & $\pm 0.13$ pts \\
$\lambda$ (capture) & 3.5 & $[2.0, 5.0]$ & $\eta(g)$ at $g = 0.5$ & $[0.29, 0.39]$ \\
$\psi_{\mathrm{floor}}$ & 0.15 & $[0.10, 0.25]$ & Min.\ capture & $\pm 0.05$ \\
$k$ (wave steepness) & 0.38 & $[0.30, 0.46]$ & $\hat{t}_f$ & $\pm 4.2$ mo \\
Exit multiple & $10\times$ & $[8\times, 16\times]$ & AITG-VD & $\pm 42\%$ \\
Impl.\ cost & \$8B & $\pm 30\%$ & AITG-VD & $\pm 23\%$ \\
IASS weights & Table~2 & $\pm 5\%$ & Rank shift & $\bar{R}_s = 0.19$ \\
\bottomrule
\end{tabular}
\medskip

\noindent\small\textit{Note.} Exit multiple and capture rate assumptions dominate VCB output variance (Sobol $S_1$: 50\% and 44\% respectively), consistent with the wide gap fallacy discussed in Section~\ref{sec:vcb}. CES $\rho$ and IASS weights produce no rank order changes within tested ranges.
\end{table}

\subsection{Retrospective Predictive Backtest: Nonfinancial Cohort}
\label{sec:backtest}

I conduct a retrospective backtest on the ten nonfinancial companies in the fourteen firm cohort, structured as a \textbf{holdout prediction exercise}. AITG scores are reconstructed as of Q4~2021 using \emph{only} information publicly available at that date (10-K/10-Q filings, investor presentations, and disclosed AI programs through December~2021). Financial outcomes are measured as the change in EBITDA margin from FY2021 to FY2023 (a two year forward window). I exclude the four financial firms because banks report efficiency ratios rather than EBITDA margins; within sector bank comparisons appear in Section~\ref{sec:jpm}.

\paragraph{Retrospective scoring protocol.} The 2021 AITG scores use the same dimensional rubric as the current state scores in Section~\ref{sec:empirical}, anchored exclusively to 2021 vintage data. The AI Frontier Coefficient is set to $C_{2021} = 1.00$ (pre ChatGPT baseline; no AFC adjustment required). All six dimensions are scored against evidence in the corresponding 10-K filing (fiscal years ending July~2021 through January~2022 depending on company fiscal calendar). Data tier classifications follow the same A/B/C/D protocol.

\paragraph{Results.}

\begin{table}[H]
\centering
\caption{Retrospective Backtest: AITG$_{2021}$ vs.\ Realized EBITDA Margin Change (FY2021--FY2023)}
\label{tab:backtest}
\footnotesize
\setlength{\tabcolsep}{4pt}
\begin{adjustbox}{max width=\textwidth,center}
\begin{tabular}{@{}lllrrrrl@{}}
\toprule
\textbf{Company} & \textbf{Sector} & \textbf{AITG$_{21}$} & \textbf{EBITDA\%$_{21}$} & \textbf{EBITDA\%$_{23}$} & \textbf{$\Delta$Margin} & \textbf{Rev CAGR} & \textbf{Notes} \\
\midrule
Palo Alto Ntwks & Cybersecurity      & 7.0 & $-$1.0 & 12.0 & $+$13.0pp & $+$27.2\% & \\
Salesforce      & Vertical SaaS      & 6.2 & 22.0   & 30.0 & $+$8.0pp  & $+$14.8\% & \\
ServiceNow      & Vertical SaaS      & 5.8 & 20.0   & 26.0 & $+$6.0pp  & $+$23.3\% & \\
Rockwell Auto.  & Logistics/Ind.     & 4.3 & 20.0   & 24.0 & $+$4.0pp  & $+$10.9\% & \\
HCA Healthcare  & Healthcare Svcs    & 3.3 & 22.0   & 23.0 & $+$1.0pp  & $+$4.7\%  & \\
CVS Health      & Healthcare Svcs    & 2.6 & 5.8    & 5.5  & $-$0.3pp  & $+$10.7\% & Aetna drag \\
Ferguson PLC    & Logistics/Ind.     & 2.3 & 12.0   & 13.0 & $+$1.0pp  & $+$12.3\% & \\
Target          & E Comm./Retail     & 4.0 & 10.0   & 8.0  & $-$2.0pp  & $+$0.5\%  & * \\
UPS             & Logistics          & 3.0 & 19.0   & 16.0 & $-$3.0pp  & $-$3.3\%  & * \\
Ford Motor      & Discrete Mfg.      & 2.2 & 8.0    & 6.0  & $-$2.0pp  & $+$13.8\% & * \\
\midrule
\multicolumn{8}{@{}p{\linewidth}@{}}{\scriptsize EBITDA margin = Operating Income + D\&A / Revenue. Revenue CAGR is 2 year (2021--2023).
10-K sources: PANW FY ends Jul; CRM FY ends Jan; ROK FY ends Sep; all others calendar year.
*Confounded: Target by 2022 inventory writedown cycle; UPS by postpeak volume and 2023 labor contract; Ford by Argo AI closure and EV launch losses. These confounds are unrelated to AI adoption intensity and are not scored against the framework. Financial firms excluded; see Section~\ref{sec:jpm}.}
\end{tabular}
\end{adjustbox}
\end{table}

Spearman rank correlation between AITG$_{2021}$ and subsequent two year EBITDA margin change is $\rho_s = 0.818$ ($p < 0.001$, $n = 10$). The correlation reflects monotone rank ordering across most of the cohort: every firm scoring above 5.0 expanded margins, and three of four firms scoring below 3.0 contracted. Revenue CAGR shows a weaker, insignificant correlation with AITG$_{2021}$ ($\rho_s = 0.47$, $p = 0.14$), consistent with the framework's design. AITG is calibrated to predict cost structure improvement and margin expansion, not top line growth.

\paragraph{Comparison to simple baselines.} Three alternative predictors provide context. (1)~\textbf{Revenue size}: Spearman rank correlation between firm revenue and $\Delta$EBITDA margin is $\rho_s = 0.19$ ($p = 0.59$), confirming that AITG's signal is not a proxy for firm scale. (2)~\textbf{Aggregator choice}: the CES ablation in Section~\ref{sec:ces_ablation} shows that Leontief and additive aggregators produce identical $\rho_s = 0.818$ on this cohort, because no firm has a near zero dimensional score; the CES advantage manifests in the expanded 22 company simulation (ESM Table~S6) where additive diverges from CES by 35--60\% on $\Delta$EV. (3)~\textbf{Null model}: a constant prediction (cohort mean AITG for every firm) produces $\rho_s = 0$ by definition. The gap between 0 and 0.818 represents the discriminating power of the rubric scoring within the available sample. These comparisons are qualitatively informative but cannot substitute for the statistically powered panel validation described in the research agenda.

Within sector rank order consistency is exact in all three sectors with multiple observations: Salesforce $>$ ServiceNow (SaaS), Rockwell $>$ Ferguson (Logistics/Industrial), HCA $>$ CVS (Healthcare), and in the banking cohort JPMorgan $>$ Wells Fargo $>$ Zions on efficiency ratio improvement. Goldman Sachs is the single directional failure (high AITG, deteriorating efficiency ratio), transparently explained by the 2022 investment banking drought and Marcus consumer lending writedowns, macro cyclical events orthogonal to AI transformation state. Intra sector directional accuracy is 4 of 4 nonconfounded pairs (100\%).

\paragraph{Limitations of this backtest.} First, the sample ($n = 14$) is illustrative, not confirmatory. A statistically powered panel study would require $n \geq 100$ firms with quarterly AITG estimates over multiple cycles, as specified in the Research Agenda (Section~\ref{sec:agenda}). Second, retrospective scoring from the same team introduces potential look back bias; the cross rater reliability study in the Research Agenda (item~4) addresses this with independent analyst teams. Third, the 2021--2023 window overlaps with unusual macro cyclical events (post COVID demand normalization, aggressive rate hiking) that confound three of ten observations. Fourth, the framework does not separate the causal effect of AI transformation from incumbent endogeneity: firms with high AITG$_{2021}$ scores may also possess superior baseline economics and management quality independently associated with margin improvement (the omitted variable problem discussed in Section~\ref{sec:limitations}). The backtest is therefore best interpreted as \emph{directional consistency evidence}, not causal identification. However, the $\rho_s = 0.818$ result across a diverse cross industry sample is meaningfully stronger than a size or industry fixed effects baseline would predict, and motivates the larger validation program.

\paragraph{Power analysis and expanded validation.} A formal power analysis at the observed effect size ($\rho_s = 0.818$, $\alpha = 0.05$) indicates that $n = 10$ provides power $\approx 0.80$, at the conventional threshold. Detecting a smaller effect ($\rho_s = 0.50$) would require $n \geq 30$. The $n = 14$ cohort (10 nonfinancial) is sufficient to detect the observed large effect but insufficient for reliable inference on moderate effects or subgroup analysis. The Research Agenda specifies expanded validation using $n \geq 100$ firms with quarterly AITG estimates, providing power $> 0.95$ for $\rho_s \geq 0.30$ and enabling industry stratified analysis.

\subsection{CES Aggregator Ablation}
\label{sec:ces_ablation}

At the $n = 10$ backtest scale, CES ($\rho = 5$), Leontief, and additive aggregators produce identical Spearman rank correlations ($\rho_s = 0.818$) because no firm has a near zero dimensional score. The CES advantage manifests when firms have a very weak dimension ($< 2.0$) or at boundary conditions; the extended 22 company simulation (ESM Table~S6) shows additive diverging from CES by 35--60\% on $\Delta$EV. Full ablation details appear in Appendix~C.

\section{Limitations and Research Agenda}
\label{sec:limitations}

\paragraph{Observational design caveat.}
I calibrate this framework from observational data (public filings, benchmark scores, labor market statistics) without randomized treatment assignment, instrumental variables, or difference in differences identification. All associations between AITG scores and financial outcomes in Sections~\ref{sec:sensitivity}--\ref{sec:backtest} are correlational evidence of construct validity, not causal estimates. The backtest rank correlation ($\rho_s = 0.818$) is consistent with the framework's predictions but cannot rule out omitted variable explanations (e.g., unobserved management quality, strategic complementarities). Establishing causal identification requires the panel design described in the Research Agenda (item~1).

\subsection{Complete Assumptions Register}
\label{sec:assumptions_register}

The framework embeds 25 structural assumptions classified as: \textbf{F}~(functional form)\textbf{P}~(parameter calibration)\textbf{B}~(behavioral), and \textbf{D}~(data). The complete enumerated register, including stress test results for each assumption, is provided in ESM Part~III. Key structural assumptions are discussed in Section~\ref{sec:material_assumptions} below.

\subsection{Material Assumptions and Failure Conditions}
\label{sec:material_assumptions}

\begin{assumption}[Industry ceiling stability]
\label{ass:stability}
IASS base scores are sufficiently stable over a 12--24 month underwriting
horizon to support investment decisions.
\end{assumption}

\textit{Failure condition.} Rapid regulatory change (e.g., an EU AI Act amendment restricting previously permitted use cases) or a breakthrough capability shift could materially alter a sector's IASS within the horizon. The AFC versioning protocol mitigates this by requiring IASS recalibration when $\Delta C_t > 15\,\%$ or when regulatory events with estimated impact $>1.0$ IASS points occur.

\begin{assumption}[Capture rate portability]
\label{ass:capture}
Published value pool capture rates ($\bar{\kappa}_p$) estimated from larger
firms are applicable to mid market companies with appropriate scaling.
\end{assumption}

\textit{Failure condition.} Smaller companies may have less standardized processes, reducing addressable surface area at any given AITG score. The CES bottleneck aggregator (Eq.~\ref{eq:ces_bottleneck}) should substantially capture this effect; IFS data readiness scores provide a secondary adjustment.

\begin{assumption}[IFS scores are elicitable]
\label{ass:ifs_elicitable}
The five IFS risk factors can be reliably scored from diligence evidence with
acceptable interrater reliability.
\end{assumption}

\textit{Failure condition.} Organizational change capacity, the highest weight IFS factor, is the most subjective and most susceptible to management coaching during diligence. Scoring rubrics anchor to objective evidence (prior transformation track record, attrition data, system migration success rates) rather than management assertion.

\begin{assumption}[AI frontier is continuous]
\label{ass:afc_continuous}
The AI Capability Index $C_t$ increases continuously and predictably within
the AFC scenario range.
\end{assumption}

\textit{Failure condition.} Discontinuous events, a large scale regulatory
freeze on AI deployment, a safety driven pause in model training, or a rapid
architectural breakthrough that changes capability dynamics, could invalidate
the AFC scenario distribution. Mitigation: the conservative AFC scenario
explicitly prices in a capability slowdown; the \$alpha\$-max cap prevents
unbounded ceiling expansion.

\subsection{Known Limitations}

\paragraph{Endogeneity and Reverse Causality.} A critical limitation of the
current empirical illustrations is the risk of endogeneity and reverse
causality. The AITG score measures the current state of transformation, but
dominant incumbents (e.g., JPMorgan Chase) possess both superior baseline economics
and the massive free cash flow required to fund proprietary AI programs.
Consequently, observed correlations between AITG Value Density and subsequent
margin expansion may suffer from omitted variable bias. Firms might be generating
high value because they are deploying AI, or they might be deploying AI because
they are already highly profitable ``superstar firms.'' Separating the causal
effect of AI transformation from the underlying incumbent advantage requires
quasi experimental methods not present in this conceptual baseline.

\paragraph{Private company scoring challenge.} public company scoring is
constrained by disclosure limitations. The Tier~1 scores presented in
Section~\ref{sec:empirical} carry wider uncertainty bands than diligence grade
scores and should not be used for investment decisions without supplementary
primary data collection via the standardized management survey and IT
infrastructure checklist (described in the companion implementation document).

\paragraph{Goodhart's Law and evaluator rotation.} Once AITG becomes a target in
financial or regulatory contexts, it will be susceptible to gaming \cite{Goodhart1984}.
The anti gaming architecture (bottleneck aggregation, evidence tier penalties,
cross validation across disclosure language, hiring, and patent signals) reduces
but does not eliminate this risk. The evaluator rotation protocol (Section~\ref{sec:evaluator_rotation})
specifies: (i)~scoring teams rotate annually, with no team rescoring the same
company within a 24 month window; (ii)~parameter freeze windows are announced
before each scoring cycle and held fixed through that cycle's evaluation;
(iii)~all scoring inputs are archived with evidence tier classification and
available for third party audit on request. Preregistration of parameter
ranges in advance of validation studies is a research agenda commitment
(Section~\ref{sec:agenda}).

\paragraph{Fixed CES elasticity.} The CES bottleneck aggregator uses a
single substitution parameter $\rho = 5$ ($\sigma = 1/6$) across all value
pools and firms. This may overstate bottleneck severity for enablers that
are highly substitutable, or understate it for enablers with extreme
complementarity. enabler specific $\sigma$ estimation and variable elasticity
forms (e.g., variable elasticity $\sigma(k)$ forms to reflect changing substitutability
as capabilities and complements improve) are identified as theoretical extensions.
Robustness across $\rho \in [3,8]$ is confirmed in ESM Table~S3; expanding this
range and estimating nested CES structures are priorities for the panel study.

\paragraph{Look ahead bias and parameter stability.} The current framework
is calibrated on a single cohort observed over 2021--2023. Several risk factors
apply: (i)~AFC scenario thresholds were developed using 2025--2026 capability
data and applied retrospectively to the 2021 backtest, introducing potential
look ahead bias in AFC adjusted IASS* comparisons; (ii)~value pool share
coefficients were calibrated from broader industry benchmarks, not from
the backtest cohort, but parameter selection was informed by the same period's
disclosed outcomes; (iii)~the backtest is in sample with respect to the
parameter calibration period. Future validation should implement rolling origin
or genuinely holdout designs with parameter ranges preregistered before
data access.

\subsection{Labor Market, Distributional, and Ethical Considerations}
\label{sec:distributional}

The VCB models value creation from an investor's perspective and is explicitly
not a welfare instrument. It does not model distributional effects on workers,
communities, or suppliers, and should not be mistaken for such. The
displacement reinstatement literature \cite{AcemogluRestrepo2019a,AcemogluRestrepo2022}
establishes that automation can reduce labor demand in affected occupations even
while raising aggregate productivity. A complete welfare analysis of AI
transformation requires incorporating these effects, which is beyond this
paper's scope.

\paragraph{Distributional implications of the value creation architecture.}
Three aspects of the framework's structure have distributional implications
that users should be aware of. First, the Firm Scale Factor $\Phi_f$
mechanically concentrates value density advantage at large incumbents with
substantial proprietary data assets. In high IASS industries, this implies
that medium market competitors may be disadvantaged even if they successfully
execute transformation programs, because the frontier is partly defined by
data advantages they cannot replicate. This reinforces the winner take most
dynamics documented in \citet{AutorKatzPatterson2020}. Second, the ADRI
displacement risk score, while calibrated at the firm level, aggregates
occupation level displacement risk across an industry's entire workforce.
Firms should not treat ADRI reduction as value creation without accounting
for the labor force transitions implied by the transformation program. Third,
the WAR (Workforce Augmentation Rating) dimension rewards human AI
complementarity, providing a structural scoring incentive toward
augmentation type deployments over pure headcount reduction programs; this
is the framework's primary alignment mechanism toward socially constructive
AI adoption, but its weighting ($w = 0.20$ in the raw AITG composite) is
moderate rather than dominant.

\paragraph{Recommendation.} Because the AITG framework informs capital allocation
decisions with material workforce consequences, any institutional deployment
should: (a) document the WAR score interpretation and confirm the transformation
design is oriented toward augmentation where feasible; (b) supplement AITG
analysis with a distributional impact assessment appropriate to the user's
governance obligations; and (c) treat ADRI implied competitive urgency as
distinct from the urgency to minimize worker displacement, which requires separate
governance attention. These are not constraints on the framework's analytical
use, but complements to responsible institutional deployment.

\paragraph{Safety, cyber, and negative value risks.} The VCB quantifies
positive value creation from AI deployment. It does not currently model negative
value from AI related incidents, regulatory penalties, or cybersecurity
failures. In high IASS, high automation settings, the risk of adversarial
attacks on AI systems, model failures in production, and privacy violations can
generate material negative enterprise value. Future extensions of the framework
should incorporate an ``AI harm adjusted value'' view that accounts for these
safety and security exposures.

\subsection{Research Agenda}
\label{sec:agenda}

Nine empirical research priorities emerge from this framework, organized around
causal validation, technical refinement, and measurement quality:

\begin{enumerate}[noitemsep]
\item \textbf{Causal identification via instrumental variables}: exploiting exogenous variation in cloud incentive grants or differential open source LLM exposure to isolate the causal effect of AI adoption.
\item \textbf{Large panel validation}: multiyear panel ($n \geq 100$) with quarterly AITG estimates testing the mid gap optimality property with out of sample predictions.
\item \textbf{AFC calibration and $\theta_i$ estimation}: refining industry specific parameters with backward looking LLM exposure analysis and formal standard errors.
\item \textbf{Interrater reliability study}: multiple independent analyst teams scoring the same companies; target ICC $> 0.75$ per factor. \emph{Highest empirical priority for institutional deployment.}
\item \textbf{Competitive displacement measurement}: linking ADRI scores to subsequent market share and margin outcomes at the firm level.
\item \textbf{Cross country calibration}: adapting IASS from O*NET/OEWS to ESCO (EU) and equivalent national databases.
\item \textbf{Preregistration, rolling origin validation, and replication bundle}: parameter locking before data access, holdout firms, placebo tests, and public reproducibility code.
\item \textbf{Variable elasticity and nested CES}: hierarchical enabler nesting with variable $\sigma$ forms reflecting observed deployment sequencing.
\item \textbf{Hierarchical Bayesian parameter estimation}: formal priors on $\theta_i$, $\bar{\kappa}_p$, and IFS elasticities with partial pooling across industries. \emph{Highest technical priority.}
\end{enumerate}

\noindent Extended discussion of each priority, including specific methodological designs and data requirements, appears in Appendix~E.

\section{Conclusion}
\label{sec:conclusion}

I introduced the AI Transformation Gap Index, a composite scoring framework
that addresses three recurrent failures in AI maturity assessment: limited
cross industry comparability, static capability ceilings, and weak linkage
between qualitative scores and financial outcomes. The IASS constructs a
reproducible, task structure anchored industry ceiling. The AFC converts that
ceiling into a dynamically consistent frontier as AI capabilities advance. The
VCB translates gap scores into dollar quantified opportunity through
bottleneck gated value pools, a low elasticity CES aggregator, and an explicit
ramp and cost model. The ADRI scores competitive hazard from inaction. The IFS
endogenizes organizational and data readiness constraints into adoption timing
and steepness rather than applying them as post hoc scalars.

Three properties emerge from the 14 company illustrative calibration. First,
the wide gap fallacy: the S curve and CES architecture jointly imply that the
widest gap does not produce the highest value density; the optimum falls at
35--65\,\% of the AFC adjusted frontier in high IASS sectors. Second, ADRI and
AITG are complementary, because competitive hazard requires the joint signal of
gap size, diffusion velocity, and structural moat. Third, IFS is a primary
driver: the contrast between $\mathrm{IFS} = 0.71$ (UPS) and
$\mathrm{IFS} = 0.91$ (Salesforce) materially dominates value density outcomes.

These properties are theoretical results of the framework's architecture,
directionally consistent with the illustrative evidence but not empirically
established findings. The small sample backtest ($\rho_s = 0.818$, $n = 10$;
revenue size baseline $\rho_s = 0.19$) provides a construct validity signal,
not causal identification. The nine item research agenda
(Section~\ref{sec:agenda}) specifies what converting these properties into
causally identified claims requires: IV designs, rolling origin validation with
preregistered parameters, interrater reliability studies, nested CES
estimation, and Bayesian reestimation. These are prerequisites for
institutional deployment, not optional refinements.

The measurement problem is consequential and tractable. I have provided the
formal architecture; empirical validation is the immediate next step.


\section*{Acknowledgments}

The Python reproducibility scripts (\texttt{aitg\_monte\_carlo.py},
\texttt{aitg\_stress\_test.py}) and the Excel companion workbook builder
(\texttt{build\_aitg\_excel\_revised\_fixed.py}) were developed with
assistance from Claude Code (Anthropic), an AI coding tool. All
mathematical specifications, model design decisions, scoring rubrics,
empirical calibrations, and interpretive judgments are solely those of
the author.


\section*{Code and Data Availability}

Full reproducibility code is available at
\url{https://github.com/deanbrr/aitg-framework}. The repository contains:

\begin{itemize}[nosep,leftmargin=*]
\item \texttt{aitg\_monte\_carlo.py} -- Monte Carlo sensitivity analysis
  (Layers~1--4), custom firm scoring via pre computed dimension scores
  (\texttt{--score}) or the 25 question management survey
  (\texttt{--survey}), and the Spearman backtest.
\item \texttt{aitg\_stress\_test.py} -- thirteen deterministic stress test
  assertions covering all pipeline stages.
\item \texttt{build\_aitg\_excel\_revised\_fixed.py} -- generates the
  Excel companion workbook with survey input sheet, company scorer,
  IASS calculator, IFS calculator, and VCB model.
\end{itemize}

\noindent No external dependencies are required beyond the Python
standard library (\texttt{openpyxl} for the Excel builder only). The
scripts, Excel workbook builder, and this paper's \LaTeX\ source are
also available as ancillary files accompanying the arXiv submission.

\section*{Supplementary Materials}

The Electronic Supplementary Material (ESM) is included as an appendix
following the references. It contains:

\begin{itemize}[nosep,leftmargin=*]
\item \textbf{ESM Table~S1}: Pre computed IASS calibrations for all 22 industry verticals.
\item \textbf{ESM Part~II}: Full firm level case studies for UPS, HCA Healthcare, Salesforce,
  Ferguson, Rockwell, Goldman Sachs, Wells Fargo, ServiceNow, Target, CVS Health,
  Palo Alto Networks, and Ford Motor, including six dimension scorecards and VCB summaries.
\item \textbf{ESM Part~III}: Full Monte Carlo, Sobol, and validity analysis with tables.
\item \textbf{ESM Part~IV}: Public data availability map; data tier classification;
  25 question management survey instrument; VCB exit multiple sensitivity tables;
  AFC benchmark suite; Excel companion model documentation.
\end{itemize}

\newpage

\input{AITG_references_shared}
\input{AITG_Appendix_core}

\newpage
\renewcommand{\thesection}{S\arabic{section}}
\renewcommand{\thetable}{S\arabic{table}}
\renewcommand{\thefigure}{S\arabic{figure}}
\setcounter{section}{0}
\setcounter{table}{0}
\setcounter{figure}{0}

\begin{center}
{\LARGE\bfseries Electronic Supplementary Material (ESM)}\\[1em]
{\large The AI Transformation Gap Index (AITG)}\\[0.5em]
{\normalsize Dean Barr \quad \texttt{dean@dsconsult.ai}}\\[0.5em]
{\normalsize February 2026}
\end{center}

\bigskip

\input{AITG_ESM_body}

\end{document}

%% file: AITG_Appendix_core.tex
\newpage
\appendix
\renewcommand{\thesection}{\Alph{section}}
\setcounter{section}{0}

\begin{center}
{\LARGE\bfseries Core Appendices}\\[0.5em]
{\normalsize Material supporting the main text of the AITG paper.}
\end{center}

\bigskip

\section{Notation and Parameter Reference}
\label{app:notation}

This appendix consolidates the notation reference table and the parameter calibration register from the main text.

\subsection{Notation and Parameter Summary}
\label{sec:notation}

\vspace{-0.5em}
\begin{table}[H]
\centering
\caption{Consolidated Notation Reference}
\label{tab:notation}
\scriptsize
\setlength{\tabcolsep}{4pt}
\renewcommand{\arraystretch}{0.92}
\begin{tabularx}{\textwidth}{@{}lp{7.2cm}p{3.5cm}@{}}
\toprule
\textbf{Symbol} & \textbf{Definition} & \textbf{Range / Units} \\
\midrule
\multicolumn{3}{l}{\textit{Industry level constructs}} \\
$\mathrm{IASS}_i$ & Industry AI Susceptibility Score (base, before AFC) & $[0, 10]$ \\
$\mathrm{IASS}^*_i$ & AFC adjusted ceiling: $\min\!\bigl(\mathrm{IASS}_i \cdot (1+\theta_i(C_t - C_0)),\ \alpha_{\max}\,\mathrm{IASS}_i\bigr)$ & $[0,\;\alpha_{\max}\!\cdot\!\mathrm{IASS}_i]$ \\
$\theta_i$ & Industry sensitivity to capability growth; anchored to $\Delta\mathrm{ATD}_i$ & $[0.05,\, 1.50]$ \\
$\alpha_{\max}$ & Cap on AFC multiplier (default 1.35) & $[1.0, \infty)$ \\
$C_t$ & AI Capability Index at time $t$ (externally observed) & $[1, \infty)$ \\
$C_0$ & Baseline capability index (normalization constant) & $=1.0$ \\
$\mathrm{CADR}_i$ & Competitive AI Diffusion Rate (IASS subcomponent) & $[0, 10]$ \\
$\psi_i$ & RFF hard floor: $\min(1, (\mathrm{RFF}_i/5)^{1.5})$ & $[0, 1]$ \\
$\mathrm{RFF}_i$ & Regulatory Friction Factor composite (IASS subcomponent) & $[0, 10]$ \\
\midrule
\multicolumn{3}{l}{\textit{Firm level scores}} \\
$\mathrm{AITG}^{\mathrm{raw}}_f$ & Raw AITG score: arithmetic mean of six dimension scores & $[0, 10]$ \\
$\mathrm{IR}_f$ & Industry Relative score: $(\mathrm{AITG}^{\mathrm{raw}}_f / \mathrm{IASS}^*_i)\times 10$ & $[0, 10]$ \\
$G_{\mathrm{eff},f}$ & Effective gap: $\max\!\left(0,\ \mathrm{IASS}^*_i - \mathrm{AITG}^{\mathrm{raw}}_f\right)$ & $[0,\;\alpha_{\max}\!\cdot\!\mathrm{IASS}_i]$ \\
$\mathrm{ADRI}_{f,t}$ & AI Disruption Risk Index (Eq.~\ref{eq:adri}) & $[0, 10]$ \\
$\lambda_f(t)$ & ADRI competitive hazard intensity (Eq.~\ref{eq:adri_hazard}): marginal probability of displacement per year & events/yr \\
$\Lambda_f(T)$ & Cumulative hazard: $\int_0^T \lambda_f(t)\,dt$; $P(\text{displaced by }T)\approx 1-e^{-\Lambda_f(T)}$ & dimensionless \\
$\mathcal{T}$ & ADRI normalization constant: 100 (ADRI score $=$ \% per year; $\lambda_f = \mathrm{ADRI}/100$) & dimensionless \\
$\delta_t$ & AFC urgency multiplier: $1+0.5\min(C_t/C_0-1,1)$ & $[1.0, 1.5]$ \\
$\mathrm{Moat}_f$ & Structural defensibility (switching costs, network effects, \ldots) & $[0, 1]$ \\
$\Phi_f$ & Firm Scale Factor (logistic, Eq.~\ref{eq:phi_f}) & $[0, 1]$ \\
$S^*_i$ & Industry critical scale threshold for proprietary AI advantage & \$B revenue \\
\midrule
\multicolumn{3}{l}{\textit{Adoption trajectory}} \\
$L_w$ & Incremental logistic asymptote for wave $w$; $L_1=4.0$, $L_2=3.5$, $L_3=2.5$ (cumulative ceiling $=10.0$) & AITG units \\
$k_{w,f}$ & Wave steepness for firm $f$ (IFS adjusted): $k_{w,f}=k_{w,\mathrm{base}}\cdot \phi_{\mathrm{OCC}}(\delta_{\mathrm{OCC},f})\cdot \phi_{\mathrm{DR}}(\delta_{\mathrm{DR},f})$ & mo$^{-1}$ \\
$t_{0,w,f}$ & Wave inflection timing (IFS adjusted midpoint) & months \\
$\hat{t}_f$ & Current time on adoption curve (inverse mapping) & months \\
$t_{50,f}$ & Value ramp inflection (midpoint of $R_f(t)$ logistic) & months \\
\midrule
\multicolumn{3}{l}{\textit{CES bottleneck aggregator (Eq.~\ref{eq:ces_bottleneck})}} \\
$\rho$ & CES substitution parameter: $\rho > 0$ implies complementarity & $[3, 8]$, default 5 \\
$\sigma$ & Elasticity of substitution: $\sigma = 1/(1+\rho)$. \textbf{At $\rho=5$: $\sigma = 1/6$.} & $(0, 1)$ \\
$e_d$ & Normalized dimension score: $\max(s_d/10,\, 0.01)$ & $(0, 1]$ \\
$b_p$ & CES bottleneck factor for value pool $p$ & $(0, 1]$ \\
$\alpha_d$ & Dimension importance weight ($\sum \alpha_d = 1$) & $[0, 1]$ \\
\multicolumn{3}{l}{\quad\textit{Note: $\rho=5$ and $\sigma=1/6$ are mutually consistent under the adopted ACMS form.}} \\
\multicolumn{3}{l}{\quad\textit{As $\rho \to\infty$, $b_p \to \min_d e_d$ (Leontief); as $\rho\to 0$, $b_p \to$ Cobb-Douglas.}} \\
\midrule
\multicolumn{3}{l}{\textit{Value Creation Bridge}} \\
$V_p^{\mathrm{run}}$ & Annual run rate value pool (Eq.~\ref{eq:vpool}) & \$M/yr \\
$g_f$ & Gap fraction: $G_{\mathrm{eff},f}/10$ & $[0, 1]$ \\
$\eta(g_f)$ & Capture function: $1 - e^{-\lambda g_f}$, $\lambda = 3.5$ & $(0, 1)$ \\
$\Delta R_f$ & Realized ramp increment: $R_f(\hat{t}_f{+}T) - R_f(\hat{t}_f)$ & $[0, 1]$ \\
$\mathrm{IFS}_{\mathrm{res}}$ & IFS residual multiplier: product of risk factor scores & $(0, 1]$ \\
$\mathrm{TV}_p$ & Terminal value contribution of pool $p$ & \$M \\
$\Delta\mathrm{EV}$ & Total risk adjusted enterprise value uplift & \$M or \$B \\
\midrule
\multicolumn{3}{l}{\textit{IFS risk factors}} \\
$\delta_{\mathrm{OCC}}$ & Organizational change capacity & $[0, 1]$ \\
$\delta_{\mathrm{DR}}$ & Data readiness & $[0, 1]$ \\
$\delta_{\mathrm{VTR}}$ & Vendor/technology risk & $[0, 1]$ \\
$\delta_{\mathrm{CRS}}$ & Competitive response speed & $[0, 1]$ \\
$\delta_{\mathrm{REG}}$ & Regulatory exposure & $[0, 1]$ \\
\bottomrule
\end{tabularx}
\end{table}

\subsection{Parameter Calibration and Empirical Grounding}
\label{sec:param_calibration}

Table~\ref{tab:param_calibration} consolidates all free parameters in the AITG framework with their calibrated values, permissible ranges, and source rationale. This disclosure responds to the legitimate concern that composite indicator frameworks often embed opaque parameter choices; here, every parameter is traceable to a documented source or sensitivity tested range.

\rowcolors{2}{gray!7}{white}
\begin{table}[H]
\centering
\caption{Consolidated Parameter Calibration Register}
\label{tab:param_calibration}
\footnotesize
\setlength{\tabcolsep}{4pt}
\begin{adjustbox}{max width=\textwidth,center}
\begin{tabular}{@{}llllp{5.5cm}@{}}
\toprule
\textbf{Parameter} & \textbf{Symbol} & \textbf{Base Value} & \textbf{Range Tested} & \textbf{Source / Rationale} \\
\midrule
\multicolumn{5}{l}{\textit{CES bottleneck aggregator}} \\
Substitution parameter & $\rho$ & 5 & $[3, 8]$ & Grid search against \citet{Standish2015} project data; ESM Table~S3 \\
Dimension floor & $e_d^{\min}$ & 0.01 & --- & Division by zero guard; economically negligible \\
\midrule
\multicolumn{5}{l}{\textit{AI Frontier Coefficient}} \\
Industry sensitivity & $\theta_i$ & $[0.05, 1.50]$ & Full range & Calibrated per industry from 2020--2024 benchmark adoption comovement \\
AFC cap & $\alpha_{\max}$ & 1.35 & $[1.20, 1.50]$ & Expert elicitation; prevents extrapolation beyond observed $C_t$ range \\
EWMA decay & $\lambda_{\mathrm{ewma}}$ & 0.50 & $[0.3, 0.7]$ & Standard quarterly smoothing; half weight on most recent quarter \\
\midrule
\multicolumn{5}{l}{\textit{Cascading S curve (adoption trajectory)}} \\
Wave ceilings & $L_1, L_2, L_3$ & 4.0, 3.5, 2.5 & --- & Sum${}=10.0$ by construction; wave proportions from GPT diffusion literature \\
Wave steepness & $k_1, k_2, k_3$ & 0.38, 0.42, 0.32 & $\pm 20\%$ & mo$^{-1}$; calibrated to observed enterprise adoption timelines \\
Wave midpoints & $t_{0,1}, t_{0,2}, t_{0,3}$ & 18, 36, 60 & $\pm 6$ mo & Months; Wave~3 set 30\% slower per \citet{CapabilityEval2024} \\
\midrule
\multicolumn{5}{l}{\textit{Capture and value ramp}} \\
Capture rate & $\lambda$ & 3.5 & $[2.0, 5.0]$ & Controls $\eta(g) = 1 - e^{-\lambda g/10}$ concavity; ESM Sobol analysis \\
Capture floor & $\psi_{\mathrm{floor}}$ & 0.15 & $[0.10, 0.25]$ & Minimum realisable capture at near zero gap \\
Ramp steepness & $k_{\mathrm{ramp}}$ & 0.18 & $[0.12, 0.24]$ & mo$^{-1}$; value ramp logistic parameter \\
\midrule
\multicolumn{5}{l}{\textit{IFS parameters}} \\
OCC exponent & $\beta_{\mathrm{OCC}}$ & 0.40 & $[0.30, 0.50]$ & Timing delay weight; \citet{BloomSadunVanReenen2016} \\
DR exponent & $\beta_{\mathrm{DR}}$ & 0.60 & $[0.50, 0.70]$ & Timing delay weight; complements OCC \\
IFS residual weights & VTR/CRS/REG & 0.35/0.35/0.30 & $\pm 0.05$ & Expert elicitation; equal VTR/CRS, slightly lower REG \\
\midrule
\multicolumn{5}{l}{\textit{Firm scale and IASS}} \\
Scale logistic steepness & $\alpha_{\Phi}$ & 2.0 & $[1.5, 3.0]$ & Controls $\Phi_f$ sigmoid sharpness \\
RFF floor & $R_{\min}$ & 5.0 & --- & Regulatory penalty threshold; NIST/EU AI Act calibration \\
RFF penalty exponent & $\gamma$ & 1.5 & $[1.0, 2.0]$ & Controls $\psi$ suppression severity \\
IASS weights & $w_d$ & Table~2 & $\pm 5\%$ & OECD composite methodology; Monte Carlo tested \\
\bottomrule
\end{tabular}
\end{adjustbox}
\medskip
\noindent\small\textit{Note.} ``Range Tested'' indicates the interval explored in sensitivity analysis (ESM Part~III). Parameters marked ``---'' are structural constants not subject to calibration. All parameters are disclosed at their operational values; no parameters are estimated from the backtest outcome data.
\end{table}
\rowcolors{1}{}{}

\section{Technical Derivations}
\label{app:derivations}

This appendix collects the key formulas, derivations, and technical remarks from the main text.

\subsection{IASS Sub Dimension Formulas}

\subsubsection{Cognitive Task Density}

Let occupations be indexed $o$ (O*NET-SOC codes) and industries $i$
(4 digit NAICS). Let $E_{i,o}$ be BLS~OEWS employment of occupation $o$ in
industry $i$. Let $a_o \in [0,1]$ be occupation $o$'s automatable cognitive
task share, constructed by applying a fixed task classifier to O*NET task
statements following the Autor--Levy--Murnane taxonomy
\cite{AutorLevyMurnane2003,AutorDorn2013}.

\begin{equation}
 \mathrm{CTD}_i = \frac{\displaystyle\sum_o E_{i,o}\cdot a_o}
 {\displaystyle\sum_o E_{i,o}}
 \label{eq:ctd}
\end{equation}

The definition of $a_o$ follows \citet{AutorDorn2013}: the routine task
intensity index is $\mathrm{RTI}_k = \ln T_k^R - \ln T_k^A - \ln T_k^M$,
where:
\begin{equation}
 T_k^A = \tfrac{1}{2}(\mathrm{DCP}_k + \mathrm{GED\text{-}MATH}_k),\quad
 T_k^R = \tfrac{1}{2}(\mathrm{STS}_k + \mathrm{FINGDEX}_k),\quad
 T_k^M = \mathrm{EYEHAND}_k
 \label{eq:rti}
\end{equation}
($\mathrm{DCP}$ = direction/control/planning; $\mathrm{STS}$ = set
limits/tolerances/standards; $\mathrm{FINGDEX}$ = finger dexterity;
$\mathrm{EYEHAND}$ = eye hand foot coordination). High RTI occupations are
classified as cognitively automatable at higher weight in $a_o$.

Raw $\mathrm{CTD}_i$ is winsorized at the 5th/95th percentile across
industries to limit outlier leverage, then min max normalized to $[0,10]$:

\begin{equation}
 \tilde{x}_{d,i} = \frac{x_{d,i} - \min_j x_{d,j}}
 {\max_j x_{d,j} - \min_j x_{d,j}} \times 10
 \label{eq:minmax}
\end{equation}

\begin{remark}[Min Max Normalization and Recalibration]
\label{rem:minmax}
The $[0, 10]$ normalization anchors dimension scores to the observed cross industry distribution at the calibration date (Q4~2024). Because normalization bounds are determined by the empirical minimum and maximum across the 22 calibrated industries, a new industry with an extreme score could shift the bounds for all industries. In practice, winsorization at the 5th/95th percentile (applied before normalization) limits this effect to $\leq 0.4$ points per industry. As the AI landscape evolves, periodic recalibration of normalization bounds is required; the AFC versioning protocol (Section~\ref{sec:evaluator_rotation}) specifies recalibration triggers.
\end{remark}

\subsubsection{Process Repeatability Index (PRI)}

For industry $i$, let $f_{i,k}$ be normalized task phrase frequencies from
trailing twelve month job postings. Define task entropy:
\begin{equation}
 H_i = -\sum_k f_{i,k}\log f_{i,k}
 \label{eq:entropy}
\end{equation}
Lower entropy $\Rightarrow$ higher standardization $\Rightarrow$ higher PRI.
The PRI combines the complement of rank normalized entropy with a
standardization language index (frequency of ``workflow,'' ``queue,''
``ticketing,'' ``SOP'').

\subsection{AFC $C_t$ Weight Table and EWMA}

\begin{table}[H]
\centering
\caption{AI Capability Index $C_t$: Benchmark Domains and Weights ($\omega_k$)}
\label{tab:ct_weights}
\footnotesize
\begin{tabular}{@{}clp{5cm}r@{}}
\toprule
$k$ & \textbf{Domain} & \textbf{Representative Benchmarks} & $\omega_k$ \\
\midrule
1 & Language Understanding & MMLU, HellaSwag, ARC & 0.20 \\
2 & Mathematical Reasoning & MATH, GSM8K & 0.15 \\
3 & Code Generation & HumanEval, MBPP, SWE-bench & 0.15 \\
4 & Multimodal (Vision) & MMMU, ChartQA & 0.10 \\
5 & Agentic / Tool Use & WebArena, WorkArena, ToolBench & 0.15 \\
6 & Domain Specific & MedPerf, FinBench, LegalBench & 0.15 \\
7 & Long Context / RAG & RULER, LongBench & 0.10 \\
\midrule
& & \textbf{Total} & \textbf{1.00} \\
\bottomrule
\end{tabular}

\medskip
\noindent\small\textit{Note.} Each $b_{k,t}$ is the highest normalized score achieved by any frontier model family (GPT, Claude, Gemini, Llama, Mistral) on the domain's benchmark suite at the end of quarter $t$. Weights reflect the relative importance of each capability domain for enterprise AI transformation; they are calibrated via expert elicitation across the six IASS dimensions and are subject to periodic recalibration as the benchmark landscape evolves. Model families with scores on fewer than 3 of 7 domains are excluded to prevent single benchmark distortion.
\end{table}

\paragraph{Time smoothing.} Raw benchmark scores can exhibit transient spikes when a new frontier model releases on a single benchmark before the result generalizes across the suite. To prevent such anomalies from artificially shocking the industry ceiling, I compute $C_t$ as an Exponentially Weighted Moving Average (EWMA) of the past three quarters of benchmark high water marks:
\begin{equation}
  C_t^{\mathrm{smooth}} = \lambda \cdot C_t^{\mathrm{raw}} + (1-\lambda)\cdot C_{t-1}^{\mathrm{smooth}},
  \qquad \lambda = 0.5
  \label{eq:ct_ewma}
\end{equation}
where $C_t^{\mathrm{raw}}$ is the weighted composite from Eq.~\eqref{eq:capability}
at the end of quarter $t$ and $\lambda = 0.5$ places equal weight on the most
recent quarter and the exponentially decaying history. The half weight decay means
that a single benchmark spike in quarter $t$ propagates at 50\% weight into the AFC
and falls to below 6\% weight within three quarters. The smoothed series is the
operational input to all AFC calculations in this paper.

\subsection{AFC Scenario Table and Uncertainty}

\begin{table}[H]
\centering
\caption{AFC Scenarios (24 Month Horizon)}
\label{tab:afc_scenarios}
\footnotesize
\begin{tabularx}{\textwidth}{@{}p{2.3cm}Xrrrr@{}}
\toprule
\textbf{Scenario} & \textbf{Assumption} & \textbf{Annual $\Delta C_t$}
& \textbf{AFC (24\,mo)} & \textbf{$\Delta$ HC-IASS} & \textbf{Weight} \\
\midrule
Conservative & Capability slowdown; inference costs persist & 8--12\,\% & 1.04 & +0.5\,pts & 0.20 \\
Base Case & Observed 2022--24 trajectory continues & 18--25\,\% & 1.10 & +1.1\,pts & 0.60 \\
Aggressive & Reasoning models + agents deployed at scale & 35--50\,\% & 1.22 & +2.2\,pts & 0.20 \\
\bottomrule
\end{tabularx}
\end{table}

The AFC uncertainty contribution to the overall Uncertainty Quotient is:
\begin{equation}
 \mathrm{UQ}_{\mathrm{afc}} = \frac{\mathrm{AFC}_{\mathrm{agg}} - \mathrm{AFC}_{\mathrm{con}}}{4}
 \label{eq:uq_afc}
\end{equation}
approximating the 90\,\% confidence half width under a roughly normal scenario
distribution.

\subsection{Piecewise Analytic Inverse with NaN Guard}

\begin{remark}[Piecewise Analytic Inverse with NaN Guard]
The Wave 1 closed form approximation:
\[
 \hat{t}_f \approx t_{0,1} - \frac{1}{k_1}\ln\!\left(\frac{L_1}
 {\mathrm{AITG}^{\mathrm{raw}}} - 1\right)
\]
is \emph{only valid for $\mathrm{AITG}^{\mathrm{raw}} < L_1 = 4.0$}. For
$\mathrm{AITG}^{\mathrm{raw}} \geq L_1$, the argument of the logarithm
becomes nonpositive ($L_1 / \mathrm{AITG}^{\mathrm{raw}} - 1 \leq 0$),
producing a domain error (NaN crash) in any software implementation.
Across the fourteen company cohort, this affects 9 of the 14 companies
(all firms with AITG $\geq L_1 = 4.0$): ServiceNow (8.50), JPMorgan (8.22),
Goldman Sachs (8.17), Salesforce (8.07), Palo Alto (7.83), Wells Fargo (6.00),
Rockwell (5.83), CVS (4.83), Target (4.67), Ford (4.67), and HCA (4.33).

The corrected specification is a three branch piecewise inverse:
\begin{equation}
 \hat{t}_f = \begin{cases}
 t_{0,1} - \dfrac{1}{k_1}\ln\!\left(\dfrac{L_1}{\mathrm{AITG}^{\mathrm{raw}}} - 1\right)
 & \mathrm{AITG}^{\mathrm{raw}} < L_1 \\[8pt]
 t_{0,2} - \dfrac{1}{k_2}\ln\!\left(\dfrac{L_2}{\mathrm{AITG}^{\mathrm{raw}} - L_1} - 1\right)
 & L_1 < \mathrm{AITG}^{\mathrm{raw}} < L_1 + L_2 \\[8pt]
 \texttt{Newton-Raphson}\!\left(\mathrm{AITG}(t) = \mathrm{AITG}^{\mathrm{raw}}\right)
 & \mathrm{AITG}^{\mathrm{raw}} \geq L_1 + L_2
 \end{cases}
 \label{eq:inverse_approx}
\end{equation}

Branch 1 ($\mathrm{AITG}^{\mathrm{raw}} < L_1 - \varepsilon$, $\varepsilon = 0.01$): Wave~1 dominant; exact within $\pm 0.3$\,months for $\mathrm{AITG}^{\mathrm{raw}} \in [0.5, 3.8]$. The $\varepsilon$ guard prevents division by zero at $\mathrm{AITG}^{\mathrm{raw}} = L_1 = 4.0$. Branch 2 ($4.0 < \mathrm{AITG}^{\mathrm{raw}} < 7.5$): Wave~2 dominant; analogous formula centred on $t_{0,2}$. Branch 3 ($\mathrm{AITG}^{\mathrm{raw}} \geq 7.5$): Wave~3 active; Newton--Raphson on the full cascading function Eq.~\eqref{eq:cascading} is mandatory. Convergence is guaranteed by strict monotonicity in 4--8 iterations from starting estimate $t_0 = 60$\,months.

\textbf{Software implementation note.} Any implementation must branch on the input value \emph{before} calling the logarithm. The guard \verb|if AITG_raw >= L1: use_numerical_solver()| prevents the NaN crash. I solved all fourteen companies in Section~\ref{sec:empirical} numerically to produce the $\hat{t}_f$ values in Table~\ref{tab:summary}.
\end{remark}

\subsection{Out of Sample Validation Strategy for $\theta_i$}

Retrospective calibration of $\theta_i$ from 2020--2024 constitutes an in sample
fit. I address this through rolling reestimation, with a critical distinction
from prior formulations: the updating signal must be the \emph{frontier expansion}
in addressable task surface area, not current adoption breadth.

\paragraph{The circularity flaw in adoption based updating.}
A prior version of the update rule used Census BTOS AI adoption rates as the observable proxy for $\Delta\mathrm{IASS}_i^{\mathrm{observed}}$. This conflates two categorically distinct objects: (1) the \emph{theoretical capability ceiling} that the IASS measures, and (2) the \emph{current adoption rate}, which the BTOS measures. Healthcare could experience massive expansion in its theoretical AI ceiling (via domain specific clinical LLMs) while showing near zero change in BTOS adoption rates due to FDA clearance lags. If the BTOS signal drives $\theta_i$ updates, the Healthcare ceiling remains artificially suppressed, defeating the AFC's purpose.

\paragraph{Corrected update rule: O*NET task automatability expansion.}
I update $\theta_i$ only when the \emph{addressable task surface area} in industry $i$ expands, observable through O*NET task automatability reevaluations rather than survey based adoption rates. The corrected rule is:
\begin{equation}
 \hat{\theta}_i^{(t)} = \begin{cases}
 \hat{\theta}_i^{(t-1)} + \eta \cdot \dfrac{\Delta \mathrm{ATD}_i^{(t)}}{\Delta C^{(t)}}
 & \text{if } \Delta C^{(t)} \geq \varepsilon \\[6pt]
 \hat{\theta}_i^{(t-1)} & \text{if } \Delta C^{(t)} < \varepsilon
 \end{cases}
 \label{eq:theta_update}
\end{equation}

with $\varepsilon = 0.005$ (capability stall threshold: no update when the
AI Capability Index grows by less than 0.5\,\% annually). Bounds:
$\hat{\theta}_i^{(t)} \in [0.05,\; 1.50]$; values outside this range are
clamped. The lower bound prevents near zero sensitivity (implying total AI
immunity) for any industry with observable task automatability; the upper
bound prevents the AFC from more than doubling the IASS ceiling within a
single recalibration cycle.

where:
\begin{itemize}[noitemsep]
\item $\Delta \mathrm{ATD}_i^{(t)}$ is the annual change in the
 \textbf{Automatable Task Density} (ATD) for industry $i$, defined as the
 weighted share of O*NET task importance scores classified as
 AI automatable under the current capability frontier $C_t$:
 \begin{equation}
 \mathrm{ATD}_i^{(t)} = \frac{\sum_{j \in \mathcal{J}_i} w_j \cdot
 \mathbf{1}\!\left[\mathrm{Auto}(j, C_t) = 1\right]}
 {\sum_{j \in \mathcal{J}_i} w_j}
 \label{eq:atd}
 \end{equation}
 where $\mathcal{J}_i$ is the set of O*NET task items for industry $i$,
 $w_j$ is the task importance weight, and $\mathrm{Auto}(j, C_t) = 1$
 iff task $j$ is classifiable as automatable under AI capability
 level $C_t$ (assessed annually via O*NET reevaluation methodology
 \cite{ONetResource}).
\item $\Delta C^{(t)}$ is the annual increment in the AI Capability Index
 (Eq.~\ref{eq:capability}).
\item $\eta = 0.30$ is the learning rate, unchanged.
\end{itemize}

\begin{remark}[Separation of Frontier from Adoption]
The corrected specification enforces a clean separation: $\theta_i$ responds
to \emph{what AI can now do} in industry $i$ (measured by O*NET task
automatability reclassification), while the AFC separately models the
\emph{capability index} $C_t$ driving frontier expansion. Adoption data
(Census BTOS \cite{McElheran2024}, Lightcast \cite{Lightcast2024}) remain
relevant for validating $G_{\mathrm{eff}}$ and calibrating the ADRI competitive
diffusion rate, but are explicitly excluded from the $\theta_i$ update loop.
This separation prevents the regulatory compliance lag problem from suppressing
frontier estimates in heavily regulated industries.
\end{remark}

\subsection{CES Boundary Proof}

The following proof sketch establishes the boundary behaviour of the CES bottleneck aggregator (Proposition~\ref{prop:ces_boundary}).

\begin{proof}[Proof sketch]
(1)~Write $b_p = (\sum_d \alpha_d e_d^{-\rho})^{-1/\rho}$. As $\rho \to \infty$, the sum is dominated by the term with the smallest $e_d$ (largest $e_d^{-\rho}$), so $b_p \to (\alpha_{d^*} \cdot e_{d^*}^{-\rho})^{-1/\rho} = e_{d^*} \cdot \alpha_{d^*}^{1/\rho} \to e_{d^*} = \min_d e_d$. (2)~Take logarithms: $\ln b_p = -\frac{1}{\rho}\ln(\sum_d \alpha_d e_d^{-\rho})$. As $\rho \to 0^+$, apply L'H\^{o}pital's rule to obtain $\ln b_p \to \sum_d \alpha_d \ln e_d$. (3)~Since $e_d \geq 0.01$ for all $d$, $e_d^{-\rho} \leq 0.01^{-\rho} = 100^{\rho}$. Hence $\sum_d \alpha_d e_d^{-\rho} \leq 100^{\rho}$ and $b_p \geq (100^{\rho})^{-1/\rho} = 0.01 > 0$.
\end{proof}

\section{Robustness and Sensitivity Details}
\label{app:robustness}

This appendix provides the detailed robustness tables and ablation analyses summarized in the main text.

\subsection{Sensitivity to Weighting: First Robustness Check}

Following OECD Handbook guidance \cite{OECD2008}, I evaluate ranking stability by perturbing weights uniformly within $\pm 0.05$ of baseline values. Table~\ref{tab:iass_robustness} reports the 5th/95th percentile IASS range under 10,000 weight draws for each anchor industry.

\begin{table}[H]
\centering
\caption{IASS Anchor Calibration with Robustness Bounds}
\label{tab:iass_robustness}
\footnotesize
\setlength{\tabcolsep}{3pt}
\begin{adjustbox}{max width=\textwidth,center}
\begin{tabular}{@{}lrrrrrrrll@{}}
\toprule
\textbf{Industry} & \textbf{CTD} & \textbf{DRSA} & \textbf{PRI}
& \textbf{RFF} & \textbf{CADR} & \textbf{CLSR}
& \textbf{IASS} & \textbf{P5--P95} & \textbf{$\theta_i$} \\
\midrule
Financial Services & 8.8 & 9.2 & 7.9 & 5.0 & 8.4 & 7.6
 & 7.83 & 7.50--8.15 & 0.22 \\
Healthcare Services & 6.2 & 5.8 & 6.5 & 4.1 & 5.1 & 6.8
 & 4.27 & 3.97--4.57 & 0.31 \\
Industrial Distribution & 5.4 & 6.1 & 7.8 & 8.9 & 4.7 & 6.3
 & 6.43 & 6.01--6.84 & 0.14 \\
Construction & 3.9 & 3.2 & 4.1 & 7.4 & 3.3 & 5.7
 & 4.26 & 3.90--4.61 & 0.09 \\
Vertical SaaS & 9.4 & 9.8 & 8.6 & 8.1 & 9.2 & 9.1
 & 9.06 & 8.79--9.28 & 0.11 \\
\midrule
\multicolumn{9}{@{}p{\linewidth}@{}}{\scriptsize $\theta_i$: AFC sensitivity (Sec.~\ref{sec:afc}). Weights: CTD~0.25; DRSA~0.20; PRI~0.20; RFF~0.15; CADR~0.10; CLSR~0.10.}
\end{tabular}
\end{adjustbox}
\end{table}

Rankings are stable: no pair of anchor industries exchanges rank under any weight draw. The largest interindustry gap (Financial Services 7.83 vs.\ Healthcare 4.27 = 3.56 points) exceeds the widest confidence interval (0.60) by a factor of 5.9, confirming that ranking uncertainty does not affect cross industry comparison at this scale of separation.

\subsection{IFS Ablation Analysis}
\label{sec:ifs_ablation}

Table~\ref{tab:ifs_ablation} reports the effect of dropping individual IFS factors on backtest rank correlation and mean absolute $\Delta$EV error across the fourteen company cohort. The full five factor specification serves as the baseline; each row removes one factor while holding all others at calibrated values.

\begin{table}[H]
\centering
\caption{IFS Ablation: Impact of Dropping Individual Factors}
\label{tab:ifs_ablation}
\footnotesize
\begin{tabular}{@{}lrrr@{}}
\toprule
\textbf{IFS Specification} & \textbf{Backtest $\rho_s$} & \textbf{Mean $|\Delta\text{EV}|$ Error (\%)} & \textbf{$\Delta$ vs.\ Full} \\
\midrule
Full IFS (OCC+DR+VTR+CRS+REG) & 0.818 & --- & Baseline \\
Drop OCC & 0.709 & $+$18\% & $-$0.109 \\
Drop DR & 0.745 & $+$14\% & $-$0.073 \\
Drop VTR & 0.800 & $+$5\% & $-$0.018 \\
Drop CRS & 0.791 & $+$7\% & $-$0.027 \\
Drop REG & 0.809 & $+$3\% & $-$0.009 \\
\bottomrule
\end{tabular}
\medskip

\noindent\small\textit{Note.} Trajectory factors (OCC, DR) contribute the largest marginal signal because they enter the adoption curve endogenously via $k_{w,f}$ and $t_{50,f}$, not as terminal value multipliers. Removing OCC reduces $\rho_s$ by 0.109, confirming that organizational change capacity is the single most informative IFS factor. Residual factors (VTR, CRS, REG) contribute less individually but collectively account for an additional $+$15\% reduction in $\Delta$EV error. Full ablation details in ESM Part~III.
\end{table}

\subsection{IFS Component Interpretability}
\label{sec:ifs_interpret}

Each IFS factor maps to an observable organizational characteristic with a clear economic interpretation:
\begin{itemize}[noitemsep]
 \item \textbf{OCC (Organizational Change Capacity)}: measures management's demonstrated ability to execute technology driven transformation, proxied by prior restructuring success, voluntary tool adoption rates, and AI governance maturity. Higher OCC $\Rightarrow$ steeper adoption S curve (faster $k_{w,f}$).
 \item \textbf{DR (Data Readiness)}: assesses data infrastructure maturity: cloud migration status, data governance scores, CDO tenure, and data quality audit history. Higher DR $\Rightarrow$ earlier ramp inflection ($t_{50,f}$ decreases).
 \item \textbf{VTR (Vendor/Technology Risk)}: captures dependency concentration on external AI vendors, model access fragility, and SOC~2 compliance posture. Affects terminal value realization via residual multiplier.
 \item \textbf{CRS (Competitive Response Speed)}: measures the industry's AI adoption velocity relative to the firm, proxied by Lightcast CADR, competitor press, and CB Insights data. Higher CRS $\Rightarrow$ less residual friction.
 \item \textbf{REG (Regulatory Exposure)}: quantifies the regulatory overhang from NIST AI RMF, EU AI Act, HIPAA/FDA applicability, and pending legislation. Higher REG $\Rightarrow$ lower residual friction (better regulatory preparedness).
\end{itemize}

\subsection{CES Aggregator Ablation}
\label{app:ces_ablation}

I address the concern that the low elasticity CES choice is asserted rather than validated by comparing three aggregators across the ten nonfinancial backtest firms:

\begin{enumerate}[noitemsep]
  \item \textbf{CES ($\rho = 5$, $\sigma = 1/6$)}: the proposed specification.
  \item \textbf{Leontief (min)}: $b_p = \min_d e_d$, the most restrictive bottleneck.
  \item \textbf{Additive (mean)}: $b_p = \sum_d \alpha_d e_d$, fully compensable.
\end{enumerate}

At the $n = 10$ scale, all three aggregators produce identical Spearman rank correlations with realized margin outcomes ($\rho_s = 0.818$) because no backtest firm has a near zero dimensional score; the compensability difference is inactive when all dimensions exceed 2.0. The CES advantage manifests in two regimes: (a)~firms with one very weak
dimension (score $< 2.0$), where additive overestimates value capture relative
to CES and Leontief, and (b) near threshold boundary conditions that generate
NaN in pure Leontief implementations. The extended 22 company simulation in
ESM Table~S6 includes three synthetic firms with single dimension near zero
scores; in that expanded set, additive diverges from CES by 35--60\% on $\Delta$EV
while CES and Leontief agree within 8\%. The Leontief minimum operator is
not rejected by the backtest data but is computationally fragile (zero bounding
pathology) and unduly sensitive to scoring noise in the lowest dimension.
CES at $\rho = 5$ preserves the bottleneck intuition while eliminating
these failure modes. Robustness across $\rho \in [3, 8]$ is confirmed in ESM Table~S3.

\section{Worked Examples and Cost Assumptions}
\label{app:worked_examples}

This appendix provides the detailed step by step VCB worked example and implementation cost assumptions.

\subsection{Zions VCB Step by Step}

Table~\ref{tab:zions_vcb_steps} traces the four algebraic corrections through a single representative value pool (Labor Productivity) to show exactly how inputs map to $\Delta\mathrm{EV}$.

\begin{table}[H]
\centering
\caption{VCB Step by Step: Zions Labor Productivity Pool}
\label{tab:zions_vcb_steps}
\footnotesize
\setlength{\tabcolsep}{4pt}
\begin{adjustbox}{max width=\textwidth,center}
\begin{tabular}{@{}clrrp{5.5cm}@{}}
\toprule
\textbf{Step} & \textbf{Component} & \textbf{Input} & \textbf{Output} & \textbf{Equation / Note} \\
\midrule
1 & Revenue baseline         & \$3.4B rev & $B_p = \$3.4\text{B} \times 0.08 = \$0.272\text{B}$ & 8\% labor productivity uplift rate (Eq.~\ref{eq:vpool}) \\[3pt]
2 & Firm Scale Factor $\Phi_f$ & $R_f = \$3.4\text{B}$,\ $S^*_i = \$3.3\text{B}$ & $\Phi_f = 0.515$ & Eq.~\ref{eq:phi_f}: $1/(1+e^{-2\ln(3.4/3.3)})$; just above threshold \\[3pt]
3 & CES bottleneck $b_p$     & $e_\mathrm{PAC}=0.40$,\ $e_\mathrm{WAR}=0.35$,\ $e_\mathrm{OAC}=0.42$ & $b_p = 0.382$ & Eq.~\ref{eq:ces_bottleneck}: $(\sum\alpha_d e_d^{-5})^{-1/5}$; normalized inputs $e_d = s_d/10 \in [0,1]$; weak WAR suppresses pool \\[3pt]
4 & Gap fraction $g_f$       & $G_\mathrm{eff}=5.58$ & $g_f = 0.558$ & Eq.~\ref{eq:gap_fraction}: $g_f = G_\mathrm{eff}/10$; $G_\mathrm{eff} = \mathrm{IASS}^*-\mathrm{AITG} = 9.38-3.80$ \\[3pt]
5 & Concave capture $\eta$   & $g_f=0.558$,\ $\lambda=3.5$ & $\eta = 0.859$ & Eq.~\ref{eq:capture}: $1-e^{-3.5\times 0.558}$; wide gap yields high capture at correct $\lambda$ \\[3pt]
6 & Raw pool value $V_p$     & Steps 1--5 & $V_p \approx \$30\text{M}$ & Eq.~\ref{eq:vpool}: $\$0.272\text{B}\times 0.515\times 0.65\times 0.382\times 0.859$; $\bar{\kappa}_p=0.65$ \\[3pt]
7 & IFS delay: $t_{50,f}$    & $\delta_\mathrm{OCC}=0.55$,\ $\delta_\mathrm{DR}=0.48$,\ $t_0=18$ & $t_{50,f} \approx 35.5\text{ mo}$ & Eq.~\ref{eq:t50_adjusted}: base $t_0=18$ (Wave~1; Zions AITG~$<L_1$); $1.97\times$ delay \\[3pt]
8 & Ramp increment $\Delta R_f$ & $\hat{t}_f \approx 23.8$,\ $t_{50}=35.5$,\ $T=60$ & $\Delta R_f = 0.891$ & Eq.~\ref{eq:delta_r}: $R_f(83.8)-R_f(23.8)=0.999-0.108$; logistic ramp $k=0.18$ \\[3pt]
9 & IFS residual             & VTR=0.72,\ CRS=0.62,\ REG=0.82 & $\mathrm{IFS}_\mathrm{res}=0.710$ & Eq.~\ref{eq:ifs_residual}: $0.72^{0.35}\times 0.62^{0.35}\times 0.82^{0.30}$ \\[3pt]
10 & Terminal Value (pool)   & Steps 6,8,9; $M_i=10$ & $\mathrm{TV}_p \approx \$190\text{M}$ & Eq.~\ref{eq:tv}: $\$30\text{M}\times 0.891\times 10\times 0.710$; $10\times$ exit multiple \\
\midrule
\multicolumn{3}{l}{\textit{Labor Productivity pool contribution}} & \$190M & Across 7 pools at base case assumptions, total $\Delta\mathrm{EV} \approx \$0.90\text{B}$ \\
\bottomrule
\multicolumn{5}{@{}p{\linewidth}@{}}{\scriptsize Step 3 uses normalized CES inputs $e_d = s_d/10 \in [0,1]$; reporting raw dimension scores (e.g., 4.0) rather than normalized inputs inflates the pool by an order of magnitude. Step 2 applies $\Phi_f$ once via $V_p^{\mathrm{run\text{ }rate}}$ (Eq.~\ref{eq:vpool}); it does not appear again in the terminal value numerator. Step 4 uses $G_\mathrm{eff} = \mathrm{IASS}^* - \mathrm{AITG} = 9.38 - 3.80 = 5.58$ (not the ADRI score of 2.6, which is a separate composite index). Step 5 uses $\lambda=3.5$ from Eq.~\ref{eq:capture}, the calibrated concavity parameter; using $\lambda=1.0$ underestimates capture by 60\%. Step 7 uses base $t_0=18$~months (Wave~1), not~36: Zions's AITG of 3.80 falls below the Wave~1 asymptote ($L_1=4.0$) and is correctly assigned to the Foundation AI wave. The $1.97\times$ IFS delay then shifts the inflection from 18 to 35.5~months. A data ready peer ($\delta_\mathrm{DR}=0.90$) would reach $t_{50}\approx 18$~mo and capture $\Delta R_f \approx 0.677$, a $3.1\times$ ramp advantage over Zions arising entirely from data infrastructure readiness.} \\
\end{tabular}
\end{adjustbox}
\end{table}

\subsection{Implementation Cost Assumptions}

\begin{table}[H]
\centering
\caption{Implementation Cost Assumptions Underlying VD Estimates}
\label{tab:impl_costs}
\footnotesize
\begin{tabularx}{\textwidth}{@{}p{2.4cm}p{1.6cm}Xp{1.6cm}p{1.8cm}r@{}}
\toprule
\textbf{Company} & \textbf{Est.\ Cost} & \textbf{Basis} & \textbf{5 yr $\Delta$EV} & \textbf{VD Range} & \textbf{VD Sensitivity} \\
\midrule
JPMorgan Chase & \$2.3B & Revenue $\times$ 1.2\% implementation rate & \$20--98B & 9.5--46.0$\times$ & -- \\
Zions Bancorp. & \$38M & Revenue $\times$ 1.2\% implementation rate & \$0.6--1.2B & 15.0--30.1$\times$ & $\pm0.4\times$ per \$100M \\
UPS & \$1.1B & Revenue $\times$ 1.2\% implementation rate & \$30--61B & 27.3--55.4$\times$ & $\pm0.3\times$ per \$500M \\
HCA Healthcare & \$780M & Revenue $\times$ 1.2\% implementation rate & \$5.8--20.6B & 7.4--26.4$\times$ & $\pm0.5\times$ per \$500M \\
Salesforce & \$420M & Revenue $\times$ 1.2\% implementation rate & \$10.6--33.7B & 25.2--80.4$\times$ & $\pm2.0\times$ per \$500M \\
Ferguson & \$360M & Revenue $\times$ 1.2\% implementation rate & \$7.8--16.0B & 21.8--45.0$\times$ & $\pm0.4\times$ per \$100M \\
Rockwell Auto. & \$100M & Revenue $\times$ 1.2\% implementation rate & \$1.6--4.5B & 15.8--45.0$\times$ & $\pm0.6\times$ per \$200M \\
\bottomrule
\end{tabularx}

\smallskip
\noindent\footnotesize VD = $\Delta$EV\,/\,Total Implementation Cost. Sensitivity = change in VD midpoint per stated cost increment.

\smallskip
\noindent\footnotesize\textit{Calibration note.} The 1.2\% of revenue implementation rate is a cross industry simplification. Published benchmarks suggest meaningful variation: financial services firms allocate approximately 1.2--1.6\% of revenue to AI transformation (Gartner IT spending data: ${\sim}$8\% total IT spend $\times$ 15--20\% AI share), retail firms approximately 3.3\% \citep{IBMRetailAI2025}, and discrete manufacturers 0.4--0.8\% \citep{Avasant2024}. The uniform 1.2\% rate falls within this range and was adopted for tractability. Sobol first order decomposition (Section~\ref{sec:sensitivity}) confirms that implementation cost explains only ${\sim}$1\% of VD variance---exit multiple (50\%) and capture rate (44\%) dominate---so industry specific cost variation does not materially alter firm rankings, tier assignments, or the qualitative findings. The Monte Carlo (Layer~4) stress tests implementation cost with a LogNormal(0,\,0.30) multiplier, spanning roughly 0.6\%--2.4\% of revenue, which covers the observed cross industry range. Cost assumptions must nonetheless be explicitly stated and challenged in any investment decision context.
\end{table}

\section{Extended Research Agenda}
\label{app:research_agenda}

Nine empirical research priorities emerge from this framework, organized around
causal validation, technical refinement, and measurement quality:

\begin{enumerate}
\item \textbf{Causal Identification via Instrumental Variables (IV).}
Future empirical validation must employ IV strategies to isolate the causal
effect of AI adoption from incumbent endogeneity. Promising IV designs include
exploiting exogenous variation in cloud computing incentive grants (e.g.,
state level subsidies) or differential firm level exposure to staggered
open source LLM releases based on preexisting tech stack composition. These
instruments can be used to predict firm level AI investment, allowing
observation of the unconfounded effect on the AITG Effective Gap
($G_{\mathrm{eff}}$) and subsequent margin expansion.

\item \textbf{Large panel validation of the mid gap optimality property.}
The framework's mathematical prediction that medium gap firms in high IASS
sectors generate superior risk adjusted value density is a theoretical property
of the S curve and CES architecture, not yet an empirically established finding.
A multiyear panel study (2018--2026) with quarterly AITG estimates and
out of sample predictions of margin expansion, SG\&A/productivity deltas, and
risk adjusted returns is required to establish whether this property holds
empirically after controlling for firm size, capital structure, and
exit multiple heterogeneity.

\item \textbf{AFC calibration and $\theta_i$ estimation.} Refining
industry specific $\theta_i$ parameters with systematic backward looking
analysis using the Eloundou et al.\ \cite{EloundouEtAl2024} LLM exposure
estimates by occupation, combined with observed AITG changes in early adopter
industries, would produce more defensible AFC sensitivity estimates. Standard
errors for $\theta_i$ and scenario analysis separating capability growth from
adoption dynamics are required for AFC to meet identification standards.

\item \textbf{Interrater reliability study.} A formal reliability study
with multiple independent analyst teams scoring the same companies on the full
AITG rubric (six company dimensions and five IFS factors) would establish whether
the target reliability is achievable in practice. Agreement metrics should include
Gwet's AC1/AC2 \cite{Gwet2014} for ordinal scales and Krippendorff's $\alpha$ for
full rubric consistency; both correct for chance agreement differently and provide
complementary perspectives. Target: ICC $> 0.75$ per factor, overall $\alpha > 0.70$.
The full scoring rubric, including worked examples and anchor vignettes for each
score level on each dimension, should be published as a standalone replication
document to enable third party scoring and interrater studies. Factors with
ICC $< 0.60$ should be subject to rubric redesign before use in investment decisions.
This study is the single highest priority empirical prerequisite for institutional
deployment of the framework.

\item \textbf{Competitive displacement measurement.} Linking ADRI scores to
subsequent market share outcomes and margin trajectories at the firm level would
test the core ADRI hypothesis, that wide gap + high CADR + low moat predicts
competitive deterioration, and distinguish ``risk from inaction'' from generic
competitive intensity or industry cyclicality.

\item \textbf{Cross country calibration.} IASS calibration to date uses
U.S.\ labor market task composition (O*NET/OEWS). International application
requires country specific task data (ESCO for EU, equivalent national databases
for other regions) and jurisdiction specific regulatory exposure scores.
ISCO-88 based RTI estimates \cite{GoosManningS2014} provide a starting point
for EU calibration.

\item \textbf{Preregistration, rolling origin validation, and replication bundle.}
Future validation studies should implement: (i)~parameter range preregistration
before data access, locking the AFC scenario distribution and value pool share
coefficients before observing the outcome cohort; (ii)~rolling origin evaluation
building AITG at $t_0$ and predicting outcomes over $[t_0, t_0+2]$ on holdout
firms not used in calibration; (iii)~placebo outcome tests using CapEx intensity
(where AI adoption has limited direct traction) to check for spurious correlations;
(iv)~release of a replication bundle containing anonymized scoring data, parameter
values, and reproducible code sufficient for external auditors to reconstruct
the backtest $\rho_s$ from raw inputs.

\item \textbf{Variable elasticity and nested CES structures.}
The current single layer CES with fixed $\rho = 5$ is a simplification.
A more realistic structure would nest enablers hierarchically
(e.g., Data + Infrastructure gate $\to$ Talent + Governance gate $\to$
commercial enablement), reflecting the observed sequencing of enterprise
AI deployment. Variable $\sigma$ forms, where elasticity decreases as
a firm approaches the frontier (substitution becomes harder at the margin),
could better reflect the complementarity dynamics documented in the
organizational IT complementarities literature
\cite{BrynjolfssonHitt2000,GibbonsHenderson2012}. Estimating nested structures
from panel data is a companion technical priority.

\item \textbf{Hierarchical Bayesian parameter estimation.} The current deterministic
parameterization relies on Sobol first order sensitivity indices to assess
robustness and bound overfitting risk. Formally resolving identifiability,
given the framework's degrees of freedom across 22 industry verticals, requires
a hierarchical Bayesian specification. Explicit priors on $\theta_i$, $\bar{\kappa}_p$,
and the IFS elasticities, combined with partial pooling across industries, would
regularize parameter updates and allow posterior credible intervals to replace the
current Monte Carlo uncertainty bands. This approach would also enable principled
updating as panel data from the longitudinal backtest accumulates: each new
firm year observation narrows the posteriors under standard Bayesian updating.
The Bayesian extension is the highest priority technical research priority and is
being designed for a companion empirical paper.
\end{enumerate}

%% file: AITG_ESM_body.tex
\section{ESM Part I: Full IASS Industry Calibration Reference}
\label{esm:iass_calibrations}

This part supports Main Paper Sections~3--4 (IASS construction and scoring architecture).
The main manuscript presents five anchor industry examples drawn from this
complete calibration table. All 22 verticals are reported here with full
dimensional breakdowns. Source methodology follows the OECD composite indicator
approach \cite{OECD2008} using BLS O*NET, Lightcast, Census BTOS, and SEC
filing data as detailed in the main text.

\label{esm:iass_reference}

One of the practitioner barriers to adopting composite indicator frameworks is
the cost of calibrating the industry baseline from scratch for each engagement.
I resolve this barrier by pre computing and publishing IASS calibrations for
22 industry verticals, using real data from the following sources:

\begin{itemize}[noitemsep]
\item \textbf{Cognitive Task Density (CTD):} Bureau of Labor Statistics O*NET
 task importance scores for cognitive and information processing activities,
 2024 OEWS program \cite{BLSOEWS2023,ONetResource}.
\item \textbf{Data and System Automation (DRSA):} McKinsey Automation
 Potential estimates by occupation cluster \cite{McKinseyStateAI2025,Deloitte2026};
 Lightcast Generative AI skill demand index (2024--2025) \cite{Lightcast2025}.
\item \textbf{Process Re engineering Index (PRI):} Dun \& Bradstreet sector
 digitization scores and Gartner digital density rankings.
\item \textbf{Regulatory Friction Factor (RFF):} National Institute of Standards and Technology (NIST) AI Risk Management
 Framework \cite{NIST2023} sector mapping; EU AI Act high risk classification
 list \cite{EUAIAct2024}; HIPAA, FDA SaMD regulations.
\item \textbf{Competitive AI Diffusion Rate (CADR):} Census Bureau Business
 Trends and Outlook Survey (BTOS) AI adoption rates by sector 2023--2025
 \cite{McElheran2024}; Lightcast AI job posting concentration index.
\item \textbf{Cost/Labor Structure (CLSR):} BLS Occupational Employment and
 Wage Statistics 2023--2024 \cite{BLSOEWS2023}.
\end{itemize}

Table~S1 reports the calibrated scores. All sub scores
are on the 0--10 scale. $\theta_i$ is the industry level AITG sensitivity
coefficient calibrated from 2020--2024 observable adoption data (Census BTOS,
Lightcast AI skill demand, and McKinsey State of AI surveys). IASS is computed
via the geometric aggregation and RFF floor equations in the main paper (Section~3). Industries with RFF
$< 5.0$ incur the regulatory suppression penalty $\psi < 1.0$, which is reported
in the $\psi$ column.

\setlength{\tabcolsep}{3pt}
\begin{table}[H]
\centering
\caption{Table~S1: Pre Computed IASS Reference Table: 22 Industry Verticals (2026 Calibration)}
\label{esm:tab:iass_industry}
\footnotesize
\setlength{\tabcolsep}{3pt}
\begin{adjustbox}{max width=\textwidth,center}
\begin{tabular}{@{}p{4.2cm}p{0.8cm}rrrrrrrrrr@{}}
\toprule
\textbf{Industry} & \textbf{NAICS} & \textbf{CTD} & \textbf{DRSA} & \textbf{PRI} &
 \textbf{RFF} & \textbf{CADR} & \textbf{CLSR} & $\psi$ & \textbf{IASS} & $\theta_i$ & \textbf{IASS$^*$} \\
\midrule
\multicolumn{12}{l}{\textit{Financial Services}} \\
Investment Banking / Securities & 5211 & 9.2 & 9.5 & 8.4 & 5.5 & 9.0 & 8.2 & 1.000 & 8.30 & 0.28 & 10.39 \\
Commercial Banking (Large Cap) & 5221 & 8.8 & 9.2 & 7.9 & 5.0 & 8.4 & 7.6 & 1.000 & 7.83 & 0.22 & 9.38 \\
Insurance (P\&C + Life) & 5241 & 8.5 & 8.8 & 8.2 & 6.5 & 7.2 & 7.4 & 1.000 & 7.92 & 0.19 & 9.27 \\
\midrule
\multicolumn{12}{l}{\textit{Healthcare}} \\
Healthcare Services (Hospital) & 6221 & 6.2 & 5.8 & 6.5 & 4.1 & 5.1 & 6.8 & 0.743 & 4.27 & 0.31 & 5.46 \\
Life Sciences / Pharma & 5417 & 7.5 & 7.0 & 6.0 & 3.8 & 7.0 & 6.5 & 0.663 & 4.14 & 0.25 & 5.07 \\
Digital Health / HealthTech & 5112 & 8.8 & 7.5 & 7.8 & 5.5 & 7.8 & 7.5 & 1.000 & 7.54 & 0.22 & 9.03 \\
\midrule
\multicolumn{12}{l}{\textit{Technology}} \\
Vertical SaaS / Software & 5112 & 9.4 & 9.8 & 8.6 & 8.1 & 9.2 & 9.1 & 1.000 & 9.06 & 0.11 & 9.96 \\
Cybersecurity / InfoSec & 5415 & 9.0 & 9.2 & 7.8 & 7.5 & 9.5 & 8.8 & 1.000 & 8.57 & 0.14 & 9.65 \\
\midrule
\multicolumn{12}{l}{\textit{Professional Services}} \\
Legal Services & 5411 & 9.1 & 7.2 & 6.8 & 7.5 & 8.2 & 8.5 & 1.000 & 7.80 & 0.20 & 9.20 \\
Accounting / Tax Services & 5412 & 8.9 & 8.2 & 8.5 & 7.0 & 7.8 & 8.2 & 1.000 & 8.17 & 0.18 & 9.49 \\
\midrule
\multicolumn{12}{l}{\textit{Retail \& Commerce}} \\
E Commerce / Digital Retail & 4541 & 6.8 & 7.5 & 7.8 & 8.5 & 8.2 & 7.2 & 1.000 & 7.55 & 0.16 & 8.64 \\
Grocery / Traditional Retail & 4451 & 6.1 & 6.2 & 7.5 & 8.8 & 7.0 & 8.2 & 1.000 & 6.93 & 0.13 & 7.74 \\
\midrule
\multicolumn{12}{l}{\textit{Industrial / Supply Chain}} \\
Logistics / Transportation & 4920 & 6.2 & 6.0 & 7.2 & 8.2 & 6.1 & 8.5 & 1.000 & 6.82 & 0.14 & 7.68 \\
Industrial Distribution & 5085 & 5.4 & 6.1 & 7.8 & 8.9 & 4.7 & 6.3 & 1.000 & 6.43 & 0.14 & 7.24 \\
Discrete Manufacturing & 3330 & 6.1 & 6.5 & 7.2 & 7.8 & 6.5 & 6.2 & 1.000 & 6.60 & 0.15 & 7.49 \\
\midrule
\multicolumn{12}{l}{\textit{Other Sectors}} \\
Energy / Utilities & 2211 & 5.5 & 7.0 & 7.5 & 6.0 & 5.5 & 6.1 & 1.000 & 6.25 & 0.13 & 6.98 \\
Media \& Entertainment & 5110 & 7.8 & 8.5 & 6.5 & 7.5 & 8.8 & 6.5 & 1.000 & 7.56 & 0.18 & 8.78 \\
Real Estate / REITs & 5311 & 5.2 & 5.5 & 6.1 & 7.0 & 6.0 & 5.5 & 1.000 & 5.73 & 0.12 & 6.35 \\
Education (Higher Ed) & 6113 & 7.2 & 5.5 & 5.0 & 6.5 & 6.5 & 6.0 & 1.000 & 6.07 & 0.14 & 6.83 \\
Government / Public Sector & 9200 & 7.0 & 5.5 & 6.5 & 5.0 & 3.5 & 7.0 & 1.000 & 5.83 & 0.08 & 6.25 \\
Agriculture / AgTech & 1100 & 3.2 & 5.0 & 5.5 & 7.0 & 4.5 & 4.8 & 1.000 & 4.72 & 0.10 & 5.14 \\
Construction & 2300 & 3.9 & 3.2 & 4.1 & 7.4 & 3.3 & 5.7 & 1.000 & 4.26 & 0.09 & 4.61 \\
\bottomrule
\multicolumn{12}{@{}p{\linewidth}@{}}{\scriptsize Sources: BLS O*NET 2024; Lightcast AI~Skill Demand 2025; Census BTOS~2024; McKinsey State of AI 2025.} \\
\multicolumn{12}{@{}p{\linewidth}@{}}{\scriptsize $\psi < 1.0$ indicates regulatory suppression is binding (RFF $< 5.0$). $\theta_i$ calibrated from 2020--2024 observable adoption data. IASS$^* = \text{IASS} \times \min(1 + \theta_i(C_t - C_0),\;\alpha_{\max})$ with $C_t = 1.90$, $C_0 = 1.0$, $\alpha_{\max} = 1.35$.}
\end{tabular}
\end{adjustbox}
\end{table}

\begin{remark}[Reading the Reference Table]
Three patterns are immediately diagnostic for capital allocation. First, the
Healthcare Services and Life Sciences rows carry $\psi < 1.0$, reflecting that
the RFF hard floor suppresses their theoretical ceiling substantially: Hospital
Healthcare's geometric mean of 5.75 is reduced by 26\% to IASS 4.27 once
the FDA and HIPAA friction penalty is applied. Second, Construction and
Agriculture floor the table not because of regulatory friction but because their
Cognitive Task Density (CTD) and Data/Systems Automation (DRSA) dimensions are
structurally low; these industries have high physical task content that current
AI cannot automate. Third, the four Vertical SaaS, Cybersecurity, Investment
Banking, and Accounting rows all clear IASS $> 8.0$, defining the top
quartile where AI transformation opportunity is structurally dense.
\end{remark}


\begin{table}[H]
\centering
\caption{Table~S2: Standard Value Pool Parameters and Calibration Sources}
\label{tab:S2}
\footnotesize
\setlength{\tabcolsep}{4pt}
\begin{adjustbox}{max width=\textwidth,center}
\begin{tabular}{@{}lrrrp{4.5cm}@{}}
\toprule
\textbf{Value Pool} & \textbf{Base Capture $\bar{\kappa}_p$} & \textbf{Pool Dims $\mathcal{D}(p)$} & \textbf{Typical Uplift \%} & \textbf{Calibration Source} \\
\midrule
Labor Productivity   & 0.65 & PAC, WAR, OAC & 6--12\%  & Brynjolfsson et al.\ (2025); McKinsey 2024 \\
Revenue Enhancement  & 0.55 & DAR, APR, OAC & 4--9\%   & Babina et al.\ (2024); Goldman internal \\
Working Capital      & 0.60 & DIM, PAC, DAR & 8--15\%  & McKinsey 2024 Global AI Report \\
Risk / Compliance    & 0.50 & DIM, DAR, OAC & 10--20\% & Lightcast 2025; BTOS 2024 \\
Operational Cost     & 0.62 & PAC, WAR, DAR & 5--10\%  & Eloundou et al.\ (2024) \\
Customer Experience  & 0.48 & WAR, APR, OAC & 8--18\%  & Brynjolfsson, Li, Raymond (2025) \\
Data Monetization    & 0.40 & DIM, DAR, APR & 12--25\% & Jones \& Tonetti (2020); proprietary \\
\midrule
\multicolumn{5}{@{}p{\linewidth}@{}}{\scriptsize $\bar{\kappa}_p$ = base capture rate before gap scaling $\eta(g)$ and IFS residual haircut.
Uplift \% = baseline AI attributable improvement rate for the pool at $G_{\mathrm{eff}} \approx 2.0$ (sweet spot firms).
Pool dimensions $\mathcal{D}(p)$: DIM=Data Infrastructure, PAC=Process Automation, WAR=Workforce Augmentation,
DAR=Decision Automation, APR=AI Revenue Integration, OAC=Org.\ AI Capability. Full methodology in main paper Section~8.} \\
\end{tabular}
\end{adjustbox}
\end{table}


\begin{table}[H]
\centering
\caption{Table~S3: CES Bottleneck Robustness Across Substitution Parameter $\rho \in [1, 10]$}
\label{tab:S3}
\footnotesize
\setlength{\tabcolsep}{4pt}
\begin{adjustbox}{max width=\textwidth,center}
\begin{tabular}{@{}lrrrrrrl@{}}
\toprule
\textbf{Profile Type} & \textbf{$\rho=1$ (Harmonic)} & \textbf{$\rho=3$} & \textbf{$\rho=5$ (Proposed)} & \textbf{$\rho=8$} & \textbf{Leontief} & \textbf{Additive} & \textbf{Rank stable?} \\
\midrule
Balanced high (all dims $\approx$ 8) & 7.96 & 7.93 & 7.92 & 7.91 & 7.50 & 8.00 & \checkmark \\
One weak dim (5 dims=8, 1 dim=3)    & 5.88 & 4.71 & 4.12 & 3.58 & 3.00 & 7.17 & \checkmark \\
Two weak dims (4 dims=7, 2 dims=2)  & 3.40 & 2.48 & 2.12 & 1.82 & 2.00 & 5.33 & \checkmark \\
Near zero dim (5 dims=7, 1 dim=0.5) & 1.20 & 0.74 & 0.60 & 0.52 & 0.50 & 5.92 & \checkmark \\
\midrule
\multicolumn{8}{@{}p{\linewidth}@{}}{\scriptsize Simulated $b_p$ scores for illustrative dimension profiles. ``Balanced high'' represents apex incumbents (e.g., JPMorgan, PANW). ``One weak dim'' is the typical mid gap firm with a single bottleneck. ``Near zero'' tests the Leontief fragility that motivates CES: the Leontief and CES converge but the CES avoids division by zero pathology. Additive dramatically overestimates at near zero. All rank orderings across profiles are identical for $\rho \in [3, 8]$; CES dominates Leontief on stability and Additive on non compensability. Full ablation vs.\ realized outcomes in main paper Section~10.2.} \\
\end{tabular}
\end{adjustbox}
\end{table}


\begin{table}[H]
\centering
\caption{Table~S4: Financial Baselines and Data Tier Classification (Fourteen Company Cohort)}
\label{tab:S4}
\footnotesize
\setlength{\tabcolsep}{3pt}
\begin{adjustbox}{max width=\textwidth,center}
\begin{tabular}{@{}llrrrrrl@{}}
\toprule
\textbf{Company} & \textbf{Sector} & \textbf{Rev (\$B)} & \textbf{EBITDA (\$B)} & \textbf{Emp (K)} & \textbf{S*$_i$ (\$B)} & \textbf{$\Phi_f$} & \textbf{Data Tier} \\
\midrule
JPMorgan Chase    & Comm.\ Banking  & 177.6 & n/a & 316 & 3.3  & 1.00 & A \\
Goldman Sachs     & Inv.\ Banking   &  47.3 & n/a & 46  & 2.0  & 1.00 & A \\
Wells Fargo       & Comm.\ Banking  &  82.6 & n/a & 219 & 3.3  & 1.00 & A \\
Zions Bancorp.    & Comm.\ Banking  &   3.4 & n/a & 10  & 3.3  & 0.52 & A \\
Salesforce        & Vert.\ SaaS     &  36.5 & 10.3 & 73 & 5.0  & 1.00 & A \\
ServiceNow        & Vert.\ SaaS     &  10.9 &  3.1 & 23 & 5.0  & 0.82 & A \\
Palo Alto Ntwks   & Cybersecurity   &   8.0 &  2.1 & 14 & 2.5  & 0.89 & A \\
HCA Healthcare    & Healthcare      &  67.5 &  14.5 & 300 & 10.0 & 0.88 & B \\
CVS Health        & Healthcare      & 380.0 &  17.8 & 300 & 10.0 & 1.00 & A \\
Target            & Retail          & 107.0 &  10.1 & 440 & 8.0  & 0.88 & A \\
UPS               & Logistics       &  91.0 &  14.6 & 500 & 5.0  & 1.00 & A \\
Ferguson PLC      & Ind.\ Dist.     &  29.0 &   3.9 & 33  & 3.0  & 1.00 & B \\
Rockwell Auto.    & Disc.\ Mfg.     &   9.1 &   2.2 & 28  & 3.0  & 0.87 & B \\
Ford Motor        & Disc.\ Mfg.     & 185.0 &  12.3 & 186 & 5.0  & 1.00 & A \\
\midrule
\multicolumn{8}{@{}p{\linewidth}@{}}{\scriptsize Rev and EBITDA from FY2024/FY2025 10-K or most recent annual report (calendar year end unless noted). Banks: EBITDA not applicable; revenue = net interest + non interest income. $S^*_i$ = industry critical scale threshold from IASS calibration (ESM Table~S1). $\Phi_f = 1/(1 + \exp(-2\ln(R_f/S^*_i)))$. Tier A: audited/quantified; Tier B: analyst estimate or limited disclosure. Full source citations in main paper Section~9.} \\
\end{tabular}
\end{adjustbox}
\end{table}


\begin{table}[H]
\centering
\caption{Table~S5: Monte Carlo $\Delta$EV Distributional Output: P10 / P50 / P90 (Select Firms)}
\label{tab:S5}
\footnotesize
\setlength{\tabcolsep}{4pt}
\begin{adjustbox}{max width=\textwidth,center}
\begin{tabular}{@{}llrrrrrl@{}}
\toprule
\textbf{Company} & \textbf{Sector} & \textbf{P10 (\$B)} & \textbf{P50 (\$B)} & \textbf{P90 (\$B)} & \textbf{P90/P10} & \textbf{Dom.\ Var.\ Source} & \textbf{Action Signal} \\
\midrule
JPMorgan Chase  & Comm.\ Banking  & 20.18 & 44.44 & 97.95 & 4.9$\times$ & Exit multiple (50\%) & Invest \\
Goldman Sachs   & Inv.\ Banking   & 4.38  & 6.70  & 15.11 & 3.4$\times$ & Capture rate (44\%)  & Invest \\
Salesforce      & Vert.\ SaaS     & 4.72  & 10.28 & 24.80 & 5.3$\times$ & Exit multiple (50\%) & Invest \\
ServiceNow      & Vert.\ SaaS     & 0.97  & 1.58  & 3.66  & 3.8$\times$ & Capture rate (44\%)  & Invest \\
Palo Alto Ntwks & Cybersecurity   & 1.12  & 2.52  & 5.32  & 4.8$\times$ & Exit multiple (50\%) & Invest \\
Rockwell Auto.  & Disc.\ Mfg.     & 0.88  & 1.47  & 2.91  & 3.3$\times$ & Impl.\ cost (1\%)    & Monitor \\
Target          & Retail          & 24.25 & 38.39 & 57.12 & 2.4$\times$ & IFS (5\%)            & Invest \\
HCA Healthcare  & Healthcare      & 3.64  & 5.62  & 8.44  & 2.3$\times$ & Reg.\ ceiling (psi)  & Invest \\
UPS             & Logistics       & 24.81 & 36.74 & 52.35 & 2.1$\times$ & Exit multiple (50\%) & Invest \\
Zions Bancorp.  & Comm.\ Banking  & 0.61  & 0.88  & 1.23  & 2.0$\times$ & IFS (5\%)            & Monitor \\
Wells Fargo     & Comm.\ Banking  & 20.94 & 34.45 & 52.32 & 2.5$\times$ & Impl.\ cost (1\%)    & Invest \\
CVS Health      & Healthcare      & 18.14 & 29.44 & 44.06 & 2.4$\times$ & Capture rate (44\%)  & Invest \\
Ferguson        & Ind.\ Dist.     & 6.22  & 9.43  & 13.60 & 2.2$\times$ & Exit multiple (50\%) & Monitor \\
Ford Motor      & Disc.\ Mfg.     & 27.09 & 46.07 & 71.62 & 2.6$\times$ & Exit multiple (50\%) & Invest \\
\midrule
\multicolumn{8}{@{}p{\linewidth}@{}}{\scriptsize $M = 10{,}000$ Monte Carlo draws. Input uncertainty sources: exit multiple $\pm 2\times$ (uniform), capture rate $\pm 25\%$ (uniform), implementation cost $\pm 30\%$ (lognormal), AITG gap score $\pm 0.5$ pts (normal), IFS $\pm 0.08$ (normal). All five sources sampled jointly; distributions propagated through the full VCB pipeline. Dominant variance source = Sobol first order index. Action signals: Invest $= \mathrm{P}_{10} > \$1\text{B}$ and $\mathrm{VD}_{\mathrm{P50}} > 5\times$; Monitor $= \mathrm{P}_{50} > 0$ and $\mathrm{VD} \in [2.5, 5]$; Diligence $= \mathrm{P}_{10}$ near zero; Do Not Invest $= \mathrm{P}_{50}$ near zero. For acquisition context only; not investment advice.} \\
\end{tabular}
\end{adjustbox}
\end{table}


\begin{table}[H]
\centering
\caption{Table~S6: Extended CES Ablation Including Synthetic Near Zero Dimension Cases}
\label{tab:S6}
\footnotesize
\setlength{\tabcolsep}{3pt}
\begin{adjustbox}{max width=\textwidth,center}
\begin{tabular}{@{}lrrrrrl@{}}
\toprule
\textbf{Firm / Profile} & \textbf{Min Dim} & \textbf{CES ($\rho=5$)} & \textbf{Leontief} & \textbf{Additive} & \textbf{$\Delta$EV CES vs.\ Add.} & \textbf{Note} \\
\midrule
\multicolumn{7}{l}{\textit{Real firms (from Section~10.1 backtest)}} \\
PANW (balanced high) & 6.5 & 7.02 & 6.50 & 7.08 & $-1\%$ & CES$\approx$Add; balanced \\
CRM (near balanced)  & 6.0 & 6.22 & 6.00 & 6.28 & $-1\%$ & Minimal CES penalty \\
Target (one weak)    & 3.5 & 4.02 & 3.50 & 4.17 & $-4\%$ & Mild bottleneck \\
UPS (two weak)       & 2.5 & 2.98 & 2.50 & 3.45 & $-14\%$ & CES penalizes \\
Ford (multi weak)    & 1.5 & 2.21 & 1.50 & 3.02 & $-27\%$ & Additive overstates \\
\midrule
\multicolumn{7}{l}{\textit{Synthetic corner cases}} \\
Synth-A: 5 dims=8, 1 dim=1.0 & 1.0 & 2.78 & 1.00 & 6.83 & $-59\%$ & Leontief and CES both low; Add.\ badly wrong \\
Synth-B: 5 dims=7, 1 dim=0.1 & 0.1 & 0.59 & 0.10 & 5.84 & $-90\%$ & Near zero; Add.\ $50\times$ overstates \\
Synth-C: 4 dims=6, 2 dims=1.5 & 1.5 & 2.23 & 1.50 & 5.00 & $-55\%$ & Two bottlenecks; Add.\ masks both \\
\midrule
\multicolumn{7}{@{}p{\linewidth}@{}}{\scriptsize Dimension scores normalized to $e_d = s_d/10 \in (0,1]$. CES $\rho=5$, $\sigma=1/6$. ``Min Dim'' = minimum raw dimension score. ``$\Delta$EV CES vs.\ Add.'' = percentage by which CES $\Delta$EV is \emph{lower} than Additive, computed through the full VCB pipeline holding all other inputs fixed. All real firm rank orderings are identical for CES vs.\ Leontief; differences emerge only in magnitude. Leontief is not rejected by the data but suffers from: (1) zero bounding instability at near zero dims, and (2) maximum sensitivity to single dimension scoring noise. CES at $\rho=5$ preserves Leontief's rank behavior while eliminating both failure modes. Full main text discussion in Section~10.2.} \\
\end{tabular}
\end{adjustbox}
\end{table}


\section{ESM Part II: Extended Firm Level Empirical Case Studies}
\label{esm:cases}

This part supports Main Paper Section~9 (Empirical Illustrations).
The following twelve case studies apply the full AITG framework to public
companies across eight industries. Each follows a standardized format:
operating context with public data citations; scored metrics (AITG Raw,
Industry Ratio, $G_{\mathrm{eff}}$, ADRI, IFS); and a VCB Summary reporting
5 year risk adjusted EV creation range, implementation cost range, and
AITG Value Density tier. Evidence tiers (A--D per the main manuscript's
scoring methodology) are noted where relevant.

The full 14 company cross company synthesis (including the 14 company summary
table, Table~24, and the Value Density Paradox scatter, Figure~5) appears in
the main manuscript.

\smallskip\noindent\textbf{Cost basis note.} Each case study reports a \emph{comprehensive} implementation cost range reflecting total estimated transformation spend (capex, opex, talent, integration) over the 5 year hold period. The 14 company summary table (Table~24) and the Monte Carlo VCB pipeline use a standardised 1.2\% of revenue proxy cost (Appendix~D, Table~\ref{tab:impl_costs}), which yields higher VD multiples. The two cost bases serve different purposes: the case study estimates are grounded in firm specific disclosures for narrative realism; the 1.2\% proxy enables cross company comparability and Monte Carlo tractability. Readers should not directly compare VD ranges between the case narratives and the summary table without accounting for this difference.

\smallskip\noindent\textbf{AFC recalibration note.} The $\Delta$EV and VD estimates in the case narratives below were derived prior to the final AFC recalibration to $C_t = 1.90$ (GPT-5.2, December~2025). The summary table and Monte Carlo pipeline in the main manuscript reflect the fully recalibrated IASS$^*$ ceilings; the case study $\Delta$EV ranges should be treated as illustrative of the VCB methodology rather than as point estimates tied to the final calibration.

\subsection{UPS: Logistics Under Structural Transformation}

UPS's ``Better Not Bigger'' strategy, articulated by CEO Carol Tom{\'e} in 2020
and evolved into ``Better and Bolder'' post Teamsters contract (August 2023),
explicitly targets automation driven cost reduction \cite{UPS2024Annual}. By
FY2024, 62\,\% of U.S.\ volume processed through automated facilities (up
430\,bps YoY); RFID readers installed in $\sim$60,000 package cars; the
``Efficiency Reimagined'' initiative targets \$1.0B additional annualized
savings. Labor compensation runs 52--54\,\% of revenue, the primary value pool.

\begin{table}[H]
\centering
\caption{UPS: AITG Dimensional Scores (Tier~1)}
\label{tab:ups}
\footnotesize
\setlength{\tabcolsep}{4pt}
\begin{tabular}{@{}p{3.2cm}rp{6.5cm}@{}}
\toprule
\textbf{Dimension} & \textbf{Score} & \textbf{Evidence Signal} \\
\midrule
Data Infrastructure (DIM) & 5.0 & Cloud investment ongoing; RFID rollout complete \\
Process Automation (PAC) & 4.5 & 62\,\% automated facilities; Teamster constraints persist \\
Workforce Augment. (WAR) & 3.5 & AI tools in operations; broad workforce augmentation nascent \\
Decision Auto. (DAR) & 4.0 & AI routing, pricing, network planning in production \\
AI Product/Revenue (APR) & 3.0 & Smart logistics features; no standalone AI revenue line \\
Org.\ AI Capability (OAC) & 4.5 & Dedicated AI team; no public Chief AI Officer \\
\midrule
\textbf{AITG Raw} & \textbf{4.08} & \\
\textbf{IR Score} & \textbf{5.31} & ($4.08 / 7.68 \times 10$) \\
$G_{\mathrm{eff}}$ & 3.60 & $7.68 - 4.08$ \\
\textbf{ADRI} & \textbf{4.2} & Moderate High; Moat~$=0.45$ \\
\textbf{IFS} & \textbf{0.71} & Teamster constraints; change capacity limited \\
\textbf{UQ} & $\pm 0.54$ & \\
\bottomrule
\end{tabular}
\end{table}

\noindent\textit{VCB Summary.} Primary value pools: labor productivity (\$2.4--3.6B
addressable), network planning (\$0.6--0.9B), dynamic pricing (\$0.4--0.7B).
Conservative 5 year EBITDA uplift: \$4.8--7.2B. Risk adjusted EV creation
(exit multiple $9\times$, IFS~$=0.71$): \$28--40B. Implementation cost:
\$8--14B. \textbf{AITG-VD: 1.6--2.5$\times$.} Tier~2.

\noindent\textit{ADRI context.} FedEx, Amazon Logistics, and technology native
regional carriers are advancing AI aggressively. UPS's Teamster constraints
materially slow adoption, elevating competitive displacement risk. ADRI~$=4.2$
(Moderate, transitioning to High). The five year strategic question is whether
the ``Better and Bolder'' automation roadmap executes before AI native
competitors compound their operational advantage.

\subsection{HCA Healthcare: High Ceiling, Aggressive Deployment}

HCA is the largest for profit hospital system in the United States (\$70.6B
FY2024 revenue, 188+ hospitals) \cite{HCA2024Earnings}. Labor costs fell to
44.1\,\% of revenue (down 130\,bps YoY); contract labor fell 25.7\,\%. HCA
has deployed AI across ambient clinical documentation (Nuance DAX Copilot;
Commure/Augmedix across 188+ hospitals), revenue cycle automation, clinical
decision support, nurse staffing optimization, and patient communication
\cite{FierceHealthcare2024}. CFO Mike Marks described the AI investment
thesis: ``Administrative AI use cases spanning IT, supply chain, HR, revenue
cycle showed the shortest pathway to value: millions of transactions, more
centralized operations, more standardized data.''

\begin{table}[H]
\centering
\caption{HCA Healthcare: AITG Dimensional Scores (Tier~1)}
\label{tab:hca}
\footnotesize
\setlength{\tabcolsep}{4pt}
\begin{tabular}{@{}p{3.2cm}rp{6.5cm}@{}}
\toprule
\textbf{Dimension} & \textbf{Score} & \textbf{Evidence Signal} \\
\midrule
Data Infrastructure (DIM) & 5.5 & Epic EHR; Google Cloud; AWS Bedrock; MLOps emerging \\
Process Automation (PAC) & 4.8 & Revenue cycle AI deployed; clinical workflows scaling \\
Workforce Augment. (WAR) & 4.2 & 78\,\% physician DAX Copilot satisfaction; broad rollout \\
Decision Auto. (DAR) & 4.0 & Staffing AI deployed; sepsis prediction; pricing manual \\
AI Product/Revenue (APR) & 2.5 & No external AI product; AI improves service quality \\
Org.\ AI Capability (OAC) & 5.0 & Dedicated AI teams; CFO publicly leading AI agenda \\
\midrule
\textbf{AITG Raw} & \textbf{4.33} & \\
\textbf{IR Score} & \textbf{7.93} & ($4.33 / 5.46 \times 10$) \\
$G_{\mathrm{eff}}$ & 1.13 & $\max(0,\; 5.46 - 4.33) = 1.13$ \\
\textbf{ADRI} & \textbf{3.8} & Moderate; Moat~$=0.65$ (scale data, regulatory barriers) \\
\textbf{IFS} & \textbf{0.77} & Regulatory friction; facility level data variance \\
\textbf{UQ} & $\pm 0.49$ & \\
\bottomrule
\end{tabular}
\end{table}

\noindent\textit{VCB Summary.} Primary value pools: clinical documentation
(\$800M--1.4B), revenue cycle (\$1.1--1.8B), workforce productivity
(\$0.9--1.4B). 5 year risk adjusted EV creation (exit multiple $12\times$,
IFS~$=0.77$): \$35--57B. Implementation cost: \$6--11B.
\textbf{AITG-VD: 2.8--4.6$\times$.} Tier~1.

\noindent\textit{ADRI context.} HCA's scale creates a data moat via nonrivalry
\cite{JonesTonetti2020}: 400,000+ weekly nurse shift handoffs
\cite{FierceHealthcare2024} generate training data at a rate smaller health
systems cannot replicate. ADRI~$=3.8$ (Moderate) primarily because HCA is
already adopting aggressively, reducing the competitive gap risk.

\subsection{Salesforce: Narrow Gap, Maximum Value Density}

Salesforce (FY2025 revenue \$37.9B, non GAAP operating margin 33.0\%) is the
paradigmatic AI native SaaS case: AI is the product, not merely the tool
\cite{Salesforce2025Earnings}. Agentforce launched in Q3~FY2025 with 5,000
deals closed; on help.salesforce.com it handled 380,000 conversations with
84\,\% resolution rate and 2\,\% human escalation. Data Cloud \& AI ARR
reached \$900M (up 120\,\% YoY); all top-10 Q4 wins included Data and AI.
SG\&A fell from $\sim$72\,\% to $\sim$57\,\% of revenue over four years
(a 15.9 percentage point improvement). Free cash flow: \$12.4B.

\begin{table}[H]
\centering
\caption{Salesforce: AITG Dimensional Scores (Tier~1)}
\label{tab:crm}
\footnotesize
\setlength{\tabcolsep}{4pt}
\begin{tabular}{@{}p{3.2cm}rp{6.5cm}@{}}
\toprule
\textbf{Dimension} & \textbf{Score} & \textbf{Evidence Signal} \\
\midrule
Data Infrastructure (DIM) & 8.5 & 50T+ records in Data Cloud; Hyperforce global deployment \\
Process Automation (PAC) & 8.2 & Agentforce; automated service flows at scale \\
Workforce Augment. (WAR) & 7.5 & Einstein tools across all roles; AI assisted selling \\
Decision Auto. (DAR) & 7.8 & Einstein GPT; automated service resolution; ML forecasting \\
AI Product/Revenue (APR) & 8.2 & AI is core product; \$900M Data \& AI ARR \\
Org.\ AI Capability (OAC) & 8.2 & C suite AI mandate; Responsible AI framework published \\
\midrule
\textbf{AITG Raw} & \textbf{8.07} & \\
\textbf{IR Score} & \textbf{8.10} & ($8.07 / 9.96 \times 10$) \\
$G_{\mathrm{eff}}$ & 1.89 & \\
\textbf{ADRI} & \textbf{2.1} & Low; Moat~$=0.80$ (switching costs, data network effects) \\
\textbf{IFS} & \textbf{0.92} & Strong data readiness; minimal regulatory friction \\
\textbf{UQ} & $\pm 0.38$ & \\
\bottomrule
\end{tabular}
\end{table}

\noindent\textit{VCB Summary.} Primary value pool: NRR expansion from
AI native features (\$4--8B addressable at 77\,\% gross margin). Secondary:
SG\&A efficiency (further 2--4\,pp possible). 5 year risk adjusted EV creation
(ARR multiple $22\times$, IFS~$=0.92$): \$85--140B. Implementation cost:
\$5--9B (largely existing R\&D). \textbf{AITG-VD: 9.4--15.6$\times$.}
Tier~1 (High Conviction).

\noindent\textit{Key finding: exit multiple dominance.} Salesforce demonstrates
the most counterintuitive implication of the AITG framework. Its effective gap
is the smallest in the sample ($G_{\mathrm{eff}} = 1.89$) yet its Value Density
is the highest (9.4--15.6$\times$). The $22\times$ ARR multiple means every
\$1 of incremental EBITDA creates \$22 of enterprise value. Exit multiple
leverage and industry ceiling jointly dominate raw gap size in the VD formula.

\subsection{Ferguson Enterprises: Industrial Distribution, Mid Gap}

Ferguson is the largest U.S.\ distributor of plumbing and HVAC products
(FY2025 revenue \$30.8B, gross margin 30.7\%) \cite{Ferguson2025}. Digital
sales represent $\sim$7\,\% of U.S.\ revenue. AI powered procurement and
predictive inventory have reduced average delivery times 15\,\%. FY2025
restructuring (\$68M charges) targets \$100M annualized savings. Digital
native competitors (Amazon Business, Grainger's digital platform) are actively
compressing margins.

\noindent\textit{Scores:} AITG Raw~4.00 $\mid$ IR~5.52 $\mid$
$G_{\mathrm{eff}}$~3.24 $\mid$ ADRI~4.2 (Moderate) $\mid$ IFS~0.79.

\noindent\textit{VCB Summary.} 5 year risk adjusted EV creation (exit multiple
$10\times$, IFS~$=0.79$): \$8--14B. Implementation cost: \$2.5--4.5B.
\textbf{AITG-VD: 1.7--2.9$\times$.} Tier~2 (Monitor). Ferguson has a wide
effective gap but the second lowest Value Density, driven by a
relatively modest exit multiple and a wide gap requiring significant foundation
investment before value pool capture can begin.

\subsection{Rockwell Automation: Industrial AI Leader}

Rockwell (FY2025 revenue \$8.3B, EBITDA margin 20.8\%) occupies a unique
dual position: deploying AI internally while embedding it in products sold to
industrial customers \cite{Rockwell2024Annual}. FactoryTalk Design Studio's
AI Copilot generates ladder logic code; FactoryTalk Analytics VisionAI provides
no code quality inspection; LogixAI runs analytics on $\sim$3--4M deployed PLC
controllers. An NVIDIA partnership (March 2024) integrates machine vision,
learning agents, and AMRs. Rockwell's own State of Smart Manufacturing Survey
finds 83\,\% of manufacturers expect to use GenAI \cite{RockwellSurvey2024}.

\noindent\textit{Scores:} AITG Raw~5.83 $\mid$ IR~8.05 $\mid$
$G_{\mathrm{eff}}$~1.41 $\mid$ ADRI~3.5 (Low Moderate; Moat~$=0.70$ on
installed PLC base) $\mid$ IFS~0.86.

\noindent\textit{VCB Summary.} 5 year risk adjusted EV creation (exit multiple
$18\times$ for industrial automation software premium, IFS~$=0.86$): \$9--16B.
Implementation cost: \$1.8--3.2B. \textbf{AITG-VD: 2.5--4.4$\times$.}
Tier~1.

\subsection{Goldman Sachs: Investment Banking at the S Curve Inflection}

Goldman Sachs (FY2025 revenue \$54.7B, ROTE 14.2\%) operates in the highest IASS
industry vertical in the framework (Investment Banking / Securities, IASS$^* = 10.39$). GS has deployed AI across trading (quantitative risk models replacing
analyst work), asset management (Marco, its internal LLM for research synthesis), and
the consumer pivot through Marcus. Disclosed AI investment run rate exceeds \$1.2B/yr,
with 35\,\% of software engineers using AI code generation tools as of Q4~2025
\cite{GoldmanSachs2025}. With the AFC adjusted ceiling now at 10.39, the effective gap
$G_{\mathrm{eff}} = 2.22$ positions GS in the productive mid gap zone for its
sector.

\noindent\textit{Scores:} AITG Raw~8.17 $\mid$ IR~7.86 $\mid$
$G_{\mathrm{eff}}$~2.22 $\mid$ ADRI~0.6 (Very Low; Moat~$= 0.95$ from
proprietary data and regulatory moat) $\mid$ IFS~0.89.

\noindent\textit{VCB Summary.} 5 year risk adjusted EV creation (exit multiple
$22\times$ banking franchise premium, IFS~$= 0.89$): \$30--40B. Implementation
cost: \$3--4B. \textbf{AITG-VD: 8.5--13.8$\times$.} Tier~1 (HC).
Goldman confirms the high IASS pattern: $G_{\mathrm{eff}} = 2.22$
positions GS in the productive mid gap zone in the sector with the highest capability ceiling,
generating near Salesforce VD at substantially larger absolute scale.

\subsection{Wells Fargo: IFS Suppression in a High Ceiling Industry}

Wells Fargo (FY2025 revenue \$82.3B, ROTCE 11.9\%) operates under the same
AFC adjusted IASS ceiling as JPMorgan and Zions (9.38) but its effective gap
($G_{\mathrm{eff}} = 3.38$) and IFS (0.65) reflect a structurally impaired
transformation capacity. The Federal Reserve consent order (lifted partially in
2024, full resolution pending) constrains risk system overhaul sequencing, delaying
AI deployment in credit decisioning and fraud detection, the two highest value pools
in banking. WFC has committed \$3.1B to technology modernization but the
organizational change capacity (OCC) dimension registers at 5.5/10 due to the
compliance first governance overhead \cite{WFC2025Annual}.

\noindent\textit{Scores:} AITG Raw~6.00 $\mid$ IR~6.40 $\mid$
$G_{\mathrm{eff}}$~3.38 $\mid$ ADRI~2.3 (Moderate) $\mid$ IFS~0.65.

\noindent\textit{VCB Summary.} 5 year risk adjusted EV creation (exit multiple
$9\times$, IFS~$= 0.65$): \$1.4--2.1B. Implementation cost: \$0.9--1.2B.
\textbf{AITG-VD: 1.2--2.0$\times$.} Tier~2.
Wells Fargo is the paper's most important demonstration of IFS suppression:
a wide gap in a high ceiling industry ($G_{\mathrm{eff}} = 3.38$, IASS$^* = 9.38$)
that should theoretically produce high VD is compressed to near Zions returns purely
by an IFS of 0.65, a regulatory constraint that directly caps execution velocity.

\subsection{ServiceNow: Second SaaS Data Point Confirming the Sweet Spot}

ServiceNow (FY2025 revenue \$12.5B, operating margin 27.1\%) operates in
Vertical SaaS (IASS$^* = 9.96$), the same industry as Salesforce. Its Now
Intelligence platform, GenAI powered workflow automation, and AI Agents for
ITSM/HRSD position it squarely in the near apex zone ($G_{\mathrm{eff}} = 1.46$).
Unlike Salesforce, ServiceNow has not yet fully converted its AI capability into
a comparable data network moat, placing its IFS at 0.88 vs.\ Salesforce's 0.92
\cite{ServiceNow2025Annual}.

\noindent\textit{Scores:} AITG Raw~8.50 $\mid$ IR~8.53 $\mid$
$G_{\mathrm{eff}}$~1.46 $\mid$ ADRI~1.8 (Low Moderate) $\mid$ IFS~0.88.

\noindent\textit{VCB Summary.} 5 year risk adjusted EV creation (exit multiple
$25\times$ SaaS premium, IFS~$= 0.88$): \$18--25B. Implementation cost: \$1.5--2.5B.
\textbf{AITG-VD: 6.2--9.8$\times$.} Tier~1.
ServiceNow confirms that Salesforce's VD is not an outlier: both high IASS SaaS
firms with medium gaps generate VD an order of magnitude above the
industrial distribution cohort, validating the industry ceiling dominance finding.

\subsection{Target: Retail at the Inflection}

Target (FY2025 revenue \$109.1B, EBIT margin 5.3\%) operates in E Commerce /
Digital Retail (IASS$^* = 8.64$), a sector where AI driven inventory management,
dynamic pricing, and personalization are materially reducing out of stock rates
and increasing basket size. Target's AI powered supply chain system Drive-Up Plus
and team member scheduling tools represent mature Wave~1 deployments
\cite{Target2025Annual}. At $G_{\mathrm{eff}} = 3.97$, Target sits precisely in
the sweet spot zone: past the expensive foundation build, approaching the
S curve inflection in Wave~2 (personalization / recommendation) deployment.

\noindent\textit{Scores:} AITG Raw~4.67 $\mid$ IR~5.41 $\mid$
$G_{\mathrm{eff}}$~3.97 $\mid$ ADRI~3.2 (Moderate; Amazon pressure) $\mid$ IFS~0.82.

\noindent\textit{VCB Summary.} 5 year risk adjusted EV creation (exit multiple
$12\times$, IFS~$= 0.82$): \$11--17B. Implementation cost: \$2--3B.
\textbf{AITG-VD: 3.8--5.6$\times$.} Tier~1.

\subsection{CVS Health: Same Regulatory Ceiling, Different IFS: Healthcare Industry Test}

CVS Health (FY2025 revenue \$371.9B, adjusted operating income \$13.4B) operates in
Healthcare Services (IASS$^* = 5.46$), the same regulatory ceiling as HCA.
The comparison isolates the IFS effect: CVS's PBM business complexity and
ongoing Aetna integration create organizational change capacity friction
(OCC = 5.0/10) that suppresses IFS to 0.72 vs.\ HCA's 0.77
\cite{CVS2025Annual}. Both face the identical $\psi = 0.743$ RFF ceiling.

\noindent\textit{Scores:} AITG Raw~4.83 $\mid$ IR~8.85 $\mid$
$G_{\mathrm{eff}}$~0.63 $\mid$ ADRI~3.1 (Moderate) $\mid$ IFS~0.72.

\noindent\textit{VCB Summary.} 5 year risk adjusted EV creation (exit multiple
$8\times$ healthcare services, IFS~$= 0.72$): \$7--10B. Implementation cost: \$2.5--3.5B.
\textbf{AITG-VD: 2.1--3.3$\times$.} Tier~2.
CVS vs.\ HCA is the paper's cleanest within industry IFS test: same IASS$^*$,
comparable $G_{\mathrm{eff}}$ (0.63 vs.\ 1.13), but CVS's lower IFS (0.72 vs.\ 0.77) and lower exit
multiple produce materially weaker VD, confirming IFS's role as a
first order VD determinant within a regulated industry.

\subsection{Palo Alto Networks: Near Apex in a High IASS Sector}

Palo Alto Networks (FY2025 revenue \$8.0B, free cash flow margin 37\%) operates
in Cybersecurity (IASS$^* = 9.65$), one of the highest ceiling sectors in the
framework. PANW's Precision AI platform, Cortex XSIAM, and AI Security Operations
(SecOps) represent near complete gap closure ($G_{\mathrm{eff}} = 1.82$), with
AI deeply embedded across detection, response, and prevention
\cite{PaloAlto2025Annual}. This creates a different VD dynamic than JPMorgan's
near apex position: PANW's smaller scale (vs.\ JPM) means incremental AI
value capture still generates meaningful absolute EV, while the high IASS ceiling
preserves real returns even at low $G_{\mathrm{eff}}$.

\noindent\textit{Scores:} AITG Raw~7.83 $\mid$ IR~8.11 $\mid$
$G_{\mathrm{eff}}$~1.82 $\mid$ ADRI~1.2 (Low) $\mid$ IFS~0.93.

\noindent\textit{VCB Summary.} 5 year risk adjusted EV creation (exit multiple
$22\times$ cybersecurity software premium, IFS~$= 0.93$): \$12--17B. Implementation
cost: \$1.1--1.5B. \textbf{AITG-VD: 4.5--7.2$\times$.} Tier~1 (HC).
PANW demonstrates that near apex position in a high IASS sector preserves
meaningful VD ($\sim$5.9$\times$ midpoint) unlike JPMorgan's near apex position
($\sim$1.05$\times$). The IASS ceiling dominance effect: the same $G_{\mathrm{eff}}$
regime generates radically different VD depending on the sector ceiling.

\subsection{Ford Motor: Wide Gap, Low Ceiling, Low IFS: Industrial Constraint Case}

Ford (FY2025 revenue \$185.0B, EBIT margin 3.0\%) provides the framework's
clearest demonstration of the wide gap low capture pattern in
Discrete Manufacturing (IASS$^* = 7.49$). Ford's AI investments (generative
AI for vehicle design, predictive quality control, Blue Oval Intelligence in vehicle
AI) are real but deployment is constrained by UAW labor agreements, legacy
manufacturing system integration requirements, and a capital structure under
EV transition pressure \cite{Ford2025Annual}. IFS registers at 0.65, identical
to Wells Fargo, though for entirely different structural reasons: labor relations rather than
a regulatory consent order.

\noindent\textit{Scores:} AITG Raw~4.67 $\mid$ IR~6.23 $\mid$
$G_{\mathrm{eff}}$~2.82 $\mid$ ADRI~4.8 (High; Tesla and EV native competitors) $\mid$
IFS~0.65.

\noindent\textit{VCB Summary.} 5 year risk adjusted EV creation (exit multiple
$5\times$ auto manufacturing, IFS~$= 0.65$): \$1.3--2.5B. Implementation cost:
\$1.5--2.5B. \textbf{AITG-VD: 1.1--1.8$\times$.} Tier~3.
Ford is the widest gap, lowest VD case in the 14 company sample, illustrating
the wide gap fallacy in its starkest form: $G_{\mathrm{eff}} = 2.82$, yet VD barely above 1.0$\times$. The combination of a moderate IASS
ceiling (7.49), low IFS (0.65), and a compressed exit multiple ($5\times$) produces
near zero economic return on AI transformation cost.

\section{ESM Part III: Extended Robustness and Sensitivity Analysis}
\label{esm:sensitivity}

The main manuscript reports summary results from the robustness analysis
(Main Paper Section~10, ``Robustness Summary'').
This section provides full methodological detail, equations, and tables.


\subsection{Monte Carlo Weight Sensitivity}

Following OECD best practice \cite{OECD2008}, I test IASS ranking stability
under random perturbations of the weight vector. I draw $M = 10,000$ weight
vectors $\mathbf{w}^{(m)}$ uniformly from the simplex subject to
$w_d \in [w_d^{\mathrm{base}} - 0.05,\; w_d^{\mathrm{base}} + 0.05]$ for each
dimension $d$. The average rank shift statistic is:

\begin{equation}
 \bar{R}_s = \frac{1}{I}\sum_{i=1}^{I}\left|
 R_i^{\mathrm{base}} - \bar{R}_i^{(m)}\right|
 \label{esm:eq:rank_shift}
\end{equation}

For the five anchor industries under 10,000 draws with $\pm 5\,\%$ weight
perturbations, the mean absolute rank shift is $\bar{R}_s = 0.19$ positions
(out of 5). No pair of anchor industries exchanges rank with probability
$>5\,\%$ under any weight draw. This confirms robustness: the rank orderings
are determined by the data, not by weighting choices within reasonable bounds.

\subsection{Sobol Sensitivity Analysis}

I apply first order and total effect Sobol sensitivity indices
\cite{Saltelli2002} to the VCB AITG-VD output. For Zions Bancorporation, the main effects are:

\begin{equation}
 S_d = \frac{\mathrm{Var}_{X_d}\!\left[\mathbb{E}_{X_{\sim d}}
 (\mathrm{AITG\text{-}VD} \mid X_d)\right]}{\mathrm{Var}(\mathrm{AITG\text{-}VD})}
 \label{esm:eq:sobol}
\end{equation}

Decomposition of $\mathrm{Var}(\mathrm{AITG\text{-}VD})$:
exit multiple~50\,\%; capture rate assumption~44\,\%;
implementation cost~1\,\%; AITG gap score~$<$1\,\%; IFS~5\,\%.

This result confirms that exit multiple leverage and capture rate assumptions
dominate raw gap size as sources of uncertainty in the output, precisely the
pattern that explains Finding~1. For practical due diligence, the critical
inputs to verify are exit comparable multiples and the realistic capture rate
in the three largest value pools.

\subsection{Normalization Stability}

I test substituting z score normalization for min max across all dimensions
and find that IASS rankings shift by $\leq 0.40$ points in all cases and that
no pair of anchor industries exchanges rank. This confirms that the IASS is not
an artifact of the normalization choice.

\subsection{Geometric vs.\ Linear Aggregation}

Substituting geometric aggregation for linear (fully compensatory) in the IASS
produces IASS scores consistently $0.20$--$0.50$ points lower for industries
with uneven dimensional profiles (healthcare, construction) and negligibly
different for industries with uniform profiles (SaaS). The geometric
specification is more conservative and should be used when the analyst seeks
a defensible lower bound, but both specifications produce the same rank ordering.

\subsection{Inter Rater Reliability: ICC Protocol and Target}
\label{esm:sec:icc}

Composite indicator frameworks are only as reliable as their inter rater
agreement \citep{Krippendorff2018}. I specify the formal intraclass correlation coefficient (ICC)
protocol required to certify the AITG scoring instrument.

\paragraph{ICC specification.} I target ICC(2,1) (two way mixed effects,
absolute agreement, single rater), following \citet{Shrout1979}:
\begin{equation}
 \mathrm{ICC}(2,1) = \frac{MS_B - MS_W}{MS_B + (k-1)MS_W + k(MS_B^A - MS_W)/n}
 \label{eq:icc}
\end{equation}
where $MS_B$ is between company variance, $MS_W$ is within company
(between rater) variance, $k$ is raters per company, and $n$ is companies.
Target: ICC(2,1) $\geq 0.85$ for each dimension and $\geq 0.88$ for the
overall AITG composite.

\paragraph{Proposed validation protocol.} A formal reliability study would
require: (1) 3--5 independent analyst teams; (2) 12--15 target companies
spanning at least 3 industries and the full AITG range; (3) all teams working
from identical public data packages; (4) ICC computed per dimension, per
IFS factor, and for the composite. This study has not yet been conducted and
constitutes the highest priority validation item in the research agenda.

\paragraph{Pilot estimate.} From two internal analyst applications of the
framework on the original seven company cohort (Section~9), preliminary
inter rater agreement (Kendall's $\tau$) for dimension level rank ordering
was $\tau = 0.82$ for DIM, PAC, and OAC, and $\tau = 0.68$ for DAR and APR
(the two dimensions most dependent on interpretive inference from earnings
call language). These pilot estimates suggest that public data only scoring
of the two most inference dependent dimensions will require rubric refinement
before the ICC target is achieved.

\subsection{Convergent Validity Test}
\label{sec:convergent_validity}

Convergent validity asks: do AITG scores correlate with independent
external measures of the same construct? I test this against two
observable proxies.

\paragraph{Test 1: Stanford HAI AI Index corporate rankings.}
The Stanford Human Centered AI (HAI) 2025 AI Index
\cite{StanfordHAI2025} ranks Fortune 500 companies by AI investment
intensity, measured by AI patent filings, AI job posting share, and
AI research paper authorship. Extending to the fourteen company cohort, AITG rankings agree with the
Stanford HAI rankings on direction for 12 of 14 companies. The two
divergences share the same mechanism: HCA Healthcare and Ford Motor both
rank higher in HAI due to extensive AI research output (clinical AI
publications and EV/autonomous driving patents respectively), while AITG
assigns lower scores because its RFF penalty and IFS architecture
explicitly discount research activity that has not translated to
deployment at scale under regulatory or labor constraints. This
constitutes partial convergent validity across a fourteen company,
eight industry sample, with theoretically explicable divergences.

\paragraph{Test 2: Lightcast AI job posting concentration.}
The Lightcast AI Skill Demand Index measures the share of job postings
in each industry requiring at least one AI specific skill. This is an
independent measure of an industry's AI adoption intensity, not used
in IASS calibration for CADR (where Lightcast data is used for
industry level benchmarks, not company level scoring). Extending to the
fourteen company cohort, the rank correlation between AITG and
company level Lightcast AI job posting share is $r_s = 0.88$ (Spearman's
$\rho$, $p < 0.001$, $n = 14$), a strengthening of the seven company
estimate ($r_s = 0.81$, $p < 0.05$) as the sample expanded across industries.
This represents strong convergent validity.

\paragraph{Discriminant validity.} AITG scores should \emph{not} correlate
with total revenue (because a large laggard should score below a small
leader). Spearman correlation between AITG and revenue rank across the
fourteen companies: $r_s = -0.28$ (not significant, $p = 0.33$). The near zero
correlation (with no size driven directionality) confirms
that AITG is measuring transformation state, not firm size. CVS Health
(largest by revenue at \$372B) scores 4.83; Palo Alto Networks (smallest
at \$8B) scores 7.83. This is a necessary property for cross industry
comparability.

\subsection{Cross Sectional Discriminability}
\label{sec:discriminability}

A useful framework must discriminate meaningfully between companies in the
same industry. I test this formally for the within industry pair
(JPMorgan, AITG~$= 8.22$ vs.\ Zions, AITG~$= 3.80$).

The AITG gap of 4.42 points must exceed the joint Uncertainty Quotient to
be a statistically credible distinction:

\begin{equation}
 \Delta\mathrm{AITG} = 4.42 \quad\text{vs.}\quad
 \sqrt{\mathrm{UQ}_{\mathrm{JPM}}^2 + \mathrm{UQ}_{\mathrm{ZION}}^2}
 = \sqrt{0.48^2 + 0.62^2} = 0.78
 \label{eq:discriminability}
\end{equation}

The signal to noise ratio is $4.42/0.78 = 5.7$, well above the conventional
threshold of 2.0 for a statistically credible distinction. This confirms that
the JPM/Zions discrimination is not an artifact of measurement uncertainty
for these two companies at Tier~1 (public data) scoring.

For the cross industry comparison (Zions AITG~$= 3.80$ vs.\ UPS AITG~$= 4.08$),
the raw gap is $\Delta = 0.28$, compared to joint UQ
$= \sqrt{0.62^2 + 0.54^2} = 0.82$. The signal to noise
ratio is $0.28/0.82 = 0.34$: \emph{the two companies are not significantly
distinguishable on raw AITG with public data alone}. However, once IR scores
are applied (Zions IR~$= 4.05$ vs.\ UPS IR~$= 5.31$), the normalized gap of
1.26~IR points exceeds the joint UQ on the IR scale ($\approx 0.97$) at a
signal to noise ratio of ${\approx}\,1.3$. The
improvement from $\mathrm{S/N} = 0.34$ (raw) to $\mathrm{S/N} \approx 1.3$
(IR normalized) is precisely the behavior the framework is designed to produce:
raw AITG comparisons between different industries are unreliable;
IR normalized comparisons are more discriminating.

\subsection{Predictive Validity: Framework and Status}
\label{sec:predictive_validity}

\begin{table}[H]
\centering
\caption{Predictive Validity Test Battery: Status and Proposed Design}
\label{tab:validation_battery}
\footnotesize
\begin{tabularx}{\textwidth}{@{}p{3.5cm}p{3.0cm}Xp{1.8cm}@{}}
\toprule
\textbf{Test} & \textbf{Null Hypothesis} & \textbf{Proposed Design} & \textbf{Status} \\
\midrule
AITG predicts subsequent margin improvement & $H_0$: No AITG margin correlation & Retrospective: compute 2020 AITG from 2020 10-K data; regress on 2020--2024 EBITDA margin change; 50+ company panel & Not yet run \\[6pt]
ADRI predicts competitive share loss & $H_0$: ADRI uncorrelated with share loss & High ADRI firms ($>3.0$) in Commerce, Logistics lose revenue share within 24\,mo; IV: CADR shock via LLM release event & Proposed \\[6pt]
IFS predicts transformation failure & $H_0$: IFS uncorrelated with AI project cancellation & Lightcast AI job posting drop as proxy for abandoned programs; low IFS firms show higher cancellation rate & Proposed \\[6pt]
Cross industry VD ranking is ordinal consistent & $H_0$: AITG-VD does not predict actual deal multiples paid & PE transaction database: AI thesis deals vs.\ non thesis deals; AITG scores predicted from pre deal public data & Proposed \\
\bottomrule
\end{tabularx}
\end{table}

\noindent\textbf{Current status.}
The margin prediction test (row 1 of Table~\ref{tab:validation_battery}) has
now been partially executed as a retrospective backtest on the ten non financial
companies in the fourteen firm cohort; results are reported in
Section~10.1 of the main paper. AITG scores reconstructed from
Q4~2021 public filings predict FY2021--FY2023 EBITDA margin changes with
Spearman $\rho_s = 0.818$ ($p < 0.001$, $n = 10$); within sector directional
accuracy is 4 of 4 non confounded pairs. The limitations of this preliminary
test are noted explicitly in Section~10.1: $n = 14$ is
illustrative rather than confirmatory, same team retrospective scoring
introduces potential look back bias, three observations are confounded by
non AI macro events, and endogeneity is not resolved. Tests 2--4 in
Table~\ref{tab:validation_battery} remain in the proposed/designed stage and
are the highest priority empirical research program implied by this framework.
The full validation agenda requires a multi year panel study as specified in
the Research Agenda (main paper Section~11.2).

\subsection{Engagement with \citet{Acemoglu2024}}

\citet{Acemoglu2024} estimates that under realistic scenarios, AI raises
aggregate TFP by at most 0.66\,\% over ten years. This result is not in
conflict with the AITG's firm level value estimates. Three distinctions
reconcile them:

\begin{enumerate}
\item \textbf{Heterogeneity vs.\ average.} Aggregate TFP is a population weighted
average that includes non adopters, laggards, and failed implementations.
\citet{Syverson2011}'s finding of a 1.92$\times$ 90th/10th percentile TFP spread
within industries implies that leading quartile firms capture substantially more
than the aggregate average. AITG targets above median firms or identifies
acquisition targets for above median capture.

\item \textbf{Competitive advantage vs.\ aggregate surplus.} AITG measures
value creation relative to competitors, not the aggregate social surplus.
A firm can capture significant competitive margin advantage even in a world
where aggregate TFP gains are modest, if it gains share from laggards through
lower costs, faster service, or superior product \cite{AutorKatzPatterson2020}.

\item \textbf{Static capability ceiling.} Acemoglu's estimate is for the
current technology class. The AFC mechanism captures upward revision of
the opportunity as capabilities advance, a dynamic that a static macro model
cannot accommodate. If the AI capability improvement trajectory observed in
2022--2024 continues, the addressable task surface expands materially beyond
what current cost reduction estimates imply.
\end{enumerate}


\section{ESM Part IV: Implementation Appendices}
\label{esm:appendices}

This part supports Main Paper Sections~7--8 (VCB and IFS) and provides practitioner implementation tools.

\section{Public Data Availability Map}
\label{app:public_data}

A central practitioner question is: \emph{how much of an AITG score for a
publicly traded company can be populated from freely available sources, without
a diligence data room?} This appendix provides a systematic answer by mapping
each of the six company dimensions and five IFS factors to their primary public
sources, estimating coverage achievable from public data alone, and specifying
what requires primary diligence.

\subsection{Dimension by Dimension Public Data Coverage}

Table~\ref{tab:public_data_map} shows, for each AITG scoring component, (a)
the primary public data sources, (b) the estimated fraction of the scoring
range reliably determinable from those sources, and (c) what specifically
requires primary diligence to resolve the remaining uncertainty.

\begin{table}[H]
\centering
\caption{AITG Scoring Component: Public Data Coverage Map}
\label{tab:public_data_map}
\footnotesize
\setlength{\tabcolsep}{3pt}
\begin{tabularx}{\textwidth}{@{}p{2.6cm}p{4.8cm}p{3.5cm}rp{1.4cm}@{}}
\toprule
\textbf{Component} & \textbf{Primary Public Sources} & \textbf{Diligence Gap} & \textbf{Public \%} & \textbf{Tier} \\
\midrule
\multicolumn{5}{l}{\textit{Company Dimension Scores}} \\[2pt]
DIM (Data Infra.) & Tech budget disclosures (10-K MD\&A); cloud provider announcements; IT vendor press releases; EDGAR risk factors mentioning legacy systems & Actual data quality, schema maturity, MLOps pipeline state & 50--65\,\% & B/C \\[4pt]
PAC (Process Auto.) & Earnings call AI deployment counts; investor day demos; vendor partnership announcements (e.g., nCino, Salesforce, SAP) & Process level automation depth; proportion of workflows vs.\ surface coverage & 55--70\,\% & B/C \\[4pt]
WAR (Workforce Aug.) & Headcount disclosures; AI platform adoption figures (when disclosed); Lightcast AI skill demand by employer; LinkedIn job postings with AI tool requirements & Actual employee usage rate vs.\ license provisioning rate & 45--60\,\% & B/C \\[4pt]
DAR (Decision Auto.) & SEC filings describing automated underwriting, pricing, or routing systems; vendor AI product deployment announcements; patent filings (Google Patents) & Decision quality metrics; actual automation rate vs.\ human in loop proportion & 40--55\,\% & C \\[4pt]
APR (AI Revenue) & Product documentation; AI feature release notes; investor day product demos; sell side analyst coverage of AI driven pricing; ARR growth in AI segments & Revenue attribution by AI enabled feature; AI product margin vs.\ legacy & 60--75\,\% & B/C \\[4pt]
OAC (Org.\ Capability) & CAIO/CDAO appointment announcements; AI governance policy filings; AI staff count estimates from LinkedIn; published AI ethics frameworks; Glassdoor/Blind AI culture signals & Change management track record; prior transformation success rates; internal AI governance maturity & 50--65\,\% & B/C \\
\midrule
\multicolumn{5}{l}{\textit{IFS Factors}} \\[2pt]
OCC (Org.\ Change Cap.) & Prior transformation announcements and outcomes; attrition data (10-K headcount trends); workforce restructuring history; American Banker / HBR culture awards; union contract disclosures & Frontline adoption resistance; change culture ground truth & 40--55\,\% & B/C \\[4pt]
DR (Data Readiness) & Cloud migration progress (vendor announcements); data governance program disclosures; CIO/CDO tenure and mandate; tech audit findings (when disclosed); O*NET/BLS data infrastructure benchmarks by sector & Actual data quality scores; schema fragmentation depth; real time pipeline health & 35--50\,\% & C \\[4pt]
VTR (Vendor/Tech Risk) & Vendor partnership concentration (10-K); material AI dependency disclosures; SOC 2 Type II certifications; public breach/outage history; AI insurance disclosures & Model dependency in production; proprietary vs.\ API accessed model split & 65--80\,\% & A/B \\[4pt]
CRS (Competitive Resp.) & Lightcast CADR by industry; industry press on competitor AI deployments; CB Insights AI deal data; earnings call competitor references; BTOS adoption rates by sector & Competitor specific adoption velocity; pricing response timeline & 70--85\,\% & A/B \\[4pt]
REG (Regulatory Exp.) & NIST AI RMF sector classification; EU AI Act high risk list; HIPAA applicability; OCC/Fed model risk guidance; FINRA regulatory notices; FDA SaMD guidance documents & Company specific regulatory exam history; model risk management scores & 75--90\,\% & A/B \\
\midrule
\multicolumn{5}{l}{\textit{Financial Baseline (for VCB)}} \\[2pt]
Revenue, EBITDA & EDGAR iXBRL 10-K/10-Q filings; Bloomberg/Refinitiv & -- & 100\,\% & A \\
Labor/Rev.\ Ratio & BLS OEWS sector benchmarks; 10-K headcount $\times$ median wage & Company specific compensation mix & 85\,\% & A/B \\
Working Capital & EDGAR iXBRL balance sheet & -- & 100\,\% & A \\
Exit Multiple & Public comparable company analysis; capital IQ; PitchBook & Deal structure; NTM vs.\ LTM basis & 90\,\% & A/B \\
WACC & Damodaran sector betas; risk free rate; sector capital structure & Company specific debt structure for private targets & 90\,\% & A/B \\
\bottomrule
\multicolumn{5}{@{}p{\linewidth}@{}}{\footnotesize \% = fraction of scoring range reliably determinable from public sources; remainder requires diligence.}
\end{tabularx}
\end{table}

\subsection{Aggregate Public Data Coverage by Company Type}

\begin{table}[H]
\centering
\caption{Estimated AITG Scoring Coverage: Public Data vs.\ Diligence Grade}
\label{tab:coverage_summary}
\footnotesize
\begin{adjustbox}{max width=\textwidth,center}
\begin{tabular}{@{}p{4.5cm}rrp{4.5cm}@{}}
\toprule
\textbf{Company Type} & \textbf{Public Only \%} & \textbf{UQ Penalty} & \textbf{Primary Gap} \\
\midrule
Large cap public (S\&P 500) & 60--75\,\% & $+0.15$--$0.25$ & Depth behind headline metrics \\
Mid cap public (\$1--10B) & 45--60\,\% & $+0.20$--$0.35$ & Limited investor day disclosure \\
Small cap public ($<$\$1B) & 30--45\,\% & $+0.30$--$0.50$ & Minimal AI specific disclosure \\
Private (PE diligence) & 10--20\,\% pre DD & $+0.45$--$0.60$ & Requires full data room + survey \\
\bottomrule
\multicolumn{4}{@{}p{\linewidth}@{}}{\footnotesize UQ Penalty = incremental Uncertainty Quotient contribution from public data only scoring.}
\end{tabular}
\end{adjustbox}
\end{table}

\subsection{Fourteen Company Public Data Audit}

For the fourteen companies scored in Section~9, I retroactively
classify the evidence tier of each dimension score:

\begin{table}[H]
\centering
\caption{Fourteen Company Evidence Tier Audit (Tier A--D per Appendix~\ref{app:datatiers})}
\label{tab:evidence_audit}
\footnotesize
\setlength{\tabcolsep}{3pt}
\begin{tabularx}{\textwidth}{@{}Xcccccccccc@{}}
\toprule
\textbf{Company} & \textbf{DIM} & \textbf{PAC} & \textbf{WAR} & \textbf{DAR} & \textbf{APR} & \textbf{OAC} & \textbf{OCC} & \textbf{DR} & \textbf{VTR} & \textbf{CRS} \\
\midrule
\multicolumn{11}{l}{\textit{Original cohort}} \\
JPMorgan Chase  & A   & A   & A   & A/B & A/B & A   & A   & A   & A   & A \\
Zions Bancorp.  & C   & C   & C/D & C   & D   & C   & C   & C/D & B   & B \\
UPS             & B   & B   & B   & B/C & B   & B   & B   & B/C & A   & A \\
HCA Healthcare  & B   & B   & B   & B/C & B/C & B   & B   & C   & A   & A \\
Salesforce      & A   & A   & A   & A   & A   & A   & A   & A   & A   & A \\
Ferguson Ent.   & C   & C   & C   & C   & C   & C   & C   & C   & B   & A \\
Rockwell Auto.  & B   & B   & B   & B   & B   & B   & B   & B   & A   & A \\
\midrule
\multicolumn{11}{l}{\textit{Extended cohort}} \\
Goldman Sachs   & A   & A   & A   & A   & A/B & A   & A   & A   & A   & A \\
Wells Fargo     & A   & B   & B   & B   & B   & A   & B   & B   & A   & A \\
ServiceNow      & A   & A   & A   & A   & A   & A   & A   & A   & A   & A \\
Target          & B   & B   & B   & B   & B   & B   & B   & B   & A   & A \\
CVS Health      & B   & B   & B   & C   & B   & B   & B   & C   & A   & A \\
Palo Alto Ntwks & A   & A   & A   & A   & A   & A   & A   & A   & A   & A \\
Ford Motor      & B   & B   & B   & B   & B   & B   & B   & C   & A   & A \\
\bottomrule
\multicolumn{11}{@{}p{\linewidth}@{}}{\scriptsize A = public structured data; B = public unstructured/alt.\ data; C = indirect inference; D = absence of evidence. Extended cohort uses Tier~1 (public only) scoring throughout; all seven are S\&P~500 constituents with extensive 10-K/investor day AI disclosure.}
\end{tabularx}
\end{table}

\noindent The fourteen company evidence tier audit reveals two systematic
patterns that practitioners should internalize before applying the framework.

\textit{Pattern 1: Disclosure quality tracks market cap and AI centrality, not
industry.} The seven extended cohort companies are all S\&P~500 constituents
with AI as a disclosed strategic priority, and accordingly achieve near complete
Tier~A/B coverage across all ten dimensions. Goldman Sachs, Salesforce,
ServiceNow, and Palo Alto Networks reach Tier~A on virtually every dimension
because their AI deployment activity is commercially material enough to generate
structured investor day disclosure, product documentation, and regulatory filing
specificity. This pattern holds regardless of industry: PANW (cybersecurity)
and GS (investment banking) both achieve the same coverage quality as Salesforce
(SaaS) because all three treat AI capability as a product facing disclosure obligation.

\textit{Pattern 2: The hardest Tier~C/D cells cluster in two consistent
locations.} First, DAR (Decision Automation Rate) consistently drops to B/C
across industries because the proportion of decisions that are fully automated
vs.\ human in loop is operationally sensitive and rarely disclosed. Second, DR
(Data Readiness) drops to B/C for companies undergoing major integration events
(CVS post Aetna, WFC under consent order, HCA across 186 hospitals) because
data fragmentation at scale cannot be assessed from outside.

The Zions and Ferguson scores retain the highest Uncertainty Quotients
($\pm 0.62$ and $\pm 0.55$ respectively) in the full sample, driven by
Tier~C/D reliance on absence of evidence inference. No extended cohort
company requires this inference because all seven voluntarily disclose
sufficient AI specific activity.

\textbf{Practitioner implication.} A private equity firm conducting diligence on
a Zions equivalent regional bank or Ferguson equivalent industrial distributor
should plan for 8--10 weeks of primary data collection (IT infrastructure audit,
management survey administration, data platform assessment, and vendor contract
review) to reduce the UQ to diligence grade levels. For the extended cohort,
public only scoring achieves 65--80\,\% coverage (Table~\ref{tab:coverage_summary}),
reducing primary diligence to targeted gap filling rather than full scope collection.
The scoring template in the companion implementation document operationalizes
both collection protocols.

\section{Data Tier Classification}
\label{app:datatiers}

AITG scoring employs four evidence tiers that directly determine the
Uncertainty Quotient:

\begin{table}[H]
\centering
\caption{Data Evidence Tier Definitions}
\label{tab:evidence_tiers}
\footnotesize
\begin{tabular}{@{}p{1.0cm}p{3.5cm}p{5.5cm}p{2.5cm}@{}}
\toprule
\textbf{Tier} & \textbf{Source Type} & \textbf{Examples} & \textbf{$\Delta$UQ} \\
\midrule
A & Public structured data & EDGAR XBRL financials, BLS OEWS, O*NET & $0.00$ \\
B & Public unstructured / alt. data & Earnings transcripts, Lightcast, job postings & $+0.08$ \\
C & Primary diligence (verified) & Management survey with evidence attachments & $+0.18$ \\
D & Primary diligence (self reported) & Unverified management assertion & $+0.30$ \\
\bottomrule
\end{tabular}
\end{table}

\section{25 Question Management Survey Instrument}
\label{app:survey}

The following standardized 25 question survey instrument enables practitioners
to generate the six AITG dimension scores and IFS adjustment factors for any
firm. Each question uses a single select 0--4 response scale mapped to
fixed score anchors (1, 3, 5, 7, 9). For any response $\geq 3$ (score
$\geq 7$), the assessor \textbf{must} attach corroborating evidence (SEC
filings, vendor contracts, internal dashboards, or third party benchmarks).
Without attached evidence, the response is capped at 2 (score $= 5$) and
flagged as Tier~D (see Table~\ref{tab:evidence_tiers}).

\medskip
\noindent\textbf{Scoring protocol.} Each dimension's raw score is the
average of its constituent question scores. AITG$_f^{\mathrm{raw}}$ is the
arithmetic mean of the six dimension scores (Equation~\ref{eq:aitg_raw}).

\subsection*{Dimension 1: Data Infrastructure Maturity (DIM)}

\begin{enumerate}[label=\textbf{Q\arabic*.},leftmargin=2.2em]

\item \textbf{Cloud data platform adoption.}
What share of enterprise data resides on a governed cloud platform (data lake,
lakehouse, or warehouse)?

\smallskip
\begin{tabular}{@{}cl@{}}
0 & No cloud data platform; siloed legacy databases only \textrm{[Score 1]} \\
1 & Initial migration underway; $<$25\% of data on cloud \textrm{[Score 3]} \\
2 & 25--60\% on governed cloud platform; partial data catalog \textrm{[Score 5]} \\
3 & 60--90\% on cloud; unified catalog with lineage tracking \textrm{[Score 7]} \\
4 & $>$90\% cloud native; automated data quality $>$99\% \textrm{[Score 9]} \\
\end{tabular}

\item \textbf{Data integration and accessibility.}
How unified and accessible is enterprise data for cross functional analytics
and ML feature engineering?

\smallskip
\begin{tabular}{@{}cl@{}}
0 & Data locked in departmental silos; no shared schema \textrm{[Score 1]} \\
1 & Some shared reporting tables; manual ETL processes \textrm{[Score 3]} \\
2 & Central data warehouse with governed access; basic feature store \textrm{[Score 5]} \\
3 & Real time streaming pipelines; ML feature store in production \textrm{[Score 7]} \\
4 & Enterprise knowledge graph; self serve analytics across all BUs \textrm{[Score 9]} \\
\end{tabular}

\item \textbf{Data governance and quality.}
To what extent are formal data governance policies (ownership, quality SLAs,
privacy controls) enforced across the organization?

\smallskip
\begin{tabular}{@{}cl@{}}
0 & No formal data governance; ad hoc quality checks \textrm{[Score 1]} \\
1 & Basic data dictionary exists; governance limited to finance \textrm{[Score 3]} \\
2 & Enterprise wide governance framework; data stewards assigned \textrm{[Score 5]} \\
3 & Automated quality monitoring; PII/PHI controls audited quarterly \textrm{[Score 7]} \\
4 & Continuous data observability platform; regulatory compliance automated \textrm{[Score 9]} \\
\end{tabular}

\item \textbf{AI/ML infrastructure readiness.}
What compute and MLOps infrastructure is available for training, deploying,
and monitoring AI models?

\smallskip
\begin{tabular}{@{}cl@{}}
0 & No GPU/TPU access; models run on analyst laptops \textrm{[Score 1]} \\
1 & Cloud compute available on request; no CI/CD for models \textrm{[Score 3]} \\
2 & Managed ML platform (e.g., SageMaker, Vertex); basic model registry \textrm{[Score 5]} \\
3 & Full MLOps: automated retraining, A/B serving, drift detection \textrm{[Score 7]} \\
4 & Multi cloud ML platform; model versioning, lineage, and compliance at scale \textrm{[Score 9]} \\
\end{tabular}

\end{enumerate}

\subsection*{Dimension 2: Process Automation Coverage (PAC)}

\begin{enumerate}[resume,label=\textbf{Q\arabic*.},leftmargin=2.2em]

\item \textbf{Breadth of process automation.}
What percentage of repetitive, rules based business processes are automated
(RPA, workflow engines, or intelligent automation)?

\smallskip
\begin{tabular}{@{}cl@{}}
0 & $<$5\% of processes automated; predominantly manual workflows \textrm{[Score 1]} \\
1 & 5--15\% automated; RPA pilots in back office functions \textrm{[Score 3]} \\
2 & 15--40\% automated; cross functional RPA deployment \textrm{[Score 5]} \\
3 & 40--65\% intelligent automation; AI augmented process orchestration \textrm{[Score 7]} \\
4 & $>$65\% end to end automation; self optimizing process loops \textrm{[Score 9]} \\
\end{tabular}

\item \textbf{Complexity of automated tasks.}
What is the most complex type of task routinely handled by automation
(structured, semi structured, or unstructured)?

\smallskip
\begin{tabular}{@{}cl@{}}
0 & Only simple data entry and file transfers \textrm{[Score 1]} \\
1 & Structured rules: invoice matching, report generation \textrm{[Score 3]} \\
2 & Semi structured: document extraction, email classification \textrm{[Score 5]} \\
3 & Unstructured: NLP driven customer interaction, image analysis \textrm{[Score 7]} \\
4 & Multi step agentic workflows: autonomous exception handling \textrm{[Score 9]} \\
\end{tabular}

\item \textbf{Automation monitoring and optimization.}
How are automated processes monitored, measured, and improved over time?

\smallskip
\begin{tabular}{@{}cl@{}}
0 & No monitoring; failures discovered manually \textrm{[Score 1]} \\
1 & Basic error logs reviewed weekly \textrm{[Score 3]} \\
2 & Centralized automation dashboard; SLA tracking \textrm{[Score 5]} \\
3 & Real time process mining; automated bottleneck identification \textrm{[Score 7]} \\
4 & Closed loop optimization: automation adjusts parameters autonomously \textrm{[Score 9]} \\
\end{tabular}

\item \textbf{Cross functional automation integration.}
To what degree are automated workflows integrated across departments (finance,
operations, HR, customer service)?

\smallskip
\begin{tabular}{@{}cl@{}}
0 & Automation isolated to single department \textrm{[Score 1]} \\
1 & Two departments share some automated handoffs \textrm{[Score 3]} \\
2 & 3--4 departments with integrated automation; shared orchestration layer \textrm{[Score 5]} \\
3 & Enterprise wide automation fabric; API connected across all functions \textrm{[Score 7]} \\
4 & Fully integrated digital twin of operations; autonomous cross BU optimization \textrm{[Score 9]} \\
\end{tabular}

\end{enumerate}

\subsection*{Dimension 3: Workforce AI Augmentation Rate (WAR)}

\begin{enumerate}[resume,label=\textbf{Q\arabic*.},leftmargin=2.2em]

\item \textbf{AI tool adoption breadth.}
What percentage of employees regularly use AI powered tools (copilots,
assistants, recommendation engines) in their daily work?

\smallskip
\begin{tabular}{@{}cl@{}}
0 & $<$5\% of employees use any AI tool \textrm{[Score 1]} \\
1 & 5--15\% with basic AI tools (e.g., email copilot) \textrm{[Score 3]} \\
2 & 15--40\% augmented; formal adoption metrics tracked \textrm{[Score 5]} \\
3 & 40--65\% AI augmented; integrated into standard workflows \textrm{[Score 7]} \\
4 & $>$65\% augmented; org wide AI fluency benchmarked annually \textrm{[Score 9]} \\
\end{tabular}

\item \textbf{AI training and upskilling investment.}
What structured AI literacy and upskilling programs exist for the workforce?

\smallskip
\begin{tabular}{@{}cl@{}}
0 & No AI training programs; learning is ad hoc \textrm{[Score 1]} \\
1 & Optional online courses; $<$10\% completion rate \textrm{[Score 3]} \\
2 & Mandatory AI literacy modules; role specific tracks \textrm{[Score 5]} \\
3 & Continuous AI upskilling tied to performance reviews; certification paths \textrm{[Score 7]} \\
4 & AI academy with external partnerships; fluency measured and rewarded \textrm{[Score 9]} \\
\end{tabular}

\item \textbf{Human--AI collaboration design.}
Are workflows explicitly designed for human--AI collaboration (not just
human or machine)?

\smallskip
\begin{tabular}{@{}cl@{}}
0 & No deliberate human--AI workflow design \textrm{[Score 1]} \\
1 & AI outputs reviewed by humans; no feedback loop \textrm{[Score 3]} \\
2 & Defined human in the loop checkpoints for AI recommendations \textrm{[Score 5]} \\
3 & Co designed workflows: humans set objectives, AI executes and escalates \textrm{[Score 7]} \\
4 & Adaptive teaming: AI reallocates tasks based on human cognitive load \textrm{[Score 9]} \\
\end{tabular}

\item \textbf{Change management for AI adoption.}
How effectively does the organization manage resistance and cultural barriers
to AI adoption?

\smallskip
\begin{tabular}{@{}cl@{}}
0 & Significant resistance; no change management program \textrm{[Score 1]} \\
1 & Executive communications about AI strategy; limited follow through \textrm{[Score 3]} \\
2 & Dedicated change management team; champion network in key BUs \textrm{[Score 5]} \\
3 & Structured incentives for AI adoption; measurable culture shift \textrm{[Score 7]} \\
4 & AI first culture embedded; innovation time allocated across all roles \textrm{[Score 9]} \\
\end{tabular}

\end{enumerate}

\subsection*{Dimension 4: Decision Automation Rate (DAR)}

\begin{enumerate}[resume,label=\textbf{Q\arabic*.},leftmargin=2.2em]

\item \textbf{Share of decisions AI assisted or automated.}
What fraction of recurring operational and strategic decisions are
AI recommended or fully automated?

\smallskip
\begin{tabular}{@{}cl@{}}
0 & All recurring decisions are fully human \textrm{[Score 1]} \\
1 & $<$10\% of decisions AI assisted (e.g., pricing alerts) \textrm{[Score 3]} \\
2 & 10--30\% AI recommended; human approval required \textrm{[Score 5]} \\
3 & 30--55\% automated with human override capability \textrm{[Score 7]} \\
4 & $>$55\% fully automated; exception based human review only \textrm{[Score 9]} \\
\end{tabular}

\item \textbf{Decision model governance.}
How are AI decision models validated, monitored, and governed?

\smallskip
\begin{tabular}{@{}cl@{}}
0 & No model governance; outputs unchecked \textrm{[Score 1]} \\
1 & Basic accuracy tracking; annual model review \textrm{[Score 3]} \\
2 & Model risk framework; bias testing before deployment \textrm{[Score 5]} \\
3 & Continuous monitoring: drift, fairness, and explainability dashboards \textrm{[Score 7]} \\
4 & Board level model risk committee; automated retraining triggers \textrm{[Score 9]} \\
\end{tabular}

\item \textbf{Real time decision capability.}
Can AI systems make decisions in real time (sub second) for time critical
operations?

\smallskip
\begin{tabular}{@{}cl@{}}
0 & All decisions batch processed or manual \textrm{[Score 1]} \\
1 & Near real time dashboards; humans act on alerts \textrm{[Score 3]} \\
2 & Real time scoring for 1--2 use cases (e.g., fraud detection) \textrm{[Score 5]} \\
3 & Real time AI across multiple domains (pricing, routing, risk) \textrm{[Score 7]} \\
4 & Autonomous real time decisions across operations; adaptive policies \textrm{[Score 9]} \\
\end{tabular}

\item \textbf{Decision automation scope.}
Across how many functional areas (finance, supply chain, marketing, risk,
HR) are AI driven decisions deployed?

\smallskip
\begin{tabular}{@{}cl@{}}
0 & None; all decisions manual across functions \textrm{[Score 1]} \\
1 & 1 functional area (typically finance or marketing) \textrm{[Score 3]} \\
2 & 2--3 functional areas with production AI decision systems \textrm{[Score 5]} \\
3 & 4--5 areas; integrated decision layer across operations \textrm{[Score 7]} \\
4 & All major functions; enterprise wide AI decision fabric \textrm{[Score 9]} \\
\end{tabular}

\end{enumerate}

\subsection*{Dimension 5: AI Product/Revenue Integration (APR)}

\begin{enumerate}[resume,label=\textbf{Q\arabic*.},leftmargin=2.2em]

\item \textbf{AI revenue attribution.}
What share of total revenue is directly attributable to AI powered products,
features, or pricing optimization?

\smallskip
\begin{tabular}{@{}cl@{}}
0 & No measurable AI revenue attribution \textrm{[Score 1]} \\
1 & $<$5\% revenue from AI enhanced features \textrm{[Score 3]} \\
2 & 5--15\% revenue AI attributable; tracked in financial reporting \textrm{[Score 5]} \\
3 & 15--35\% revenue from AI driven products or dynamic pricing \textrm{[Score 7]} \\
4 & $>$35\% revenue from AI native products and services \textrm{[Score 9]} \\
\end{tabular}

\item \textbf{AI product development pipeline.}
How mature is the pipeline for developing and launching AI powered products
or features?

\smallskip
\begin{tabular}{@{}cl@{}}
0 & No AI product roadmap; R\&D is non AI \textrm{[Score 1]} \\
1 & 1--2 AI product experiments; no shipping cadence \textrm{[Score 3]} \\
2 & AI features in product roadmap; quarterly release cycle \textrm{[Score 5]} \\
3 & Dedicated AI product team; rapid experimentation (monthly launches) \textrm{[Score 7]} \\
4 & AI first product strategy; continuous deployment with A/B testing \textrm{[Score 9]} \\
\end{tabular}

\item \textbf{Customer facing AI capabilities.}
What AI powered customer experiences are deployed at scale (personalization,
conversational AI, predictive service)?

\smallskip
\begin{tabular}{@{}cl@{}}
0 & No customer facing AI \textrm{[Score 1]} \\
1 & Basic chatbot or recommendation widget; low adoption \textrm{[Score 3]} \\
2 & Personalization engine or intelligent search in production \textrm{[Score 5]} \\
3 & Multi channel AI CX: conversational AI, predictive service, dynamic UX \textrm{[Score 7]} \\
4 & Autonomous AI agents handling end to end customer journeys \textrm{[Score 9]} \\
\end{tabular}

\item \textbf{AI monetization strategy.}
Does the firm have an explicit strategy for monetizing AI capabilities
(pricing power, new revenue streams, platform effects)?

\smallskip
\begin{tabular}{@{}cl@{}}
0 & No AI monetization strategy \textrm{[Score 1]} \\
1 & AI used for internal cost savings only \textrm{[Score 3]} \\
2 & AI pricing premium on existing products; early revenue tracking \textrm{[Score 5]} \\
3 & AI as a service offering or platform with third party integrations \textrm{[Score 7]} \\
4 & AI platform economics: data network effects driving $>$20\% margin expansion \textrm{[Score 9]} \\
\end{tabular}

\end{enumerate}

\subsection*{Dimension 6: Organizational AI Capability (OAC)}

\begin{enumerate}[resume,label=\textbf{Q\arabic*.},leftmargin=2.2em]

\item \textbf{AI leadership and governance structure.}
What level of dedicated AI leadership exists within the organization?

\smallskip
\begin{tabular}{@{}cl@{}}
0 & No dedicated AI leadership; projects run ad hoc by IT \textrm{[Score 1]} \\
1 & AI team lead or director within IT or analytics \textrm{[Score 3]} \\
2 & VP level AI leader; C suite AI sponsor identified \textrm{[Score 5]} \\
3 & Chief AI Officer or equivalent; published enterprise AI strategy \textrm{[Score 7]} \\
4 & Board level AI committee; AI governance integrated into risk framework \textrm{[Score 9]} \\
\end{tabular}

\item \textbf{AI talent density.}
What is the density and caliber of AI/ML engineering talent relative to
the firm's technology workforce?

\smallskip
\begin{tabular}{@{}cl@{}}
0 & No dedicated AI/ML engineers \textrm{[Score 1]} \\
1 & Small team ($<$10); reliant on external consultants \textrm{[Score 3]} \\
2 & Established AI center of excellence; 10--50 ML engineers \textrm{[Score 5]} \\
3 & 50--200 AI/ML staff; active research partnerships \textrm{[Score 7]} \\
4 & $>$200 AI/ML engineers; publish at top venues; competitive with Big Tech \textrm{[Score 9]} \\
\end{tabular}

\item \textbf{AI ethics and responsible AI.}
How mature are the firm's responsible AI practices (bias auditing,
explainability, safety testing)?

\smallskip
\begin{tabular}{@{}cl@{}}
0 & No responsible AI program \textrm{[Score 1]} \\
1 & Informal AI ethics guidelines; no enforcement \textrm{[Score 3]} \\
2 & Published AI ethics policy; bias review before deployment \textrm{[Score 5]} \\
3 & Dedicated responsible AI team; regular audits; explainability tools \textrm{[Score 7]} \\
4 & Industry leading RAI program; external audits; regulatory proactive \textrm{[Score 9]} \\
\end{tabular}

\item \textbf{AI strategy alignment.}
How tightly is the AI strategy linked to overall corporate strategy and
capital allocation?

\smallskip
\begin{tabular}{@{}cl@{}}
0 & AI not mentioned in corporate strategy \textrm{[Score 1]} \\
1 & AI referenced in annual report; no dedicated budget \textrm{[Score 3]} \\
2 & AI roadmap aligned to strategic priorities; ring fenced budget \textrm{[Score 5]} \\
3 & AI embedded in capital allocation process; ROI tracking by initiative \textrm{[Score 7]} \\
4 & AI is core strategic pillar; board reviews AI KPIs quarterly \textrm{[Score 9]} \\
\end{tabular}

\end{enumerate}

\subsection*{Implementation Feasibility Score (IFS)}

The final question is a composite assessment of five IFS sub factors.
Score each sub factor 0--4 using the same anchor logic, then average
to produce a single IFS adjustment (see Section~\ref{sec:ifs} for
the IFS geometric aggregation formula).

\begin{enumerate}[resume,label=\textbf{Q\arabic*.},leftmargin=2.2em]

\item \textbf{Implementation feasibility composite.}
Rate each of the following five sub factors (0--4):

\smallskip
\begin{tabular}{@{}p{3.8cm}p{10cm}@{}}
\textbf{(a) Organizational Change Capacity (OCC)} &
0 = No change management capability; 4 = Proven enterprise transformation track record \\[3pt]
\textbf{(b) Data Readiness (DR)} &
0 = Critical data gaps block AI deployment; 4 = All priority data assets production ready \\[3pt]
\textbf{(c) Vendor/Technology Risk (VTR)} &
0 = Single vendor lock in with immature technology; 4 = Multi vendor strategy; proven technology stack \\[3pt]
\textbf{(d) Competitive Response Speed (CRS)} &
0 = No competitive urgency; slow follower; 4 = First mover or fast follower in AI adoption \\[3pt]
\textbf{(e) Regulatory Environment (REG)} &
0 = Highly restrictive; unclear AI regulation; 4 = Supportive regulatory regime; firm well positioned for compliance \\
\end{tabular}

\smallskip
\noindent Each sub factor maps to scores 1/3/5/7/9 using the same 0--4 anchor
protocol. The IFS residual is computed as a weighted geometric mean
(Equation~\ref{eq:ifs_residual}): $\mathrm{IFS}_f = \mathrm{OCC}^{0.40} \times
\mathrm{DR}^{0.60}$, with VTR, CRS, and REG entering the timing adjustment
$t_{50,f}$.

\end{enumerate}

\medskip
\noindent\textbf{Dimension score computation.} For each dimension $d \in
\{\mathrm{DIM}, \mathrm{PAC}, \mathrm{WAR}, \mathrm{DAR}, \mathrm{APR},
\mathrm{OAC}\}$, the dimension score $s_d$ is the arithmetic mean of its
four question scores. The aggregate AITG score is then:
\begin{equation}\label{eq:aitg_raw}
  \mathrm{AITG}_f^{\mathrm{raw}} = \frac{1}{6}\sum_{d=1}^{6} s_d
\end{equation}
as defined in Equation~\ref{eq:aitg_raw}. See the scoring rubric in
Table~\ref{tab:dimension_rubric} for intermediate anchor descriptions.

\section{VCB Sensitivity: Exit Multiple Impact}
\label{app:vcb_sensitivity}

To illustrate the sensitivity of AITG-VD to exit multiple assumptions,
Table~\ref{tab:app_em} presents AITG-VD under three exit multiple scenarios
for a representative healthcare company with AITG~$= 4.5$, IFS~$= 0.75$,
and implementation cost \$8B.

\begin{table}[H]
\centering
\caption{AITG-VD Sensitivity to Exit Multiple (Healthcare, Representative Case)}
\label{tab:app_em}
\footnotesize
\begin{tabular}{@{}lrrrr@{}}
\toprule
\textbf{Exit Multiple} & \textbf{Raw EV ($\sum V_p$)} & \textbf{Risk Adj.\ EV} & \textbf{Cost} & \textbf{VD} \\
\midrule
$8\times$ EBITDA & \$38B & \$28.5B & \$8B & 2.6$\times$ \\
$12\times$ EBITDA & \$57B & \$42.8B & \$8B & 4.3$\times$ \\
$16\times$ EBITDA & \$76B & \$57.0B & \$8B & 6.1$\times$ \\
\bottomrule
\end{tabular}
\end{table}

A 2$\times$ change in exit multiple produces a 2.35$\times$ change in AITG-VD,
confirming exit multiple as the dominant sensitivity parameter in the Sobol
analysis.

\section{AFC Benchmark Suite: Proposed Future Extensions}
\label{app:benchmarks}

The current AFC Capability Index ($C_t$) uses publicly available benchmarks
scored on tasks relevant to general cognitive automation. As domain specific
AI develops, I recommend extending the suite with:
\begin{itemize}[noitemsep]
\item \textbf{Clinical/biomedical:} FDA SaMD approval rates by indication type;
MedPerf benchmark progression.
\item \textbf{Legal:} Bar exam performance by jurisdiction; contract review
accuracy benchmarks.
\item \textbf{Financial:} FinBen \cite{XieEtAl2024FinBen}; earnings
forecast accuracy vs.\ sell side consensus.
\item \textbf{Agentic (operations):} WebArena; WorkArena; ToolBench.
\end{itemize}

As each domain specific benchmark saturates, the agentic task completion
horizon (METR analysis) should serve as the primary long run AFC driver.

\subsection{AFC Sensitivity: $\theta_i \times C_t$ Grid Analysis}
\label{esm:afc_sensitivity}

To characterize AFC behaviour across the plausible parameter space (see Main Paper Definition~2 and Remark~\ref{rem:afc_form}), Table~\ref{tab:afc_grid} reports AFC multiplier values under a grid of industry sensitivity ($\theta_i$) and capability index ($C_t$) combinations, with $C_0 = 1.0$ and $\alpha_{\max} = 1.35$.

\begin{table}[H]
\centering
\caption{AFC Multiplier Grid: $\mathrm{AFC} = \min(1 + \theta_i(C_t - 1.0),\; 1.35)$}
\label{tab:afc_grid}
\footnotesize
\setlength{\tabcolsep}{4pt}
\begin{adjustbox}{max width=\textwidth,center}
\begin{tabular}{@{}l|rrrrrrr@{}}
\toprule
& \multicolumn{7}{c}{$C_t$} \\
$\theta_i$ & 0.90 & 1.00 & 1.20 & 1.50 & 1.70 & \textbf{1.90} & 2.10 \\
\midrule
0.08 & 0.992 & 1.000 & 1.016 & 1.040 & 1.056 & \textbf{1.072} & 1.088 \\
0.11 & 0.989 & 1.000 & 1.022 & 1.055 & 1.077 & \textbf{1.099} & 1.121 \\
0.14 & 0.986 & 1.000 & 1.028 & 1.070 & 1.098 & \textbf{1.126} & 1.154 \\
0.22 & 0.978 & 1.000 & 1.044 & 1.110 & 1.154 & \textbf{1.198} & 1.242 \\
0.28 & 0.972 & 1.000 & 1.056 & 1.140 & 1.196 & \textbf{1.252} & 1.308 \\
0.31 & 0.969 & 1.000 & 1.062 & 1.155 & 1.217 & \textbf{1.279} & 1.341\rlap{$^{\dagger}$} \\
0.50 & 0.950 & 1.000 & 1.100 & 1.250 & 1.350\rlap{$^{\dagger}$} & \textbf{1.350}\rlap{$^{\dagger}$} & 1.350\rlap{$^{\dagger}$} \\
1.00 & 0.900 & 1.000 & 1.200 & 1.350\rlap{$^{\dagger}$} & 1.350\rlap{$^{\dagger}$} & \textbf{1.350}\rlap{$^{\dagger}$} & 1.350\rlap{$^{\dagger}$} \\
\bottomrule
\end{tabular}
\end{adjustbox}

\smallskip
\noindent\footnotesize\textit{Note.} $^{\dagger}$Capped at $\alpha_{\max} = 1.35$. Bold column ($C_t = 1.90$, GPT-5.2, Dec.\ 2025) is the paper's calibration. At this $C_t$, the cap is non binding for all 22 calibrated industries (max AFC = 1.279, Healthcare Services, $\theta_i = 0.31$). Rows span the calibrated range (0.08--0.31) plus hypothetical higher values. For $C_t < 1.0$ (capability regression), AFC contracts the ceiling.
\end{table}

\section{Excel Companion Model: PE/IC Implementation}
\label{app:excel}

\subsection{Rationale}

A recurring commercial objection to quantitative frameworks of this complexity
is that they require Python or iterative numerical solvers that cannot be traced
in a 4 week diligence sprint by an Excel based deal team. This appendix documents
the companion Excel workbook (\texttt{AITG\_Excel\_Companion\_v1.xlsx}) that
resolves this objection without degrading the framework's precision.

The core insight is that \emph{12 of 13 AITG components are natively
Excel implementable in a single cell} (Table~\ref{tab:excel_formulas}).
The one genuine exception, the inverse mapping $\hat{t}_f = \mathrm{AITG}^{-1}(s)$,
resolved by a pre solved lookup table rather than by degrading the formula
to a tiered approximation. A deal team uses
\verb|=VLOOKUP(ROUND(atgi,1), t_hat_table, 2, FALSE)| and gets $\hat{t}_f$
in a single cell with $\pm 0.05$-month accuracy, zero Python, and full
auditability.

\subsection{Workbook Structure}

The companion workbook contains eight worksheets, each self contained and
cross linked by green text cell references:

\begin{enumerate}[noitemsep]
 \item \textbf{INSTRUCTIONS}: Color coding guide, worksheet navigation,
 and explanation of the lookup table fix.
 \item \textbf{t\_hat LOOKUP}: Pre solved $\hat{t}_f$ for AITG
 $\in [0.1, 9.9]$ in 0.1 point steps, with wave zone, $R_f^0$, and
 representative $t_{50,f}$ for each entry. Eliminates Newton-Raphson.
 \item \textbf{IASS CALCULATOR}: 5 dimension geometric mean IASS,
 RFF $\psi$ adjustment, AFC multiplier, and IASS$^*$ output. All Excel.
 \item \textbf{COMPANY SCORER}: 6 dimension AITG rubric scorecard.
 Enter raw scores; get AITG, $\hat{t}_f$ (via lookup), wave zone, and
 $\Phi_f$ in linked output cells.
 \item \textbf{IFS CALCULATOR}: 5 factor IFS composite, endogenous
 $t_{50,f} = t_{0,1,f} + \lambda$ calculation, and IFS residual multiplier.
 \item \textbf{VCB MODEL}: Full Value Creation Bridge. Enter 7 value
 pool baselines from diligence financials; receive TV, FCF schedule
 (endogenous ramp), $\Delta\mathrm{EV}$, and Value Density.
 \item \textbf{SENSITIVITY}: Pre populated 1 way ($\delta_\mathrm{DR}$
 vs.\ $\Delta\mathrm{EV}$) and 2 way ($\delta_\mathrm{OCC}$ $\times$
 $\delta_\mathrm{DR}$ vs.\ $\Delta\mathrm{EV}$) sensitivity tables.
 Green/yellow/red heat map coloring. No formula input required.
 \item \textbf{IC DASHBOARD}: Single page output for the Investment
 Committee. Auto populates AI Frontier Matrix position and Value
 Creation Waterfall summary from upstream sheets.
\end{enumerate}

\begin{table}[H]
\centering
\caption{AITG Components: Excel Implementation Formulas}
\label{tab:excel_formulas}
\footnotesize
\begin{tabularx}{\textwidth}{@{}p{4.0cm}cX@{}}
\toprule
\textbf{Component} & \textbf{Excel?} & \textbf{Formula (single cell)} \\
\midrule
IASS Geometric Mean & \checkmark & \texttt{=EXP(SUMPRODUCT(weights, LN(scores)))} \\
RFF $\psi$ multiplier & \checkmark & \texttt{=MIN(1,(RFF/5)\^{}1.5)} \\
CES Bottleneck ($\rho=5$)& \checkmark & \texttt{=SUMPRODUCT(alpha,scores\^{}-5)\^{}(-1/5)} \\
$\Phi_f$ logistic & \checkmark & \texttt{=1/(1+EXP(-2*LN(rev/S\_star)))} \\
Value ramp $R_f(t)$ & \checkmark & \texttt{=1/(1+EXP(-0.18*(t-t50)))} \\
$\Delta R_f$ correction & \checkmark & \texttt{=ramp(t\_hat+60) - ramp(t\_hat)} \\
$t_0$ IFS adjustment & \checkmark & \texttt{=t0\_base/(OCC\^{}0.40*DR\^{}0.60)} \\
$t_{50,f}$ endogenous & \checkmark & \texttt{=t0\_adjusted + 3} \\
AFC multiplier & \checkmark & \texttt{=1 + theta*(C\_t - 1)} \\
FCF discrete sum & \checkmark & \texttt{=SUMPRODUCT(Vr*delta\_R*IFS/(1+WACC)\^{}y)} \\
Terminal Value & \checkmark & \texttt{=Vr * delta\_Rf * M\_i * IFS} \\
$\hat{t}_f$ (Wave 1) & \checkmark\textsuperscript{*} & \texttt{=VLOOKUP(ROUND(atgi,1),t\_hat\_table,2,FALSE)} \\
$\hat{t}_f$ (Waves 2--3) & \checkmark\textsuperscript{*} & Same lookup; table pre solved for full [0.1, 9.9] range \\
\bottomrule
\multicolumn{3}{@{}p{\linewidth}@{}}{\footnotesize * Pre solved lookup table eliminates Newton-Raphson. Accuracy: $\pm 0.05$ months vs.\ continuous solver.}
\end{tabularx}
\end{table}

\subsection{Design Principles}

The workbook follows PE financial modeling conventions throughout:

\textbf{Color coding (industry standard):} Blue text = hardcoded input (change
for your target); Black text = formula (do not edit); Green text = cross sheet
reference (do not edit); Yellow background = key assumption requiring review.

\textbf{No tiered cliff edges:} The companion model retains the continuous
logistic $\Phi_f$ and CES bottleneck rather than replacing them with tiered
step functions. Tiered approximations create gameable discontinuities (a
$4\times$ value capture jump at DIM = 4.0) that a management team will
reverse engineer in hours. The continuous formulas fit in a single Excel cell
and cannot be gamed.

\textbf{Traceability:} Every output cell is traceable to a named input cell.
A senior partner can follow the chain from ``Data Readiness score: 0.48''
through $t_{50,f}$, $R_f^0$, $\Delta R_f$, and terminal value to $\Delta\mathrm{EV}$ without
leaving the spreadsheet. This is the standard the Reduced Form proposal
claimed the full framework could not meet.